%% file: main_writeup.tex
\documentclass[acmsmall]{acmart}

\setcopyright{rightsretained}
\acmPrice{}
\acmDOI{}
\acmYear{2023}
\copyrightyear{2023}
\acmJournal{PACMPL}
\acmVolume{7}
\acmNumber{OOPSLA2}
\acmArticle{235}
\acmMonth{10}
\received{2023-04-14}
\received[accepted]{2023-08-27}

\usepackage{amsmath,amsthm}
\usepackage{proof}
\usepackage{upgreek}
\usepackage{chngcntr}
\usepackage{hoasmacros}
\usepackage{booktabs}
\usepackage{multicol}

\usepackage{newunicodechar}
\newunicodechar{⅋}{\ensuremath{\parr}}
\newunicodechar{⊗}{\ensuremath{\otimes}}
\newunicodechar{⊕}{\ensuremath{\oplus}}
\newunicodechar{⊥}{\ensuremath{\bot}}
\newunicodechar{λ}{\ensuremath{\lambda}}
\newunicodechar{→}{\ensuremath{\to}}
\newunicodechar{∥}{\ensuremath{\|}}
\newunicodechar{ν}{\ensuremath{\nu}}

\newcommand{\ms}{\mathsf}

\newcommand{\mt}{\mathtt}
\newcommand{\mc}{\mathcal}

\newcommand{\fns}{\mathtt}

\usepackage{mathtools}

\usepackage{xparse}

\usepackage{rulemgr}
\NewDocumentCommand{\crn}{m}{\textup{(\textsc{#1})}}
\NewDocumentCommand{\srn}{m}{\textup{\texttt{[#1]}}}
\NewDocumentCommand{\rn}{m}{\textup{#1}}

\usepackage{cleveref}

\usepackage{xspace}
\newcommand{\ie}{i.e.\@\xspace}
\newcommand{\eg}{e.g.\@\xspace}

\makeatletter
\newcommand*{\etc}{%
  \@ifnextchar{.}%
  {etc}%
  {etc.\@\xspace}%
}
\makeatother

\theoremstyle{remark}
\newtheorem{case}{Case}

\counterwithin*{subcase}{case}

\author{Chuta Sano}
\affiliation{
  \department{School of Computer Science}              
  \institution{McGill University}            
  \streetaddress{3480 rue University}
  \city{Montr{\'e}al}
  \state{QC}
  \postcode{H3A 0E9}
  \country{Canada}                    
}
\email{chuta.sano@mail.mcgill.ca}
\author{Ryan Kavanagh}
\orcid{0000-0001-9497-4276}
\affiliation{
  \department{School of Computer Science}              
  \institution{McGill University}            
  \streetaddress{3480 rue University}
  \city{Montr{\'e}al}
  \state{QC}
  \postcode{H3A 0E9}
  \country{Canada}                    
}
\email{rkavanagh@cs.mcgill.ca}
\author{Brigitte Pientka}
\orcid{0000-0002-2549-4276}
\affiliation{
  \department{School of Computer Science}              
  \institution{McGill University}            
  \streetaddress{3480 rue University}
  \city{Montr{\'e}al}
  \state{QC}
  \postcode{H3A 0E9}
  \country{Canada}                    
}
\email{bpientka@cs.mcgill.ca}

\definecolor{burntumber}{rgb}{0.54, 0.2, 0.14}

\usepackage{listings}

\newcommand{\nicettfamilysize}{\footnotesize}

\newcommand{\nicettfamily}{%
  \nicettfamilysize\ttfamily%
}%

\lstdefinelanguage{Beluga}{%
  morecomment = [l][\color{gray}]{\%},
  morekeywords = [1]{
    LF, type, ctype, schema, rec, proof, as, case, by, of, unbox, inductive, %
    stratified, some, block, total, mlam, fn, let, in, %
    intros, solve, msplit, suffices, split, by, as, invert, case, unboxed, %
    impossible, undo, toshow, %
    strengthen, %
  },%
  keywordstyle = [1]{\bfseries},%
  morekeywords = [2]{
    unit, arr, c, lam, app, v_lam, v_c, s_app_1, s_app, s_beta, next, refl, halts/m, %
    Unit, Arr, t_lam, t_app, t_c,
    fwd, close, wait, out, inp, pcomp, inl, inr, choice,
    l_fwd1, l_fwd2, l_close, l_wait, l_out, l_inp, l_inl, l_inr, l_choice,
    l_wait2, l_out2, l_out3, l_inp2, l_pcomp1, l_pcomp2, l_inl2, l_inr2,
    l_choice2,
    wtp_fwd, wtp_close, wtp_wait, wtp_out, wtp_inp, wtp_pcomp, wtp_inl, wtp_inr,
    wtp_choice,
    cut1, cut2,
  },%
  keywordstyle = [2]{\bfseries\color{purple!90}},%
  morekeywords = [3]{
    tp, tm, val, step, steps, halts, Reduce, RedSub, ctx, nctx, oftype, eq, %
    halts_step, bwd_closed,
    name, hyp, dual, proc, pproc, linear, wtp, equiv
  },%
  keywordstyle = [3]{\bfseries\color{magenta}},%
  alsoletter=/,%
  columns=flexible,%
  sensitive = true,%
  basicstyle=\nicettfamily,%
  mathescape=true,%
  texcl=true ,%
  escapechar={*},
  literate={%
    {→}{$\rightarrow$ }1%
    {⇒}{$\Rightarrow$ }1%
    {->}{$\rightarrow$ }1%
    {par}{$\parr$ }1%
    {times}{$\otimes$}1%
    {bot}{$\bot$}1%
    {⊥}{\(\bot\)}1%
    {||}{\(\|\)}1%
    {∥}{\(\|\)}1%
    {+}{$\oplus$ }1%
    {⊗}{\(\otimes\) }1%
    {⅋}{\(\parr\) }1%
    {⊕}{\(\oplus\) }1%
    {&}{\(\&\) }1%
    {\\}{$\lambda$}1%
    {--nu}{$\nu$}1%
    {beta}{{\color{purple!90}$\beta$}}1
    {⊢}{$\vdash\;$}1%
    {|-}{$\vdash\;$}1%
    {sigma}{$\sigma$}1%
    {Gamma}{$\Gamma$}1%
    {Delta}{$\Delta$}1%
  }%
}%

\lstMakeShortInline[columns=fixed]@

\usepackage[colorinlistoftodos,prependcaption,textsize=tiny]{todonotes}
\NewDocumentCommand{\rknote}{s m}{
  \IfBooleanTF {#1}
    { \todo[backgroundcolor=green!30]{rak: #2} }
    { \todo[inline, backgroundcolor=green!30]{rak: #2} }
}

\newenvironment{gatherrules}{%
  \addtolength{\jot}{0.6ex}%
  \csname gather*\endcsname%
}{%
  \csname endgather*\endcsname%
  \addtolength{\jot}{-0.6ex}%
}

\include{rules}

\begin{document}
\title{Mechanizing Session-Types using a Structural View: Enforcing Linearity without Linearity}

\
\begin{abstract}
  Session types employ a linear type system that ensures that communication channels cannot be
  implicitly copied or discarded. As a result, many mechanizations of these
  systems require modeling channel contexts and carefully ensuring that they treat channels
  linearly. We demonstrate a technique that localizes
  linearity conditions as additional predicates embedded within type judgments, which allows us
  to use structural typing contexts instead of linear ones. This technique is especially relevant
  when leveraging (weak) higher-order abstract syntax to handle channel
  mobility and the intricate binding structures that arise in~session-typed~systems.

 Following this approach, we mechanize a session-typed system based on classical
  linear logic and its type preservation proof in the proof
  assistant Beluga, which uses the logical framework LF as its
  encoding language. We also prove adequacy for our encoding. This
  shows the tractability and effectiveness of our approach in modelling substructural
  systems such as session-typed languages.
\end{abstract}



\maketitle

\section{Introduction}

The $\pi$-calculus~\cite{Milner80} is a well-studied formalism for message passing concurrency.
Although there have been many efforts to mechanize variants of the \(\pi\)-calculus by encoding their syntax and semantics in proof assistants, mechanization remains an art.
For example, process calculi often feature rich binding structures and semantics such as channel mobility, and these must be carefully encoded to respect \(\alpha\)-equivalence and to avoid channel name clashes.

Even harder to mechanize are session-typed process calculi, in part because they treat communications channels linearly.
\emph{Session types}~\cite{Honda93concur, Honda98esop} specify interactions on named communication channels, and linearity ensures that communication channels are not duplicated or discarded.
As a result, session types can be used to statically ensure safety properties such as session fidelity or deadlock freedom.
However, mechanizing linear type systems adds another layer of complexity;
most encodings of linear type systems encode contexts explicitly: they develop some internal representation of a collection of channels, for example, a list, implement relevant operations on it, and then prove lemmas such as \(\alpha\)-equivalence and substitution.
Though explicit encodings have led to successful mechanizations~\cite{Castro-Perez20tacas, Thiemann19ppdp, Zalakain21forte, Jacobs22popl}, they make it cumbersome to formalize metatheoretic results like subject~reduction.

\emph{Higher-order abstract syntax}~\cite{Pfenning88pldi} (HOAS) relieves us from the bureaucracy of explicitly encoded contexts.
With this approach, variable abstractions are identified with functions in the proof assistant or the host language.
Thus, we can obtain properties of bindings in the host language for free, such as the aforementioned \(\alpha\)-equivalence and substitution lemmas.
This technique had been studied in process calculi without modern linear session types by R{\"o}ckl, Hirschkoff, and Berghofer~\cite{Rockl01fossacs} in Isabelle/HOL and by Despeyroux~\cite{Despeyroux00tcs} in Coq.
However, HOAS has rarely been used to encode linear systems, and it has not yet been applied to mechanize session-typed languages.
This is because most HOAS systems treat contexts structurally while session-typed systems require linear contexts.
Consequently, naively using HOAS to manage channel contexts would not guarantee that channels are treated linearly.
This would in turn make it difficult or impossible to prove metatheoretic properties that rely on linearity, such as deadlock freedom.

In our paper, we develop a technique to bridge the gap between structural and linear contexts.
We use this technique to mechanize a subset of Wadler's Classical Processes (CP)~\cite{Wadler12icfp}.
CP is a well-studied foundation for investigating the core ideas of concurrency due to its tight relation with linear logic.
For our mechanization, we first introduce \emph{Structural Classical Processes} (SCP), a system whose context is structural.
This calculus encodes linearity using a technique heavily inspired by the one \citet{Crary10icfp} used to give a HOAS encoding of the linear $\lambda$-calculus.
The key idea is to define a predicate
\[
  \lin{x}{P}
\]
for some process $P$ that uses a channel $x$.
This predicate can informally be read as ``channel $x$ is used linearly in $P$,'' and it serves as a localized well-formedness predicate on the processes.
We embed these additional proof obligations within type judgments for rules that introduce channel bindings.
Thus, well-typed processes use all of their internally bound names linearly, and we further give a bijection between CP and SCP typing derivations to show that these linearity predicates precisely capture the notion of linear contexts.

We then mechanize SCP in Beluga~\cite{Pientka10ijcar} using weak HOAS.
The mechanization is mostly straightforward due to the strong affinity SCP has with LF, and we prove adequacy of our encoding with respect to SCP.
This adequacy result is compatible with our prior bijection result between CP and SCP, meaning our encoding is also adequate with respect to CP.
Finally, we mechanize type preservation in our encoding in a very elegant manner, taking advantage of the various properties we obtain for free from a HOAS encoding such as renaming, variable dependencies that are enforced via higher-order unification, \etc.

\paragraph{Contributions}
We describe a structural approach to mechanizing session-types and their metatheory without relying on the substructural properties of the session type system, by using explicit linearity check for processes.
In particular:
\begin{itemize}
  \item We introduce an on-paper system equivalent to a subset of Wadler's Classical
    Processes (CP)~\cite{Wadler12icfp}, which we call Structural
    Classical Processes (SCP).
    This system uses a structural context as opposed to a linear context but still captures the intended properties of linearity using linearity predicates.
    SCP is well-suited to a HOAS-style encoding as we demonstrate in
    this paper, but it is also well-suited to other styles of
    mechanizations given that it does not require any context splits.

  \item We define a linearity predicate inspired by \citet{Crary10icfp} for the linear \(\lambda\)-calculus.
    By doing so, we demonstrate the scalability of Crary's technique to richer settings.

  \item We encode processes and session types using weak HOAS in the logical framework LF.
    Our encoding illustrates how we leverage HOAS/LF and its built-in
    higher-order unification to model channel bindings and hypothetical
    session type derivations as intuitionistic functions.

  \item We prove the equivalence of CP and SCP and then show that our encoding of SCP in Beluga is adequate, \ie, that there exist bijections between all aspects of SCP and their encodings.
    We therefore show that our encoding of SCP is adequate with respect to CP as well. Given that adequacy for session typed systems is quite difficult, we believe that the techniques presented in SCP is a useful baseline for more complex systems.
  \item We encode and mechanize SCP in Beluga and prove (on paper) that the encoding is adequate. We further mechanize a subject reduction proof of SCP to illustrate how metatheoretic proofs interact with our linearity predicates.
\end{itemize}
The full mechanization of SCP in Beluga is available as an artifact~\cite{Sano23oopsla-art}.

\section{Classical Processes (CP)}

We present a subset of Wadler's Classical Processes (CP), making minor syntactic changes to better align with our later development.
CP is a proofs-as-processes interpretation of classical linear logic.
It associates to each proof of a classical, linear (one-sided) sequent
\begin{align*}
  &\vdash A_1, \ldots, A_n
  \\
  \intertext{a process $P$ that communicates over channels \(x_1, \dotsc, x_n\):}
  \tp{x_1:A_1, \ldots, x_n:A_n}{P &}.
\end{align*}
We interpret linear propositions $A_1, \dotsc, A_n$ as session types that specify the protocol that $P$ must follow when communicating on channels $x_1, \dotsc, x_n$, respectively.
\autoref{tab:session-types} summarizes the operational interpretation
of the standard linear connectives without exponentials and
quantifiers:

\begin{table}[h!]
  \centering
  \begin{tabular}{c l}
    \toprule
    Type & Action \\
    \midrule
    $1$ & Send a termination signal and then terminate \\
    $\bot$ & Receive a termination signal \\
    $\stensor{A}{B}$ & Send a channel of type $A$ and proceed as $B$ \\
    $\spar{A}{B}$ & Receive a channel of type $A$ and proceed as $B$ \\
    $\ssum{A}{B}$ & Send a ``left'' or ``right'' and then proceed as $A$ or $B$
    accordingly \\
    $\samp{A}{B}$ & Receive a ``left'' or ``right'' and then proceed as $A$ or $B$ accordingly\\
    \bottomrule
  \end{tabular}
  \caption{Interpretation of propositions in linear logic as session types on
  channels in CP}
  \label{tab:session-types}
\end{table}

Logical negation induces an involutory notion of duality on session types, where two types are dual if one can be obtained from the other by exchanging sending and receiving.
This duality will be used in process composition: we can safely compose a process \(P\) communicating on \(x : A\) with a process \(Q\) communicating on \(x : B\) whenever \(A\) and \(B\) are dual.
We write \(A^\bot\) for the dual of \(A\); it is inductively defined on the structure of \(A\):
\begin{align*}
  1^\bot &= \bot  & \bot^\bot &= 1\\
  (\stensor{A}{B})^\bot &= \spar{A^\bot}{B^\bot} & (\spar{A}{B})^\bot &= \stensor{A^\bot}{B^\bot} \\
  (\samp{A}{B})^\bot &= \ssum{A^\bot}{B^\bot} & (\ssum{A}{B})^\bot &= \samp{A^\bot}{B^\bot}
\end{align*}

\subsection{Type Judgments}
Since each inference rule in linear logic corresponds to a process construct, we
define the syntax of the processes alongside the type judgments.

\paragraph{Identity and process composition}
The identity rule globally identifies two channels \(x\) and \(y\).
The duality between the types $A$
and $A^\bot$ ensures that this identification only occurs between channels with compatible
protocols.
\[
  \defrule{Cid}{\crn{Id}}{\tp{x:A, y:A^\bot}{\fwd{x}{y}}}{}
  \getrule*{Cid}
\]

The process composition \(\pcomp{x}{A}{P}{Q}\) spawns processes \(P\) and \(Q\) that communicate along a bound private channel \(x\).
Its endpoints in \(P\) and \(Q\) have type \(A\) and \(A^\bot\), respectively.
Linearity ensures that no other channels are shared between \(P\) and \(Q\).
\[
  \defrule{Ccut}{\crn{Cut}}{\tp{\Delta_1, \Delta_2}{\pcomp{x}{A}{P}{Q}}}
  {\tp{\Delta_1, x:A}{P} & \tp{\Delta_2, x:A^\bot}{Q}}
  \getrule*{Ccut}
\]

\paragraph{Channel transmission}
The two multiplicative connectives $\otimes$ and $\parr$ correspond to sending
and receiving a channel, respectively.
The process $\out{x}{y}{P}{Q}$ sends a channel name $y$ across the channel $x$, and spawns concurrent processes \(P\) and \(Q\) that provide \(x\) and \(y\), respectively.
\[
  \defrule{Cotimes}{\crn{\(\otimes\)}}{\tp{\Delta_1, \Delta_2, x:A \otimes B}{\out{x}{y}{P}{Q}}}
  {\tp{\Delta_1, y:A}{P} & \tp{\Delta_2, x:B}{Q}}
  \getrule*{Cotimes}
\]

The process $\inp{x}{y}{P}$ receives a channel over $x$, binds it to a fresh name $y$, and proceeds as~$P$.
\[
  \defrule{Cparr}{\crn{\(\parr\)}}{\tp{\Delta, x:A \parr B}{\inp{x}{y}{P}}}
  {\tp{\Delta, x:B, y:A}{P}}
  \getrule*{Cparr}
\]

\paragraph{Internal and external choice}
The two additive connectives $\oplus$ and $\with$ respectively specify internal and
external choice.
Internal choice is implemented by processes $\inl{x}{P}$ and $\inr{x}{P}$ that respectively send a ``left'' and ``right'' choice across $x$.
\[
  \defrule{Cinl}{\crn{\(\oplus_1\)}}{\tp{\Delta, x:A \oplus B}{\inl{x}{P}}}
  {\tp{\Delta, x:A}{P}}
  \getrule*{Cinl}
  \qquad
  \defrule{Cinr}{\crn{\(\oplus_2\)}}{\tp{\Delta, x:A \oplus B}{\inr{x}{P}}}
  {\tp{\Delta, x:B}{P}}
  \getrule*{Cinr}
\]
External choice is implemented by a case analysis on a received choice:
\[
  \defrule{Cwith}{\crn{\(\with\)}}{\tp{\Delta, x:A \with B}{\choice{x}{P}{Q}}}
  {\tp{\Delta, x:A}{P} & \tp{\Delta, x:B}{Q}}
  \getrule*{Cwith}
\]
Contrary to previous rules, the context \(\Delta\) in the conclusion is not split between premisses.
This does not violate linearity because only one of the branches will be taken.

\paragraph{Termination}
The multiplicative units $1$ and $\bot$ specify termination and waiting for termination, respectively.
\[
  \defrule{C1}{\crn{\(1\)}}{\tp{x:1}{\close{x}}}{}
  \getrule*{C1}
  \qquad
  \defrule{Cbot}{\crn{\(\bot\)}}{\tp{\Delta, x:\bot}{\wait{x}{P}}}{\tp{\Delta}{P}}
  \getrule*{Cbot}
\]

\subsection{Reductions and Type Preservation} \label{ssec:cp-reduction}
Cut elimination in classical linear logic corresponds to reduction rules for CP processes and therefore reduces parallel compositions of form $\pcomp{x}{A}{P}{Q}$.
For example, if $P = \fwd{x}{y}$, then we have the reduction rule
\[
  \defrule{Cbetafwd}{\crn{$\beta_{\text{fwd}}$}}{\pcomp{x}{A}{\fwd{x}{y}}{Q} \st \subst{y}{x}{Q}}{}
  \getrule*{Cbetafwd}
\]
Other reduction rules are categorized into \emph{principal} reductions, where both $P$ and $Q$ are attempting to communicate over the same channel, \emph{commuting conversions}, where we can push the cut inside $P$, and \emph{congruence} rules.
We treat all other processes, \eg, $\inp{x}{y}{P}$, as stuck processes waiting to communicate with an external agent.

An example of a principal reduction occurs with the composition of $P = \inl{x}{P'}$ and $Q = \choice{x}{Q_1}{Q_2}$.
After communication, the left process continues as $P'$ and the right process as $Q_1$, since the ``left'' signal was sent by $P$.
\[
  \defrule{Cbetainl}{\crn{$\beta_{\text{inl}}$}}{\pcomp{x}{A \oplus B}{\inl{x}{P'}}{\choice{x}{Q_1}{Q_2}} \st \pcomp{x}{A}{P'}{Q_1}}{}
  \getrule*{Cbetainl}
\]

An example of a commuting conversion occurs when $P = \inl{x}{P'}$ and the abstracted channel is some $z$ such that $x \neq z$. In this case, we push the cut inside $P$.
\[
  \defrule{Ckappainl}{\crn{$\kappa_{\text{inl}}$}}{\pcomp{z}{C}{\inl{x}{P'}}{Q} \st \inl{x}{\pcomp{z}{C}{P'}{Q}}}{}
  \getrule*{Ckappainl}
\]

Finally, the congruence rules enable reduction under cuts.
We follow Wadler's formulation and do not provide congruence rules for other process constructs.
Such rules would eliminate internal cuts and do not correspond to the intended notion of computation, analogously to not permitting reduction under \(\lambda\)-abstractions.
\[
  \defrule{Cbetacut1}{\crn{$\beta_{\text{cut}1}$}}{\pcomp{x}{A}{P}{Q} \st \pcomp{x}{A}{P'}{Q}}{P \st P'}
  \defrule{Cbetacut2}{\crn{$\beta_{\text{cut}2}$}}{\pcomp{x}{A}{P}{Q} \st \pcomp{x}{A}{P}{Q'}}{Q \st Q'}
  \getrule*{Cbetacut1}
  \quad
  \getrule*{Cbetacut2}
\]

\defrule{Cequivcomm}{\crn{$\equiv_{\text{comm}}$}}{\pcomp{x}{A}{P}{Q} \equiv \pcomp{x}{A^\bot}{Q}{P}}{}
\defrule{Cequivassoc}{\crn{$\equiv_{\text{assoc}}$}}{\pcomp{y}{B}{\pcomp{x}{A}{P}{Q}}{R} \equiv \pcomp{x}{A}{P}{\pcomp{y}{B}{Q}{R}}}{}

We close these rules under structural equivalences $P \equiv Q$, which says that parallel composition is commutative and associative:
\begin{gatherrules}
  \getrule*{Cequivcomm}
  \\
  \getrule*{Cequivassoc}
\end{gatherrules}
For \getrn{Cequivassoc}, it is implicit that the process $P$ does not depend on the channel $y$.

For later developments, we define the closure explicitly as a reduction rule:
\[
  \defrule{Cbetaequiv}{\crn{$\beta_\equiv$}}{P \st S}{P \equiv Q & Q \st R & R \equiv S}
  \getrule*{Cbetaequiv}
\]
which also requires adding reflexitivity and transitivity to $\equiv$.

\begin{theorem}[Type Preservation of CP] \label{thm:cp-sr}
  If\/$\tp{\Delta}{P}$ and $P \st Q$, then $\tp{\Delta}{Q}$.
\end{theorem}

\section{Structural Classical Processes (SCP)}

We introduce \textit{Structural Classical Processes} (SCP).
SCP is a reformulation of Classical Processes using a structural context, \ie, in which weakening and contraction hold.
The property we would like to enforce is that the context can only grow as we move upwards in a typing derivation.
This property makes SCP well-suited for mechanizations and in particular HOAS encodings since no complex operations such as context splitting are necessary.
Of course, simply adopting structural rules on top of CP is insufficient because linearity is needed to prove its safety theorems.
Instead, we use local linearity predicates to enforce linearity on a global level.
These linearity checks are given by a judgment $\lin{x}{\fns{P}}$ that informally means ``$x$ occurs linearly in the process $\fns{P}$''.

SCP's syntax is similar to CP's. In particular, we use the same syntax for session types and the same notion of duality.
However, SCP's process syntax explicitly tracks the continuation channels that are left implicit in CP's typing rules (and other on-paper systems).
To illustrate, contrast the CP process $\inl{x}{P}$ with the corresponding SCP process \(\sinl{x}{w}{\fns{P}}\) and their associated typing~rules:
\[
  \getrule{Cinl} \qquad \getrule{Sinl}
\]
In \getrn{Cinl}, the assumption $x:A \oplus B$ in the conclusion is replaced by $x:A$ in the premise, violating our principle that we may only grow contexts.
SCP respects the principle thanks to two changes.
First, the syntax \(\sinl{x}{w}{\fns{P}}\) binds a name \(w\) in \(P\) for the continuation channel of \(x\).
This in turn lets us grow the context in the premise of \getrn{Sinl} with an assumption \(w : A\), while keeping the assumption $x:A \oplus B$.
Our linearity predicate ensures that the continuation channel \(w\) is used instead of \(x\) in \(P\), making these modifications safe.
We explain SCP typing judgments~below.

SCP is a faithful structural encoding of CP: we give a bijection between well-typed CP processes and well-typed linear SCP processes.
Accordingly, we encode SCP instead of CP in LF, and we rely on our equivalence proof to mediate between CP and our LF mechanization of SCP.

\subsection{Type Judgments}

We write \(\stp{\Gamma}{\fns{P}}\) for SCP typing judgments to differentiate them from CP typing judgments \(\tp{\Delta}{P}\).
The context \(\Gamma\) is structural: it enjoys the weakening, contraction, and exchange properties.
Intuitively, it represents the ambient LF context.

\paragraph{Identity and Cut}

Axioms use arbitrary contexts \(\Gamma\) to allow for weakening:
\[
  \getrule*{Sid}
\]

We write \(\spcomp{x}{A}{\fns{P}}{\fns{Q}}\) for the composition of \(\fns{P}\) and \(\fns{Q}\) along a private, bound channel \(x\).
Contrary to the typing rule \getrn{Ccut} in CP, the cut rule in SCP does not split contexts.
This is because contexts can only grow as we move upwards in SCP typing derivations.
\[
  \getrule*{Scut}
\]
This rule illustrates a general design principle of SCP: we must check that any channel introduced in the continuation of a process is used linearly.
In particular, \getrn{Scut} checks that \(\fns{P}\) and \(\fns{Q}\) use the free channel \(x\) linearly.

\paragraph{Choices}

The choice rules explicitly track continuation channels.
In particular, the processes \(\sinl{x}{w}{\fns{P}}\) and \(\sinr{x}{w}{\fns{P}}\) bind the name \(w\) in \(\fns{P}\).
This name stands in for the continuation channel of \(x\) after it has transmitted a left or right label.
The rules \getrn{Sinl} and \getrn{Sinr} grow the context and ensure that \(w\) has the appropriate type in \(\fns{P}\).
We remark that these two rules do not preclude \(x\) and \(w\) from both appearing in \(\fns{P}\).
However, this will be ruled out by our linearity predicate, which checks that \(x\) and its continuation channels are used linearly in \(\sinl{x}{w}{\fns{P}}\) or \(\sinr{x}{w}{\fns{P}}\).
The treatment of continuation channels in the rule \getrn{Swith} is analogous.
\begin{gatherrules}
  \getrule*{Sinl}
  \qquad
  \getrule*{Sinr}
  \\
  \getrule*{Swith}
\end{gatherrules}

\paragraph{Channel Transmission}

The channel transmission rules follow the same principles as the identity and cut rules.
In particular, they do not split channel contexts between processes, and they check that freshly introduced channels are used linearly.
The names \(y\) and \(w\) are bound in \(\sout{x}{y}{\fns{P}}{w}{\fns{Q}}\) and in \(\sinp{x}{y}{w}{\fns{P}}\).
\begin{gatherrules}
  \getrule*{Sotimes}
  \\
  \getrule*{Sparr}
\end{gatherrules}

\paragraph{Termination} The rules for termination are analogous:
\begin{gatherrules}
  \getrule*{S1}
  \qquad
  \getrule*{Sbot}
\end{gatherrules}

\subsection{Linearity Predicate} \label{ssec:scp-lin}
We now define the predicate $\lin{x}{\fns{P}}$. It syntactically checks
that a free channel $x$ and its continuations occur linearly in
$\fns{P}$. This judgment is generic relative to an implicit context of channel names that can be freely renamed, and we assume that this implicit context contains the free names  $\fn(\fns{P})$ of the process $\fns{P}$.
The linearity predicate $\lin{x}{\fns{P}}$ is inductively defined by the following rules, which we informally group into two categories.
The first category specifies when a process uses its principal channels linearly.
The axioms in this category are:
\begin{gatherrules}
  \getrule*{Lfwd1}
  \quad
  \getrule*{Lfwd2}
  \quad
  \getrule*{Lclose}
  \quad
  \getrule*{Lwait}
\end{gatherrules}
For process constructs whose principal channel $x$ would persist in CP, we must check that its continuation channel \(w\) is used linearly in its continuation process \textbf{and} that the original channel \(x\) does not appear in the continuation, thereby capturing the property that \(w\) is the continuation of \(x\).
\begin{gatherrules}
  \getrule*{Lout}
  \quad
  \getrule*{Linp}
  \\
  \getrule*{Linl}
  \quad
  \getrule*{Linr}
  \\
  \getrule*{Lcase}
\end{gatherrules}
These rules do not check the linearity of freshly bound channels, for example, of the channel $y$ in channel output or channel input.
This is because the predicate only checks the linearity of free channels and their continuations.
Although this predicate does not check the linearity of fresh channels such as \(y\), our type system ensures their linear use in well-typed processes.

The second category of rules are congruence cases in which we check the linearity of non-principal channels.
We implicitly assume throughout that \(z\) is distinct from any bound name:
\begin{gatherrules}
  \getrule*{Lwait2}
  \quad
  \getrule*{Lout2}
  \quad
  \getrule*{Lout3}
  \\
  \getrule*{Linp2}
  \quad
  \getrule*{Linl2}
  \quad
  \getrule*{Linr2}
  \\
  \getrule*{Lcase2}
  \quad
  \getrule*{Lpcomp1}
  \quad
  \getrule*{Lpcomp2}
\end{gatherrules}
When checking that \(z\) appears linearly in processes whose context would be split by the typing rules in CP, namely, in channel output and parallel composition, we ensure that \(z\) appears in at most one of the subprocesses.
This lets us use our linearity predicate to mimic context splitting in the presence of structural ambient contexts.

\begin{example}
  \label{ex:main_writeup:1}
  There exists a well-typed SCP process that is not linear, to wit,
  \[
    \infer[\getrn{Sbot}]{
      \stp{x : 1, y : \bot}{\swait{y}{\swait{y}{\sclose{x}}}}
    }{
      \infer[\getrn{Sbot}]{
        \stp{x : 1, y : \bot}{\swait{y}{\sclose{x}}}
      }{
        \infer[\getrn{S1}]{
          \stp{x : 1, y : \bot}{\sclose{x}}
        }{
        }
      }
    }
  \]
  However, it is not the case that \(\lin{y}{\swait{y}{\swait{y}{\sclose{x}}}}\).
  Indeed, the only rule with a conclusion of this form is \getrn{Lwait}, but it is subject to the side condition \(y \notin \fn(\swait{y}{\sclose{x}})\).
\end{example}

\subsection{Equivalence of CP and SCP}

\newcommand{\escp}[1]{\varepsilon(#1)}
\newcommand{\dscp}[1]{\delta(#1)}

We establish a correspondence between CP and SCP typing derivations.
Because CP and SCP use slightly different process syntax, we first define an encoding $\escp{P}$ and a decoding $\dscp{\fns{P}}$ that maps a process in CP to SCP and SCP to CP respectively.
We give several representative cases:
\begin{align*}
  \escp{\fwd{x}{y}} &= \sfwd{x}{y} & \dscp{\sfwd{x}{y}} &= \fwd{x}{y} \\
  \escp{\pcomp{x}{A}{P}{Q}} &= \spcomp{x}{A}{\escp{P}}{\escp{Q}} & \dscp{\spcomp{x}{A}{\fns{P}}{\fns{Q}}} &= \pcomp{x}{A}{\dscp{\fns{P}}}{\dscp{\fns{Q}}} \\
  \escp{\inp{x}{y}{P}} &= \sinp{x}{y}{x}{\escp{P}} & \dscp{\sinp{x}{y}{w}{\fns{P}}} &= \inp{x}{y}{\subst{x}{w}{\dscp{\fns{P}}}} \\
  \escp{\inl{x}{P}} &= \sinl{x}{x}{\escp{P}} & \dscp{\sinl{x}{w}{\fns{P}}} &= \inl{x}{\subst{x}{w}{\dscp{\fns{P}}}}
\end{align*}

The bijection between well-typed processes is subtle because we must account for different structural properties in each system and slight differences in the process syntax.
For example, the judgment \(\stp{\Gamma, x : 1}{\close{x}}\) is derivable in SCP for any \(\Gamma\), whereas the judgment \(\tp{\Gamma, x : 1}{\close{x}}\) is derivable in CP only if \(\Gamma\) is empty.
The key insight is that the bijection holds only if the SCP process uses each channel in its context linearly.
This restriction to linear SCP processes is unproblematic because we only ever consider such processes in our development.

Before stating the equivalence theorem, we introduce two lemmas that we use in its proof.
Both lemmas are proved by induction on the derivation of the typing judgment.

\begin{lemma}[Weakening] \label{lem:scp-wk}
  If\/ $\stp{\Gamma}{\fns{P}}$, then $\stp{\Gamma, x:A}{\fns{P}}$.
\end{lemma}

\begin{lemma}[Strengthening] \label{lem:scp-str}
  If\/ $\stp{\Gamma, x:A}{\fns{P}}$ and $x \notin \fn(\fns{P})$, then $\stp{\Gamma}{\fns{P}}$.
\end{lemma}

\newcommand{\linp}[1]{\ms{lin}(#1)}

\paragraph{Notation}
We write $\lin{\Delta}{\fns{P}}$ as shorthand for $\forall x \in \dom(\Delta).
\lin{x}{\fns{P}}$.

The equivalence theorem shows that we can not only faithfully embed CP processes in SCP but also their typing derivations.
Indeed, \Cref{thm:cp-scp} states that each CP derivation determines the typing derivation of a linear SCP process and that each typing derivation of a linear SCP process can be obtained by weakening a CP typing derivation.
This structure-preserving embedding of CP derivations in SCP is given by induction on the derivation.
The general strategy is that we interleave the CP derivation with the appropriate linearity checks.

\begin{theorem}[Adequacy] \label{thm:cp-scp}
  The function \(\delta\) is left inverse to \(\varepsilon\), \ie, \(\dscp{\escp{P}} = P\) for all CP processes~\(P\).
  The syntax-directed nature of \(\varepsilon\) and \(\delta\) induces functions between CP typing derivations and typing derivations of linear SCP processes:
  \begin{enumerate}
    \item If \(\mc{D}\) is a derivation of\/ \(\tp{\Delta}{P}\), then there exists a derivation \(\escp{\mc{D}}\) of \(\stp{\Delta}{\escp{P}}\), and $\lin{\Delta}{\escp{P}}$ and \(\dscp{\escp{\mc{D}}} = \mc{D}\).
    \item If \(\mc{D}\) is a derivation of \(\stp{\Gamma, \Delta}{\fns{P}}\) where \(\fn(\fns{P}) = \dom(\Delta)\) and \(\lin{\Delta}{\fns{P}}\), then there exists a derivation \(\dscp{\mc{D}}\) of\/ \(\tp{\Delta}{\dscp{\fns{P}}}\), and \(\escp{\dscp{\fns{P}}} = \fns{P}\).
      Moreover, \(\mc{D}\) is the result of weakening the derivation \(\escp{\dscp{\mc{D}}}\) of\/ \(\stp{\Delta}{\fns{P}}\) by \(\Gamma\).
  \end{enumerate}
\end{theorem}

\subsection{Reduction and Type Preservation}

The dynamics of SCP is given by translation to and from CP.
In particular, we write \(\fns{P} \sst \fns{Q}\) whenever \(\dscp{\fns{P}} \st Q\) and \(\escp{Q} = \fns{Q}\) for some CP process \(Q\).
This translation satisfies the usual type-preservation property:


\begin{lemma} \label{lem:scp-lin-fn}
  If\/ \(\stp{\Delta}{\fns{P}}\) and \(\lin{\Delta}{\fns{P}}\), then \(\fn(\fns{P}) = \dom(\Delta)\).
\end{lemma}
\begin{proof}
  By induction, $\lin{x}{\fns{P}}$ implies $x \in \fn(\fns{P})$, so $\lin{\Delta}{\fns{P}}$ implies $\dom(\Delta) \subseteq \fn(\fns{P})$.
  For the opposite inclusion, $\stp{\Delta}{\fns{P}}$ implies $\dom(\Delta) \supseteq \fn(\fns{P})$ by induction, so $\fn(\fns{P}) = \dom(\Delta)$.
\end{proof}

\begin{theorem}[Subject Reduction]
  \label{thm:scp-sr}
  If\/ \(\stp{\Delta}{\fns{P}}\), \(\lin{\Delta}{\fns{P}}\), and \(\fns{P} \sst \fns{Q}\), then \(\stp{\Delta}{\fns{Q}}\) and \(\lin{\Delta}{\fns{Q}}\).
\end{theorem}

\begin{proof}
  Assume \(\stp{\Delta}{\fns{P}}\), \(\lin{\Delta}{\fns{P}}\), and \(\fns{P} \sst \fns{Q}\).
  Then \(\fn(\fns{P}) = \dom(\Delta)\) by \Cref{lem:scp-lin-fn}.
  Adequacy (\Cref{thm:cp-scp}) implies \(\tp{\Delta}{\dscp{\fns{P}}}\).
  By the assumption \(\fns{P} \sst \fns{Q}\), there exists a \(Q\) such that \(\dscp{\fns{P}} \st Q\) and \(\escp{Q} = \fns{Q}\).
  Subject reduction for CP (\Cref{thm:cp-sr}) implies \(\tp{\Delta}{Q}\), so \(\stp{\Delta}{\fns{Q}}\) and \(\lin{\Delta}{\fns{Q}}\) by adequacy again.
\end{proof}

We could instead directly prove \Cref{thm:scp-sr} by induction on the reduction.
This direct proof is mechanized as \Cref{thm:lf-sr}.

\defrule{Sbetafwd}{\srn{$\beta_{\text{fwd}}$}}{\spcomp{x}{A}{\sfwd{x}{y}}{\fns{Q}} \sst \subst{y}{x}{\fns{Q}}}{\mathstrut}
\defrule{Sbetainl1}{\srn{$\beta_{\text{inl}1}$}}{\spcomp{x}{A \oplus B}{\sinl{x}{w}{\fns{P}}}{\schoice{x}{w}{\fns{Q_1}}{w}{\fns{Q_2}}} \sst \spcomp{w}{A}{\fns{P}}{\fns{Q_1}}}{\mathstrut}
\defrule{Skappainl}{\srn{$\kappa_{\text{inl}}$}}{\spcomp{z}{C}{\sinl{x}{w}{\fns{P}}}{\fns{Q}} \sst \sinl{x}{w}{\spcomp{z}{C}{\fns{P}}{\fns{Q}}}}{\mathstrut}
\defrule{Sbetacut1}{\srn{$\beta_{\text{cut}1}$}}{\spcomp{x}{A}{\fns{P}}{\fns{Q}} \sst \spcomp{x}{A}{\fns{P'}}{\fns{Q}}}{\fns{P} \sst \fns{P'}}
\defrule{Sequivcomm}{\srn{$\equiv_{\text{comm}}$}}{\spcomp{x}{A}{\fns{P}}{\fns{Q}} \equiv \spcomp{x}{A^\bot}{\fns{Q}}{\fns{P}}}{\mathstrut}
\defrule{Sequivassoc}{\srn{$\equiv_{\text{assoc}}$}}{\spcomp{y}{B}{\spcomp{x}{A}{\fns{P}}{\fns{Q}}}{\fns{R}} \equiv \spcomp{x}{A}{\fns{P}}{\spcomp{y}{B}{\fns{Q}}{\fns{R}}}}{\mathstrut}

Since we mechanize SCP, it is convenient to have the reduction and equivalence rules expressed directly in SCP.
We show some such rules below.
They are obtained by translating the rules in \autoref{ssec:cp-reduction} (the second congruence rule for cut omitted).
\begin{gatherrules}
  \getrule*{Sbetafwd}
  \quad
  \getrule*{Sbetacut1}
  \\
  \getrule*{Sbetainl1}
  \\
  \getrule*{Skappainl}
\end{gatherrules}
We obtain SCP's structural equivalence in a similar manner:  \(\fns{P} \equiv \fns{Q}\) whenever \(\dscp{\fns{P}} \equiv \dscp{\fns{Q}}\). We show two cases of this direct translation.
\begin{gatherrules}
  \getrule*{Sequivcomm}
  \quad
  \getrule*{Sequivassoc}
\end{gatherrules}

\section{Encoding SCP in LF}
\label{sec:encoding}
We now encode each component of SCP in the logical framework LF. 
Throughout this section, we make liberal modifications to the working code for presentation/readability purposes.

\subsection{Types}
We encode session types in LF by defining the LF type  @tp: type@.
The type constants for this type correspond to the type constructors in SCP.
\begin{multicols}{2}
  \begin{lstlisting}
1 : tp.               % termination ("provider")
bot : tp.               % termination ("client")
times : tp → tp → tp. % channel output
  \end{lstlisting}
  \columnbreak
  \begin{lstlisting}
par : tp → tp → tp. % channel input
& : tp → tp → tp. % receive choice
+ : tp → tp → tp. % send choice
  \end{lstlisting}
\end{multicols}
We use the LF type family @dual: tp → tp → type@ to represent duality as a relation between two types.
The constants of this type family correspond to the equational definition of duality.
In particular, @dual A A'@ encodes $A = A^\bot$ (where $A^\bot = A'$).
\begin{multicols}{2}
  \begin{lstlisting}
D1 : dual 1 ⊥.
D⊥ : dual ⊥   1.
D⊗: dual A A' → dual B B'
   → dual (A ⊗ B) (A' ⅋ B').
  \end{lstlisting}
  \columnbreak
  \begin{lstlisting}
D⅋: dual A A' → dual B B'
   → dual (A ⅋ B) (A' ⊗ B').
D&: dual A A' → dual B B'
   → dual (A & B) (A' ⊕ B').
D⊕: dual A A' → dual B B'
   → dual (A ⊕ B) (A' & B').
 \end{lstlisting}
\end{multicols}

\subsection{Processes}
We give an encoding of processes by interpreting all channel
bindings as intuitionistic functions in LF.  First, we define channel names as the
type family @name@. Unlike in the functional setting where everything is an
expression, in the process calculus setting, channels and processes are
distinct. This leads to a so-called weak-HOAS encoding~\cite{Despeyroux:TLCA95}. We then
introduce the predicate @proc@, standing for processes.

\begin{lstlisting}
name : type.      % channel names
proc : type.      % process
\end{lstlisting}

We first encode $\sfwd{x}{y}, \sclose{x},$  and $\swait{x}{\fns{P}}$, which introduce no channel bindings.
The former requires two names and the latter two require one name.
\begin{lstlisting}
fwd : name → name → proc.       % fwd x y
close : name → proc.            % close x
wait : name → proc → proc. % wait x; P
\end{lstlisting}

The processes $\sinl{x}{w}{\fns{P}}, \sinr{x}{w}{\fns{P}},$ and $\schoice{x}{w}{\fns{P}}{w}{\fns{Q}}$ first require some name $x$.
They then bind a fresh continuation channel $w$ to the continuation processes $P$ and $Q$.
We therefore encode the continuation processes as intuitionistic functions @name → proc@.
\begin{lstlisting}
inl : name → (name → proc) → proc.  % x.inl; w.P
inr : name → (name → proc) → proc.  % x.inr; w.P
choice : name → (name → proc) → (name → proc) → proc. % case x (w.P, w.Q)
\end{lstlisting}

Channel output $\sout{x}{y}{\fns{P}}{w}{\fns{Q}}$ binds the channel $y$ to the process $\fns{P}$ and
the continuation channel $w$ to the process $\fns{Q}$.
We therefore encode $y.\fns{P}$ and $w.\fns{Q}$ as intuitionistic functions @name → proc@:
\begin{lstlisting}
out : name → (name → proc) → (name → proc) → proc.  % out x y; (y.P || w.Q)
\end{lstlisting}

Similarly, channel input $\sinp{x}{y}{w}{\fns{P}}$ binds two channels to $\fns{P}$: the continuation channel $w$ and the received channel $y$.
We therefore encode $\fns{P}$ as a function with two channel names as input.
\begin{lstlisting}
inp : name → (name → name → proc) → proc.   % inp x; (w.y.P)
\end{lstlisting}

Parallel composition $\spcomp{x}{A}{\fns{P}}{\fns{Q}}$ takes some session type $A$ and binds a fresh $x$ to both $\fns{P}$ and $\fns{Q}$, so we encode both processes as functions.
\begin{lstlisting}
pcomp : tp → (name → proc) → (name → proc) → proc. % νx:A. (P || Q)
\end{lstlisting}

\subsection{Linearity Predicate}
On paper, we inductively defined a predicate $\lin{x}{\fns{P}}$ that checks if $x$
occurs ``linearly'' in a process $\fns{P}$. This predicate clearly respects
renaming -- if $\lin{x}{\fns{P}}$ and $y$ is fresh with respect to $\fns{P}$, then $\lin{y}{[y/x]\fns{P}}$.
We encode this predicate in LF as a type family over functions from names $x$ to processes $\fns{P}$.
Inhabitants of this family correspond to functions that produces a process that treats its input channel linearly.
\begin{lstlisting}
linear : (name → proc) → type.
\end{lstlisting}


Unlike our encodings of types and duality, processes can depend
on assumptions of the form @x1:name, ..., xn:name@ that are stored
in the so-called \emph{ambient} context. In fact, in Beluga, we always
consider an object with respect to the context in which it is
meaningful. In the on-paper definition of linearity (see
\cref{ssec:scp-lin}) we left this context implicit and only remarked
that the set of free names  $\fn(\fns{P})$ of a process $\fns{P}$ is a subset of
this ambient context of channel names. However, when we encode the
linearity predicate in LF, we need to more carefully quantify over
channel names as we recursively analyze the linearity of a given process.

Intuitively, we define the constructors for linearity by pattern
matching on various process constructors. By convention, we will use
capital letters for metavariables that are implicitly quantified at
the outside. 
These metavariables describe closed LF terms; in particular when the
metavariables stand for processes, it requires that the processes
\emph{not} depend on any local, internal bindings.
We heavily exploit this feature in our encoding to obtain side conditions of the form $x \notin \fn(P)$ for free.

We begin by translating the axioms in \autoref{ssec:scp-lin}:\nopagebreak\\[-0.5\baselineskip]\nopagebreak%
\begin{minipage}[t]{0.4\linewidth}
\begin{lstlisting}
l_fwd1  : linear (\x. fwd x Y).
l_fwd2  : linear (\x. fwd Y x).
l_close : linear (\x. close x).
l_wait  : linear (\x. wait x P).
\end{lstlisting}
\end{minipage}%
\begin{minipage}[t]{0.6\linewidth}
  \begin{gatherrules}
    \getrule{Lfwd1}
    \qquad
    \getrule{Lfwd2}
    \\
    \getrule{Lclose}
    \qquad
    \getrule{Lwait}
  \end{gatherrules}
\end{minipage}\\

\smallskip
\noindent
Here, @Y:name@ in both @l_fwd1@ and @l_fwd2@ are implicitly quantified
at the outside and cannot depend on the input channel \ie. $x
\neq Y$.
Similarly, the metavariable @P:proc@ in @l_wait@ cannot depend on
the input channel $x$, satisfying the condition that
$\getrh{Lwait}{1}$.

The remaining principal cases must continue to check for linearity in the
continuation process. Consider the principal case for channel output:\\[-0.5\baselineskip]
\begin{minipage}[t]{0.5\linewidth}%
  \begin{lstlisting}
% where Q : (name → proc)
l_out : linear Q → linear (\x. out x P Q).
  \end{lstlisting}
\end{minipage}%
\begin{minipage}[t]{0.5\linewidth}%
  \[
    \getrule{Lout}
  \]
\end{minipage}\\
\noindent The premise $\getrh{Lout}{1}$ corresponds to the input @linear Q@ for this constructor because we encode $\fns{Q}$ as a function @name → proc@.
The additional condition that $x$ does not appear in $\fns{P}$ and $\fns{Q}$ follows because @P@ and @Q@ are metavariables, meaning they cannot depend on the internally bound @x:name@.

The encoding of the principal case for channel input requires a bit more care.
Recall the on-paper rule:
\[
  \getrule{Linp}
\]
Following the strategy for channel output, we would like to continue checking that the continuation channel $w$ appears linearly in $\fns{P}$ by requiring it as an input in our encoding.
But since we encode $\fns{P}$ as a two argument function @name → name → proc@, we cannot simply say
\begin{lstlisting}
l_inp : linear P
      → linear (\x. inp x P). % WRONG
\end{lstlisting}
Instead, what we need as our premise is the fact that $\fns{P}$ is linear with respect to some input $w$ given \emph{any} $y$.
To check this, we universally quantify over @y@
using the syntax @{y:name}@:
\begin{lstlisting}
l_inp : ({y:name} linear (\w. P w y))
      → linear (\x. inp x P).
\end{lstlisting}
The condition that $x$ does not appear in $\fns{P}$ again follows from the
fact that $\fns{P}$ must be closed.

The other principal cases are standard translations, which we present in a less verbose manner.
The continuation channels are checked in the same style as in channel output.\\[-0.5\baselineskip]
\begin{minipage}[t]{0.45\linewidth}
  \begin{lstlisting}
l_inl : linear P → linear (\x. inl x P).
l_inr : linear P → linear (\x. inr x P).
l_choice : linear P → linear Q
         → linear (\x. choice x P Q).
  \end{lstlisting}
\end{minipage}\hspace{-2.1em}%
\begin{minipage}[t]{0.65\linewidth}
\begin{gatherrules}
    \getrule{Linl}
    \quad
    \getrule{Linr}
    \\
    \getrule{Lcase}
\end{gatherrules}
\end{minipage}
\smallskip
The congruence cases follow similar ideas except with complex bindings as in the
principal case for input. The simplest case is the encoding of wait:
\begin{multicols}{2}
  \begin{lstlisting}
l_wait2 : linear P → linear (\z. wait X (P z)).
  \end{lstlisting}
  \columnbreak
  \[
    \getrule{Lwait2}
  \]
\end{multicols}
\noindent
Here, it is important to recognize that @(P z)@ is of type @proc@ according to the @wait@ constructor, meaning @P@ is of type @name → proc@.
Therefore, requiring @linear P@ corresponds to checking $\getrh{Lwait2}{1}$.

The congruence case for input is perhaps the most extreme instance of this complex binding:
\begin{multicols}{2}
  \begin{lstlisting}
l_inp2 : ({w:name}{y:name} linear (\z. P z w y))
       → linear (\z. inp X (P z)).
  \end{lstlisting}
  \columnbreak
  \[
    \getrule{Linp2}
  \]
\end{multicols}
\noindent
Here, @(P z)@ is of type @name → name → proc@, so we check for linearity of @z@ by requiring it to be linear with any @w@ and @y@.

Next, we consider the congruence cases for parallel composition.
\begin{multicols}{2}
  \begin{lstlisting}
l_pcomp1 : ({x:name} linear  (\z. P x z))
         → linear (\z. (pcomp A (\x. P x z) Q)).
l_pcomp2 : ({x:name} linear  (\z. Q x z))
         → linear (\z. pcomp A P (\x. Q x z)).
  \end{lstlisting}
  \columnbreak
  \begin{gatherrules}
    \getrule{Lpcomp1}
    \\
    \getrule{Lpcomp2}
  \end{gatherrules}
\end{multicols}
Since @Q@ is a metavariable in @l_pcomp1@, it must be closed with respect to @z@, so it satisfies the condition $\getrh{Lpcomp1}{2}$.
The condition $\getrh{Lpcomp2}{2}$ in @l_pcomp2@ is satisfied for the same reason.

We summarize the remaining cases below.\\[-0.5\baselineskip]
\begin{minipage}[t]{0.55\linewidth}
  \begin{lstlisting}
l_out2 : ({y:name} linear (\z. P z y))
       → linear (\z. out X (P z) Q).
l_inl2 : ({x':name} linear (\z. P z x'))
       → linear (\z. inl X (P z)).
l_choice2 : ({x':name} linear (\z. P z x'))
          → ({x':name} linear (\z. Q z x'))
          → linear (\z. choice X (P z) (Q z)).
  \end{lstlisting}
\end{minipage}%
\begin{minipage}[t]{0.45\linewidth}
  \begin{lstlisting}
l_out3 : ({x':name} linear (\z. Q z x'))
       → linear (\z. out X P (Q z)).
l_inr2 : ({x':name} linear (\z. P z x'))
       → linear (\z. inr X (P z)).
  \end{lstlisting}
\end{minipage}

\subsection{Type Judgments}

To encode session typing, we follow the encoding for the sequent
calculus in the logical framework LF (see for example \cite{twelf-popl-tutorial}).
Since type judgments depend on assumptions of the form $x:A$, we introduce
the type family @hyp : name → tp → type@ to associate a channel name
with a session type.
We then encode the type judgment $\stp{\Gamma}{P}$ as a judgment on a
process: @wtp : proc → type@ with ambient assumptions of the form
@x1:name,h1:hyp x1 A1, ..., xn:name,hn:hyp xn An@ which represent
$\Gamma$. Note that the use of these assumptions is unrestricted, but
the linearity predicate ensures that if an assumption is used,
then it is used linearly.
As an example, we could encode the rule
\[
  \getrule{S1}
\]
in an obvious manner:
\begin{lstlisting}
wtp_close : {X:name}hyp X 1 → wtp (close X).
\end{lstlisting}
To establish \lstinline!wtp (close X)!, we must have an assumption
\lstinline!hyp X 1!. While it is not strictly necessary to
explicitly  quantify over the channel name \lstinline!X!,
doing so makes encoding the metatheory easier.

Forwarding requires two channels of dual type:\\[-0.5\baselineskip]
\begin{minipage}[t]{0.5\linewidth}%
  \begin{lstlisting}
wtp_fwd : dual A A'
        → {X:name} hyp X A → {Y:name} hyp Y A'
        → wtp (fwd X Y).
  \end{lstlisting}
\end{minipage}%
\begin{minipage}[t]{0.5\linewidth}
  \vspace{0.75em}
  \[
    \getrule{Sid}
  \]
\end{minipage}\\
We encode this rule by requiring a duality relation between two session types $A$ and $A'$ alongside corresponding hypotheses that $X$ and $Y$ are of type $A$ and $A'$ respectively.

The encoding of parallel composition requires a similar trick for duality.\\[-0.5\baselineskip]
\noindent\begin{minipage}[t]{0.4\linewidth}%
  \begin{lstlisting}
wtp_pcomp : dual A A'
  → ({x:name} hyp x A → wtp (P x))
  → ({x:name} hyp x A' → wtp (Q x))
  → linear P → linear Q
  → wtp (pcomp A P Q).
  \end{lstlisting}
\end{minipage}%
\begin{minipage}[t]{0.6\linewidth}%
  \vspace{1em}
  \[
    \getrule{Scut}
  \]
\end{minipage}

\noindent
We encode the premise $\getrh{Scut}{1}$ as a function that takes
some @x:name@ and assumption @hyp x A@ to prove that @(P x)@ is well-typed.
A different reading of this premise is simply as ``for all
@x:name@, assuming  @hyp x A@, we show that @wtp (P x)@''.
The premise $\getrh{Scut}{2}$ corresponds to @linear P@ since @P@ is of type @name -> proc@, and the remaining two premises follow the same idea.

Continuation channels are simply treated as bindings in the same way we treat cut. For instance:\\[-0.5\baselineskip]
\begin{minipage}[t]{0.5\linewidth}%
  \begin{lstlisting}
wtp_inl : {X:name} hyp X (A + B)
          → ({w:name} hyp w A → wtp (P w))
          → wtp (inl X P).
\end{lstlisting}
\end{minipage}%
\begin{minipage}[t]{0.5\linewidth}
  \vspace{0.5em}
  \[
    \getrule{Sinl}
  \]
\end{minipage}\\
The first two inputs to the constructor is a name @X@ and a hypothesis that @X@ is of type @A + B@.
The next input is that the continuation process @(P w)@ is well-typed given an assumption @w:name@ and @hyp w A@, corresponding to the premise of the \getrn{Sinl} rule.

The remaining cases follows a similar pattern. Linearity is checked for the freshly bound channels on channel output and input as in the typing for parallel composition.
We defer the full encoding to the attached artifact.

\subsection{Reductions and Structural Equivalence}
We model both reductions $\fns{P} \sst \fns{Q}$ and structural equivalences $\fns{P} \equiv \fns{Q}$ as
relations.
\begin{lstlisting}
step : proc → proc → type.
equiv : proc → proc → type.
\end{lstlisting}

The encoding is fairly simple. For example, consider\\
\begin{minipage}[t]{0.5\linewidth}%
  \begin{lstlisting}
betafwd : step (pcomp A (\x. fwd x Y) Q) (Q Y).
  \end{lstlisting}
\end{minipage}%
\begin{minipage}[t]{0.5\linewidth}%
  \vspace{-1em}
  \[
    \getrule{Sbetafwd}
  \]
\end{minipage}
\medskip

\noindent
Since @Y:name@ and @Q:name -> proc@, we rely on the LF application
@(Q Y)@ to accomplish the object-level substitution $\subst{y}{x}{\fns{Q}}$.

We write congruence rules by requiring the inner process to step under some arbitrary @x:name@:
\begin{multicols}{2}
  \begin{lstlisting}
betacut1 : ({x:name} step ((P x) (P' x)))
      → step (pcomp A P Q) (pcomp A P' Q).
  \end{lstlisting}
  \columnbreak
  \[
    \getrule{Sbetacut1}
  \]
\end{multicols}
Principal rules, such as
\[
  \getrule{Sbetainl1}
\]
can be encoded straightforwardly:
\begin{lstlisting}
betainl : step (pcomp (A + B) (\x. inl x P) (\x. choice x Q R))
            (pcomp A P Q).
\end{lstlisting}
The names of the bound channels $x$ and $w$ are not explicit since the metavariables @P@, @Q@, and @R@ are all functions @name -> proc@ and can take an arbitrary name.

The remaining reduction rules and structural equivalences are similarly encoded.
Since there are no interesting cases to discuss, we defer the complete presentation to the included artifact.

\section{Adequacy of the Encoding}
We prove adequacy for each component of our encoding of SCP.
The proofs are quite tedious as is usual for these proofs, so we defer a more detailed overview of the proofs in an attached appendix.
In this section we focus on stating the right adequacy lemmas while also providing a high-level overview on the proof strategy for the more complex lemmas.

\subsection{Notation}
We use the sequent $\lf{\Gamma}{M : \tau}$ to refer to judgments within LF.
For instance, $\lftp{M}$ asserts that the LF term $M$ is of type ${\color{magenta}\mathrm{tp}}$ under no assumptions.
Similarly, $\lfwtp{\Gamma}{D}{P}$ asserts that the LF term $D$ is of type ${\color{magenta}\mathrm{wtp}} \; P$ where $P$ is some LF term of type ${\color{magenta}\mathrm{proc}}$.
Informally, $D$ in this context would correspond to a typing derivation.
We also work with LF \emph{canonical} forms, essentially the $\beta \eta$ normal forms of a given type, as is standard in adequacy statements.

\subsection{Session Types and Duality}
Adequacy for the encoding of session types can be shown with the obvious
translation function $\etp{-}$ that maps session types $A$ to LF terms $\etp{A}$ of type @tp@.
\begin{lemma}[Adequacy of {\color{magenta}tp}] \label{lem:adeq-tp}
  There exists a bijection between the set of session types and canonical LF terms $M$ such that $\lftp{M}$.
\end{lemma}

Adequacy of duality is also easy to show once stated properly.
Since there is a slight difference between the on-paper definition of duality as a unary function and the LF encoding of duality as a relation, we state adequacy for the encoding of duality as follows.
\begin{lemma}[Adequacy of {\color{magenta}dual}] \label{lem:adeq-dual}
  \leavevmode
  \begin{enumerate}
    \item For any session type $A$, there exists a unique LF canonical form $D$ such
      that $\lfdual{D}{\etp{A}}{\etp{A^\bot}}$.
    \item For any LF canonical form $D$ such that $\lfdual{D}{\etp{A}}{\etp{A'}}$, $\;A' = A^\bot$.
  \end{enumerate}
\end{lemma}

\subsection{Processes}
Adequacy of the process encoding also follows naturally from our encoding. In
particular, all channel bindings, which we encode as intuitionistic functions,
precisely match the process syntax of SCP. We can therefore define a translation
$\enc{-}$ from processes in SCP to LF normal forms and its decoding $\dec{-}$ in
the obvious manner.

\begin{definition}
  The encoding of name sets to an LF context is given as follows:
  \[
    \enc{x_1, \ldots, x_n} = \lfn{x_1}, \ldots, \lfn{x_n}
  \]
\end{definition}

\begin{lemma}[Adequacy of {\color{magenta}proc}] \label{lem:adeq-proc}
  For each SCP process \(\fns{P}\), there exists a unique canonical LF form \(\lfproc{\enc{\fn(\fns{P})}}{\enc{\fns{P}}}\) and \(\dec{\enc{\fns{P}}} = \fns{P}\).
  Conversely, if\/ \(\lfproc{\Gamma}{M}\) is a canonical LF form, then \(\dec{M}\) is an SCP process, \(\enc{\dec{M}} = M\), and \(\enc{\fn(\dec{M})} \subseteq \Gamma\).
\end{lemma}

The context $\enc{\fn(\fns{P})}$ captures the required assumptions to construct a LF term corresponding to a given process.
For example, an encoding of $\sfwd{x}{y}$ corresponds to the LF term \\${\lfproc{\lfn{x}, \lfn{y}}{\lffwd{x}{y}}}$.
Indeed, $\enc{\fn(\sfwd{x}{y})} = \lfn{x}, \lfn{y}$, allowing the @fwd@ constructor to be applied with the assumptions $\lfn{x}$ and $\lfn{y}$.

Unfortunately, we cannot give a clean bijection result due to weakening in LF derivations.
For example, there is a derivation of $\lfproc{\Gamma, \lfn{x}, \lfn{y}}{\lffwd{x}{y}}$ for any $\Gamma$, and such derivations all correspond to the SCP process $\sfwd{x}{y}$.
Therefore, we only require that the overall context include the free names for the converse direction.
This weaker statement does not affect later developments since weakening in LF does not change the structure of the derivation.
This phenomenon repeats for later adequacy results due to weakening.

\subsection{Linearity}
We define an encoding $\elin{-}$ that maps derivations of linearity predicates in SCP of form $\lin{x}{\fns{P}}$ to LF canonical forms of type
$\lfL{x}{\eproc{\fns{P}}}$. Similarly, we define a decoding $\dlin{-}$ that maps LF canonical forms of type $\lfLo{M}$, where $M$ is of type @name -> proc@, to
derivations of $\lin{x}{\dproc{M\ x}}$.
\begin{lemma}[Adequacy of {\color{magenta}linear}] \label{lem:adeq-linear}
  For each derivation $\mc{D}$ of \/ $\lin{x}{\fns{P}}$, there exists a unique canonical LF term $L = \elin{\mc{D}}$ such that
  $\lflind{\enc{\fn(\fns{P}) \setminus x}}{L}{x}{\eproc{\fns{P}}}$ and $\dlin{L} = \mc{D}$.
  Conversely, if \/ $\lflinm{\Gamma}{L}{M}$ is a canonical LF form, then $\dlin{L}$ is a derivation of \/ $\lin{x}{\dproc{M\ x}}$
  and $\lflinm{\enc{\fn(\dproc{M\ x}) \setminus x}}{\elin{\dlin{L}}}{M}$ where $\enc{\fn(\dproc{M\ x})} \subseteq \Gamma$.
\end{lemma}

Here, the encoding of the context is slightly tricky because we define the linearity predicate on paper using the syntax $\lin{x}{\fns{P}}$, meaning ${x \in \fn(\fns{P})}$.
In LF however, since we encode the linearity predicate
@linear: (name → proc) → type@ over intuitionistic functions taking some name
$x$, we must use the context $\enc{\fn(\fns{P}) \setminus x}$ when encoding an on-paper derivation of some linearity predicate.
More informally, we establish a correspondence between derivations of
$\lin{x}{\fns{P}}$ and LF canonical forms of $\lfL{x}{\eproc{\fns{P}}}$ under an LF
context \emph{without} the assumption $\lfn{x}$.

At a high level, the proof of this lemma mostly involves ensuring that the various $x \notin \fn(\fns{P})$ conditions are fulfilled by our higher-order encoding and vice versa.
For example, the encoding of
\[
  \getrule{Linl}
\]
is @l_inl: linear M -> linear (\x.inl x M)@, and in particular, @M@ is a metavariable, meaning it cannot depend on the internally bound $x$, satisfying the side condition of $\getrh{Linl}{2}$.

\subsection{Type Judgments}
To establish a relation between SCP type judgments $\stp{\Gamma}{\fns{P}}$ and LF derivations of $\mathtt{wtp }\enc{\fns{P}}$, we must define a context mapping of typing assumptions $\Gamma = x_1:A_1, \ldots, x_n:A_n$.
\begin{definition}
  A context encoding $\enc{\Gamma}$ is defined by introducing LF assumptions $\lfb{h}{x}{A}$ for each typing assumption in $\Gamma$:
  \[
    \enc{x_1:A_1, \ldots, x_n:A_n} = \lfb{h_1}{x_1}{A_1}, \ldots, \lfb{h_n}{x_n}{A_n}
  \]
\end{definition}
We define an encoding $\ewtp{-}$ and decoding $\dwtp{-}$ of type derivations in our adequacy statement.
\begin{lemma}[Adequacy of {\color{magenta}wtp}] \label{lem:adeq-wtp}
  There exists a bijection between typing derivations in SCP of form
  $\stp{\Gamma}{\fns{P}}$ and LF canonical forms $D$ such that
  $\lfwtp{\enc{\Gamma}}{D}{\enc{\fns{P}}}$
\end{lemma}
The proof mostly involves appealing to previous adequacy lemmas and is otherwise fairly straightforward.
In fact, the proof for the linearity predicate is more involved due to the implicit implementation of the free name side-conditions using higher-order encoding.
This is not too surprising: the design of SCP was heavily motivated by a desire for a system more amenable to mechanization in LF.
Furthermore, we have a bijection for type judgments because type judgments in SCP also have weakening, making the adequacy statement very clean.

\subsection{Reductions and Structural Equivalences}
Adequacy of reductions is easy to show; most rules are axioms, so we simply appeal to the adequacy of the underlying processes.
The congruence cases are very simple and follows from the appropriate induction hypotheses.
Adequacy of structural equivalence is similarly easy to show.

The adequacy statements are unfortunately slightly cumbersome for the same reason as \Cref{lem:adeq-proc} and \Cref{lem:adeq-linear} since weakening in LF does not allow for a clean bijection.
Again, we want to emphasize that this does not change the structure of the derivations of both @step@ and @equiv@.

\begin{lemma}[Adequacy of {\color{magenta}step}] \label{lem:adeq-step}
  For each SCP reduction $S$ of \/ $\fns{P} \sst \fns{Q}$, there exists a unique canonical LF derivation
  $\lfstep{\enc{\fn(\fns{P})}}{\enc{S}}{\enc{\fns{P}}}{\enc{\fns{Q}}}$ and $\dec{\enc{S}} = S$.
  Conversely, if\/ \({\lfstep{\Gamma}{D}{M}{N}}\) is a canonical LF form, then \(\dec{D}\) is a derivation of
  a reduction $\dproc{M} \sst \dproc{N}$, \(\enc{\dec{D}} = D\), and \(\enc{\fn(\dec{M})} \subseteq \Gamma\).
\end{lemma}
\begin{lemma}[Adequacy of {\color{magenta}equiv}] \label{lem:adeq-equiv}
  For each SCP structural equivalence $S$ of \/ $\fns{P} \equiv \fns{Q}$, there exists a unique canonical LF derivation
  $\lfequiv{\enc{\fn(\fns{P})}}{\enc{S}}{\enc{\fns{P}}}{\enc{\fns{Q}}}$ and $\dec{\enc{S}} = S$.
  Conversely, if\/ \(\lfequiv{\Gamma}{D}{M}{N}\) is a canonical LF derivation, then \(\dec{D}\) is a derivation of
  a structural equivalence $\dproc{M} \equiv \dproc{N}$, \(\enc{\dec{D}} = D\), and \(\enc{\fn(\dec{M})} \subseteq \Gamma\).
\end{lemma}

\subsection{Adequacy with respect to CP}
Since we establish a bijection between SCP and our encoding and there exists a
bijection between CP and SCP when restricted to well-typed and linear processes,
we also conclude that our encoding is adequate with respect to CP when
restricted to well-typed and linear processes (in the encoding).

\newcommand{\efull}[1]{\varepsilon_\circ(#1)}
\newcommand{\dfull}[1]{\delta_\circ(#1)}
\begin{definition}
  An encoding map $\varepsilon_\circ$ of processes and typing derivations in CP to LF
  is defined by the composition of the encoding $\varepsilon$ of CP to SCP with the encoding $\enc{-}$ of SCP to LF, \ie, $\varepsilon_\circ = \enc{\varepsilon(-)}$.
  Similarly, a decoding map $\delta_\circ$ of processes and typing derivation in LF to CP is defined by the composition of the decoding $\dec{-}$ of LF to SCP with the decoding $\delta$ of SCP to CP,
  \ie, $\delta_\circ = \delta(\dec{-})$.
\end{definition}

\begin{corollary}
  The encoding function \(\varepsilon_\circ\) is left inverse to \(\delta_\circ\) and
  \begin{enumerate}
    \item If \(\mc{D}\) is a derivation of\/ \(\tp{\Delta}{P}\) where $\Delta = \h{x_1}{A_1}, \ldots, \h{x_n}{A_n}$,
      then there exists a collection of LF canonical forms $\{W, L_1, \ldots, L_n\}$ such that
      \begin{itemize}
        \item $W = \efull{\mc{D}}$ such that $\lfwtp{\enc{\Delta}}{W}{\efull{P}}$
        \item $\lflind{\enc{\fn(P) \setminus x_i}}{L_i}{x_i}{\efull{P}}$ for $1 \leq i \leq n$
        \item $\dfull{\efull{\mc{D}}} = \mc{D}$
      \end{itemize}
    \item If \(\{W, L_1, \ldots, L_n\}\) is a collection of LF derivations such that
      \begin{itemize}
        \item $\lfwtp{\Gamma}{W}{M}$ where $\Gamma = \{\lfb{h_1}{x_1}{A_1}, \ldots, \lfb{h_n}{x_n}{A_n}\}$
        \item $\lflind{\Gamma \setminus \{\lfn{x_i}, \lfh{h_i}{x_i}{A_i}\}}{L_i}{x_i}{M}$ for $1 \leq i \leq n$
      \end{itemize}
      then there exists a derivation $\dfull{W}$ of $\tp{\Delta}{\dfull{M}}$ and $\efull{\dfull{M}} = M$ such that $\Gamma = \enc{\Delta}$.
  \end{enumerate}
\end{corollary}

\section{Mechanizing the Type Preservation Proof}
\label{sec:mechanization}
In the previous sections, we focused our attention to the encoding of SCP and its adequacy, which were purely done in the logical framework LF.
Now, we give a brief overview of our mechanization of type preservation in the proof assistant Beluga.
Mechanizations in Beluga involve encoding the syntax and semantics of
the object language in the \emph{LF Layer} and then manipulating LF
terms in the \emph{Computational Layer} using contextual types to
characterize derivation trees together with the context in which they
make sense \cite{Nanevski08tocl,Pientka:POPL08,Pientka:PPDP08,Cave:POPL12}.
The contextual types enable clean statements of various strengthening statements, which comprise the majority of the lemmas used in the type preservation proof.

Since the computational layer in Beluga is effectively a functional programming language, inductive proofs of metatheorems are (terminating) recursive functions that manipulate LF objects.
For presentation purposes, we assume no familiarity with the
computational layer of Beluga and explain the lemmas and theorems informally in words.
We defer to the accompanying artifact for the implementation details
of all the lemmas and theorems below.

\subsection{Lemmas of ${\color{magenta}\mathrm{dual}}$}
Due to our encoding of duality as a relation between two types, we must prove symmetry and uniqueness.
The encoding of symmetry is a recursive function \lstinline!dual_sym! that takes as input a closed LF object of type $\lfDual{A}{A'}$ and outputs a closed LF object of type $\lfDual{A'}{A}$.
The encoding of uniqueness takes two closed LF objects of type $\lfDual{A}{A'}$ and $\lfDual{A}{A''}$ and outputs a proof that $A' = A''$.
To encode the equality of session types $A' = A''$, we follow the
standard technique of defining an equality predicate
\lstinline!eq:  tp → tp → type! over session types with reflexivity as its constructor.
\begin{lstlisting}
% Symmetricity and Uniqueness
rec dual_sym : [ ⊢ dual A A' ] → [ ⊢ dual A' A] =
/ total 1 /
fn d ⇒
case d of
| [ ⊢ D1] ⇒ [ ⊢ D⊥]
| [ ⊢ D⊗ Dl Dr] ⇒
  let [ ⊢ l] = dual_sym [ ⊢ Dl] in
  let [ ⊢ r] = dual_sym [ ⊢ Dr] in
  [ ⊢ D⅋ l r]
| ...

rec dual_uniq : [ ⊢ dual A A' ] → [ ⊢ dual A A''] → [ ⊢ eq A' A''] = ...
\end{lstlisting}
The use of the contextual box with no assumptions @[ ⊢ ...]@ captures
closed objects. The contextual variables (or metavariables) \lstinline!A! and \lstinline!A'!
are implicitly quantified at the outside.
The implementations of the two functions pattern match on the
input with appropriate recursive calls for the binary type
constructors, corresponding to the usual induction proofs for these
lemmas. We show only one base case and one recursive case to give the
flavour of how proofs are written as recursive programs. The totality
annotation checks that the program is covering and that all recursive calls
on the first (explicit) argument are structurally smaller and decreasing.

\subsection{Strengthening Lemmas}
Next, we encode strengthening lemmas for contextual LF terms of various types.
First, we present them informally below using LF-like syntax, using $\vdash$ instead of $\vdash_{LF}$ and omitting LF term names for economical purposes:
\begin{lemma}[Strengthening Lemmas]
  \leavevmode
  \begin{enumerate}
    \item If \/ @Gamma,z:name,h:hyp z C |- hyp X A@ and $z \neq X$, then
      @Gamma |- hyp X A@.
    \item If \/ @Delta,z:name |- linear \x.P@ and $z \notin \fn(P)$, then
      @Delta |- linear \x.P@.
    \item If \/ @Gamma,z:name,h:hyp z C |- wtp P@ and $z \notin \fn(P)$, then
      @Gamma |- wtp P@.
    \item If \/ @Delta,z:name |- step P Q@ and $z \notin \fn(P)$, then $z
      \notin \fn(Q)$ and \/ @Delta |- step P Q@.
    \item If \/ @Delta,z:name |- equiv P Q@ and $z \notin \fn(P)$, then $z
      \notin \fn(Q)$ and \/ @Delta |- equiv P Q@.
  \end{enumerate}
  where @Gamma@ consists of assumptions of form @x1:name,h1:hyp x1 A1,..,xn:name,hn:hyp xn An@ and @Delta@ consists of assumptions of form @x1:name,..,xn:name@.
\end{lemma}

The use of different contexts @Gamma@ and @Delta@ in these statements mostly indicate the spirit of the judgments that we strengthen.
Linearity for instance should not depend on typing assumptions, so we use @Delta@.
In practice, picking the right kind of context to use proved immensely useful in simplifying the final type preservation proof.
In particular, we found that it is more convenient to weaken the final two lemmas regarding @step@ and @equiv@ by stating them under the richer context @Gamma@.

To encode @Delta@ and @Gamma@ in Beluga, we first define \textit{context schemas}.
In our case, we are interested in contexts containing assumptions of names, \ie, @Delta@, and assumptions of names alongside their types for the typing judgments, \ie, @Gamma@:
\begin{lstlisting}
schema nctx = name;
schema ctx = some [A:tp] block x:name, h:hyp x A;
\end{lstlisting}

In the statement of our lemma, we exploit the full power of contextual
variables to cleanly state the strengthening lemmas.
For instance, we encode the side-condition that $z \neq X$ in the strengthening of @hyp X A@ by requiring that @X@ does not depend on @z@:
\begin{lstlisting}
rec str_hyp : (Gamma:ctx) [Gamma, z:name, h:hyp z C[] ⊢ hyp X[..] A[]] → [Gamma ⊢ hyp X A[]] = ...
\end{lstlisting}
We first implicitly abstract over the context \lstinline!Gamma!
specifying what kind of context we are working in.
Further, contextual variables such as @X@ or @A@ are associated with a
substitution. By default, they are associated with the identity
substitution which can be omitted by the user. However, Beluga also allows
us to associate contextual variables with more interesting substitutions.
The weakening substitution on the name @X[..]@ ensures that @X@ only depends on
@Gamma@ and not @z@ or @h@, which indeed captures the requirement $z
\neq X$. The empty substitutions on the session types @A[]@ and @C[]@ indicate
that they do not depend on anything, \ie, they are closed.
We encode the requirement that $z \notin \fn(P)$ in the strengthening lemmas for linearity and typing using a similar technique:
\begin{lstlisting}
rec str_lin : (Delta:nctx) [Delta, z:name ⊢ linear \y. P[.., y]] → [Delta ⊢ linear \y. P] = ...
rec str_wtp : (Gamma:ctx) [Gamma, z:name, h:hyp z C[] ⊢ wtp P[..]] → [Gamma ⊢ wtp P] = ...
\end{lstlisting}
The substitutions associated with the variable @P@ in @P[.., y]@
and @P[..]@ encode that the process @P@ does not depend on the
assumption @z@ that we want to strengthen out, properly capturing the
side-condition of $z \notin \fn(P)$ in both lemmas. Indeed, @str_wtp@
turns out to be a mechanization of \Cref{lem:scp-str}. The proofs of
these lemmas are straightforward and are given by pattern matching on
the input.

The final two strengthening lemmas are a bit different because of the additional free-name condition in the conclusions.
Suppose we naively follow the prior attempts:
\begin{lstlisting}
rec str_step : (Gamma : ctx) [Gamma, x:name ⊢ step P[..] Q] → [Gamma ⊢ step P Q] = ...
\end{lstlisting}
Unfortunately, the conclusion @Gamma |- step P Q@ is not well-typed since @Q@ as used in the premise depends on @Gamma, x:name@ whereas @Q@ as used in the conclusion only depends on @Gamma@.
If we change the premise to @[Gamma, x:name ⊢ step P[..] Q[..]@ to require that @Q@ only depends on @Gamma@, then the lemma is not strong enough.
Indeed, encoding the strengthening lemma actually requires an existential; we must say that there exists some process @Q'@ such that @Gamma |- step P Q'@ and @Q = Q'@.
However, since LF does not have sigma types, we must further encode this existential using a data structure @Result@, whose only constructor takes the process @Q'@, a proof that @Q = Q'@, and a proof that @step P Q'@.
As before, we define equality of processes @eq_proc@ as a relation with only the reflexivity constructor.

\begin{lstlisting}
inductive Result : (Gamma : ctx){P : [Gamma ⊢ proc]}{Q : [Gamma, x:name ⊢ proc]} → ctype =
| Res : {Q' : [Gamma ⊢ proc]}
      → [Gamma, x:name ⊢ eq_proc Q Q'[..]]
      → [Gamma ⊢ step P Q']
      → Result [Gamma ⊢ P] [Gamma, x:name ⊢ Q];
\end{lstlisting}
We can now state the lemma using this data structure:
\begin{lstlisting}
rec str_step : (Gamma : ctx) [Gamma, x:name ⊢ step P[..] Q] → Result [Gamma ⊢ P] [Gamma, x:name ⊢ Q] = ...
\end{lstlisting}
We follow an analogous procedure for strengthening structural equivalences and prove the two lemmas simultaneously via mutual recursion.

\subsection{Auxiliary Lemmas}
We prove two additional lemmas to aid in the type preservation proof.
The first lemma states that $\lin{x}{\fns{P}}$ implies $x \in \fn(\fns{P})$.
We however work with its contrapositive since we do not directly encode $\fn(\fns{P})$.
\begin{lemma}[Linearity requires usage]
  If $x \notin \fn(\fns{P})$, then @Gamma |- linear (\x.P)@ is not derivable.
\end{lemma}
We encode the contradiction in the lemma using the standard LF technique of defining a type @imposs@ without any constructors.
The encoding of the lemma is therefore a function that takes as input @[Delta ⊢ linear (\x. P[..])]@ and outputs some @imposs@.
The substitution @P[..]@ indicates that the process does not depend on the input name $x$ which properly captures the premise $x \notin \fn(\fns{P})$.
\begin{lstlisting}
imposs : type.
% no constructor for imposs
rec lin_name_must_appear : (Delta : nctx) [Delta ⊢ linear (\x. P[..])] → [ ⊢ imposs] = ...
\end{lstlisting}

Next, we show that structural equivalence preserves both linearity and typing.
To state preservation for linearity, we have to reconcile the fact that linearity is defined parametric to some channel name, so we must extend the context of @equiv@ with an additional name.
\begin{lemma}[Structural Equivalence preserves linearity and typing]
  \label{lem:s_equiv}
  \leavevmode
  \begin{enumerate}
    \item If \/ @Gamma,x:name |- equiv P Q@ and @Gamma |- linear \x.P@,
      then @Gamma |- linear \x.Q@.
    \item If \/ @Gamma |- equiv P Q@ and @Gamma |- wtp P@,
      then @Gamma |- wtp Q@.
  \end{enumerate}
\end{lemma}
Although the first lemma can in spirit be stated under a context of names @Delta@, we used the more general context of names and types @Gamma@ to better suit our type preservation proof.
\begin{lstlisting}
rec lin_s_equiv : (Gamma : ctx) [Gamma, x:name ⊢ equiv P Q]
                            → [Gamma ⊢ linear (\x. P)]
                            → [Gamma ⊢ linear (\x. Q)] = ...
rec wtp_s_equiv : (Gamma : ctx) [Gamma ⊢ equiv P Q]
                            → [Gamma ⊢ wtp P]
                            → [Gamma ⊢ wtp Q] = ...
\end{lstlisting}
Note that our proof shows that linearity is preserved for any given (free) channel $x$, meaning that the on-paper predicate $\lin{\Delta}{\fns{P}}$ is also preserved by structural equivalence.

\subsection{Type Preservation}
Finally, we are ready to state the main theorem.
To state preservation of linearity, we extend the contexts of other judgments appropriately in the same manner as for @equiv@.
\begin{theorem}[Type Preservation] \label{thm:lf-sr}
  \leavevmode
  \begin{enumerate}
    \item If \/ @Gamma,x:name |- step P Q@ and
      @Gamma,x:name,h:hyp x A |- wtp P@ and\\
      @Gamma |- linear \x.P@, then
      @Gamma |- linear \x.Q@.
    \item If \/ @Gamma |- step P Q@ and @Gamma |- wtp P@,
      then @Gamma |- wtp Q@.
  \end{enumerate}
\end{theorem}

\noindent
The encodings for these statements are very similar to the encodings for \Cref{lem:s_equiv}:

\begin{lstlisting}
rec lin_s : (Gamma : ctx) [Gamma, x:name, h:hyp x A[] ⊢ wtp P[..,x]]
                     → [Gamma, x:name ⊢ step P Q]
                     → [Gamma ⊢ linear (\x. P)]
                     → [Gamma ⊢ linear (\x. Q)] = ...
and rec wtp_s : (Gamma : ctx) [Gamma ⊢ wtp P]
                     → [Gamma ⊢ step P Q]
                     → [Gamma ⊢ wtp Q] = ...
\end{lstlisting}
The implementations for both functions proceed by case analysis on the term of type \\@[Gamma, x:name |- step P Q]@.
Preservation of linearity is perhaps the more interesting part of this theorem.
For instance, consider the case $\getrn{Sbetainl1}$:
\[
  \getrc{Sbetainl1}
\]
To show that linearity of some free channel $z$ is preserved under this reduction, we must check for the case where $z$ appears in the left process or in the right process by pattern matching on the linearity assumption.
\begin{lstlisting}
rec lin_s : (Gamma : ctx) [Gamma, x:name, h:hyp x A[] ⊢ wtp P[..,x] ]
                     → [Gamma, x:name ⊢ step P Q]
...
=
/ total 2 /
fn tpP ⇒ fn sPQ ⇒ fn linP ⇒
case sPQ of
...
| [g, z:name ⊢ betainl1] ⇒
  (case linP of
% z appears on the left -- the linearity must be the congruence case for inl
  | [g ⊢ l_pcomp1 (\x. l_inl2 (\w.linP'))] ⇒
    [g ⊢ l_pcomp1 (\w. linP'[..,w,w])]
% z appears on the right -- the linearity must be the congruence case for the 'case' construct
  | [g ⊢ l_pcomp2 (\x. l_choice2 (\w. linP') (\w._))] ⇒
    [g ⊢ l_pcomp2 (\w. linP'[..,w,w])]
  )
\end{lstlisting}
The first @w@ in the substitution @linP'[..,w,w]@ correspond to substituting $w$ for $x$, which may seem like a violation of linearity.
However, for well-typed processes, the linearity predicate for $x$ will ensure that $x$ is no longer used in the inner process, meaning this substitution does not lead to duplication of $w$ and is safe.

The implementation for @wtp_s@ is mostly bureaucratic and involves using many of the prior strengthening lemmas to ensure that the communicated channel $x$ can be safely removed from the context.

One interesting observation is that although preservation of typing does not require
any assumptions about linearity, preservation of linearity does require the
assumption that the original process is well-typed. This is primarily due to the
reduction rule $\getrn{Sbetafwd}$:
\[
  \getrc{Sbetafwd}
\]
Here, if we want to show that the linearity of channel $y$ is preserved, we need to know that $\fns{Q}$ treats $x$ linearly, or $\lin{x}{\fns{Q}}$.
We can only obtain this from the assumption that the original process is
well-typed since $x$ in process $\fns{Q}$ is not a continuation channel of $y$
in $\fns{P}$.

\section{Related Work} \label{sec:related-works}
The linearity predicate that we develop in this paper is based on Crary's mechanization of the linear $\lambda$-calculus in Twelf~\cite{Crary10icfp}.
Adapting his ideas to the session-typed setting was non-trivial due to the many differences between the two systems, such as channel mobility, the distinction between names and processes, and continuation channels.
Our bijection proof between CP and SCP is similar to Crary's adequacy proof of his encoding,
where he showed that typing derivations of linear $\lambda$-calculus expressions were in bijection with typing derivations in the encoding alongside a proof of linearity for each free variable.
Indeed, this side condition is analogous to our criterion that $\lin{\Delta}{\fns{P}}$.

\subsection{HOAS Mechanizations}
R{\"o}ckl, Hirschkoff, and Berghofer~\cite{Rockl01fossacs} encode the untyped $\pi$-calculus in Isabelle/HOL and prove that their encoding is adequate.
Much of their technical development concerns eliminating \textit{exotic terms}.
To do so, they introduce local well-formedness conditions, similar in spirit to how we use the linearity predicates to eliminate non-linear processes.
In LF, such exotic terms do not typically arise, as there is a
bijection between the canonical representation in LF and its on-paper
counterpart.
Moreover, they do not encode any process reductions or mechanize any metatheorems.

\citet{Despeyroux00tcs} gives a HOAS encoding of a typed $\pi$-calculus in Coq and uses it to mechanize a proof of subject reduction.
This encoding is less involved than ours because their type system is very simple and, in particular, does not involve linearity.
Thus, they did not need to account for complex operations on contexts.
Furthermore, they do not discuss the adequacy of the encoding.

\citet{Tiu:TOCL10} give a weak HOAS encoding of the finite
$\pi$-calculus together with its operational semantics using the late
transition system within a logic that contains the $\nabla$ quantifier
for encoding generic judgments and definitions.
They then specify a bisimulation for late transition systems
and show that it is reflexive and transitive.
Tiu and Miller prove that their encoding is adequate.
However, their system does need to deal with linearity
and is also not typed and hence does not face the same challenges as ours.

The closest existing literature to our work is by \citet{ZalakainMS19}, who uses parametric HOAS~\cite{Chlipala:ICFP08} to mechanize a session-typed process calculus in Coq.
They use a global linearity predicate as a well-formedness condition and directly encode the $x \notin \fn(P)$ style side conditions as a predicate.
They further prove that linearity is preserved under all reductions except those using the structural equivalence $P \mid Q \equiv Q \mid P$, which corresponds to \getrn{Sequivcomm} in our setting.
This equivalence is problematic in their setting because of interactions between their linearity predicate, scope expansion, and parallel composition.
They do not discuss the adequacy of their encoding. We instead localize the linearity predicates within type judgments and leverage higher-order encoding to obtain some side conditions ``for free''.
As in their setting, we prove subjection reduction for linearity but also for typing, obtaining the usual type preservation result.
Furthermore, the structural equivalence rule $\getrc{Sequivcomm}$ presents no notable difficulties in our setting.

\subsection{Other Approaches to Mechanizing Session Types and Typed
Process Calculi}
\citet{Gay01tphols} uses Isabelle/HOL to give one of the first mechanizations of a linearly typed process calculus and its reduction relation.
Bindings are handled via de Bruijn indexing and linearity is enforced by modeling a linear context with relevant operations.
Interestingly, he does not directly encode processes in Isabelle/HOL.
Instead, he mechanizes a \(\lambda\)-calculus with constants as a
metalanguage and then encodes channel bindings in the process calculus
through \(\lambda\)-abstractions in the metalanguage in  a HOAS-like manner.

\citet{Thiemann19ppdp} mechanizes a functional language with session-typed communication in Agda.
He too uses de Bruijn indexing to handle binding and directly implements linear contexts.
The system is intrinsically typed, meaning subject reduction is obtained ``for free''.
However, the encoding is operational in nature, and for example, the operational semantics depends on a ``scheduler'' that globally identifies channels and performs communication.
Showing adequacy of the encoding is therefore quite complicated because of the disconnect between the on-paper theory and the actual implementation, which the author mentions.

Zalakain and Dardha model contexts using leftover typing in Agda~\cite{Zalakain21forte}.
This technique avoids context splits by modifying type judgments to add an additional output context, making explicit what resources are not used by a given process in a type judgment.
However, their approach still requires proving certain metatheorems about their leftover typing and still embeds some form of linearity.
It is therefore not well-suited for a HOAS-style encoding in LF, although it is less clear what are the trade-offs between their approach and our approach in non-HOAS settings.
They also make no mention of adequacy.

Castro-Perez, Ferreira, and Yoshida~\cite{Castro-Perez20tacas} use a locally nameless representation to develop a general framework of mechanizing session-typed process calculi in Coq.
They observe that a naïve usage of locally nameless representations cannot handle higher-order communication, \ie, channel transmission.
To encode such communications, they employ a strategy to syntactically distinguish between different forms of channel bindings, working with four sets of channel names.
Our approach encodes all forms of channel bindings via intuitionistic functions over the same set of names in LF and handles higher-order communication.

\subsection{HOAS with Linearity}
Perhaps one natural approach to a HOAS encoding of a linear system like session types is to use a logical framework with direct support for linear implications.
Unfortunately, these systems are far less understood, and implementations of such systems are often preliminary.

Concurrent LF ~\cite{Schack-Nielsen:IJCAR08} is an extension of the
logical framework LF to support the specification of linear and even
concurrent formal systems.
Its implementation, Celf, has been used to encode systems such as the untyped $\pi$-calculus~\cite{Cervesato02tr}.
Although encoding a session-typed system certainly seems plausible in Celf, it remains unclear how to encode metatheoretic proofs such as subject reduction.

LINCX~\cite{Georges:ESOP17} is a proof environment that follows in the
footsteps of Beluga.
Instead of specifying formal systems in LF as in Beluga,
one specifies formal systems in linear LF in LINCX.
Metatheoretic proofs are then implemented as recursive functions over linear
contextual objects. This framework should in principle be capable of
representing session-type systems and their metatheory more
directly, but there is presently no implementation for it.

Linear Hybrid ~\cite{Felty19,Felty:MSCS21} is designed
to support the use of higher-order abstract syntax for representing
and reasoning about formal systems, and it is implemented in the Coq Proof
Assistant. To support representation of linear systems it implements a
linear specification logic in Coq. Felty and collaborators have used
this framework to, for example, encode the type system of a quantum
$\lambda$-calculus with linear typing and its
metatheoretic properties.
It would be interesting to see how to use this framework to specify
session types together with their metatheory.

\section{Conclusion}
We demonstrate a higher-order encoding and mechanization of CP, a session-typed process calculus.
Our main technique is using linearity predicates that act as well-formedness conditions on processes.
In particular, this lets us encode linearity without relying on linear contexts which are difficult to work with in mechanizations and which are not well-suited for HOAS-style encodings.
We decomposed our encoding in two steps: an on-paper formulation of SCP using linearity predicates, and a mechanization of SCP in Beluga.

Our development of SCP, which arose as a byproduct of our mechanization, provides a foundation for mechanizing session-typed process calculi in settings with structural contexts.
We prove that CP is fully embedded in SCP and furthermore, that the restriction imposed by the linearity predicates captures the fragment of SCP that correspond to CP.
More precisely, we prove that there is a structure-preserving bijection between the processes and typing derivations in CP and those in SCP when we subject SCP to the condition that it treats its free names linearly.

We then mechanize SCP in Beluga and prove the adequacy of our encoding, thereby showing that our encoding is adequate with respect to CP. As we demonstrate through our mechanization, SCP particularly synergizes with a HOAS encoding over Beluga, which utilizes contextual type theory, allowing for side-conditions related to free names to be encoded ``for free''.

In general however, using an SCP-like presentation has the benefit of using intuitionistic contexts, which are better understood and easier to work with in proof assistants.
Whether the encoding style implicitly uses an intuitionistic context like for LF is not particularly important; even an encoding style that explicitly models a context can benefit from this approach.
Our development of SCP shows how to shift the work required for linear context management to local side conditions, or linearity predicates, which we believe leads to a more tractable way to both encode and reason with linearity.
Although our approach is certainly heavily inspired by the constraints imposed by LF and HOAS, SCP is still a promising system to mechanize over CP using other proof assistants and encoding styles such as de Bruijn or locally nameless.
In particular, Zalakain's encoding~\cite{ZalakainMS19} of a similar session-typed system using parametric HOAS gives strong evidence that an SCP-style calculus extends well to Coq.

It is however important to acknowledge that this approach comes at the cost of managing linearity predicates and free names in processes.
Although these were easy to work with in our setting (in particular, managing free names was obtained for free from higher-order unification),
it would be interesting to understand more clearly the costs and benefits from the additional side conditions compared to dealing with linear contexts in the context of other proof assistants and encoding styles.

\subsection{Towards more complex language constructs}

We illustrated how linearity predicates could be used to mechanize a fragment of Wadler's CP~\cite{Wadler12icfp}, and it is natural to ask whether this technique scales to the full system.
It is also natural to ask whether this technique scales to more complex extensions of session-typed systems, such as notions of sharing~\cite{Balzer17icfp, Rocha21icfp}, equi-recursion~\cite{Gay05acta}, and integrations with functional languages~\cite{Gay10jfp, Toninho13esop}.
We believe that linearity predicates are a mechanization technique that is sufficiently robust and scalable to handle these richer language constructs.
To guide future applications of our approach, we sketch the key patterns and principles for its application to new program constructs:
\begin{enumerate}
  \item Determine if the construct binds any new linear channels. If so, then its typing judgments must check their linearity. In our development, this is illustrated by the typing rules \getrn{Sparr}, \getrn{Sotimes}, and \getrn{Scut}.
  \item Determine if the construct requires the absence of other linear assumptions. If so, then there should be no congruence rules for the linearity predicate. In our development, this is illustrated by the linearity predicates for $\close{x}$ and $\fwd{x}{y}$.
  \item Determine if the construct uses a continuation channel. If so, then the linearity predicate should check that the continuation channel is used linearly. Otherwise, the linearity predicate should be an axiom. These two cases are respectively illustrated by \getrn{Linl} and \getrn{Lwait}.
  \item Determine if linear channels are shared between subterms composed by the construct. If they are not shared, then the linearity predicate must ensure that no sharing occurs. This is illustrated by \getrn{Lpcomp1} and \getrn{Lpcomp2}.
\end{enumerate}

With regard to extending our mechanization to the entirety of CP, we believe that its polymorphic constructors $\forall$ and $\exists$ will pose no technical challenges.
Indeed, they operationally correspond to receiving and sending types, and types are treated in an unrestricted manner.
Therefore, they do not interact with linearity in an interesting way.

However, the exponentials $!$ and $?$ may be more challenging to mechanize.
Channels of type $?A$ are not treated linearly: they may be dropped or copied.
Intuitively, this means that we should \textit{not} check for linearity of channels of type $?A$.
In Crary's encoding of the linear $\lambda$-calculus, there was only one syntactical construct that bound assumptions of type $?\tau$, making this easy to do.
In contrast, CP channels of type $?A$ can arise from many sources, such as inputs from channels of form $(?A) \parr B$, as channel continuations of any connective such as ${?A} \oplus {?B}$.
This means that we cannot determine solely from the syntax of processes whether a bound channel is of type $?A$.
However, we only ever use the linearity predicate to check the linearity of channels whose type is known.
We believe that by using this type information and by making the linearity predicate type aware, \ie, of the form $\lin{\h{x}{A}}{\fns{P}}$, we can give a sufficiently refined analysis of linearity to support channels of type \(?A\).

\subsection{Future Work}

Our work lays the groundwork for two main directions of future work.
The first is to explore the trade-offs encountered when encoding SCP in various proof assistants and mechanization styles.
Given that SCP was designed with an LF encoding in mind, it is not entirely clear whether the overhead of linearity predicates and free name conditions is offset by the advantages of working with unrestricted contexts in other settings.
Nevertheless, we believe that SCP provides a scalable basis for mechanizations with proofs of adequacy in mind.

The second direction is to extend SCP and its encoding to better understand the scalability of our technique.
Although we sketched the general roadmap for such extensions, it is interesting to verify that our technique is indeed scalable and to also understand its limitations.
Mechanizing metatheory beyond subject reduction will further elucidate our technique's scalability.
For example, we believe that our linearity predicate will be essential to mechanizing a progress theorem for SCP processes.
Progress for SCP processes corresponds to top-level cut elimination.
Well-typed linear SCP processes support top-level cut elimination by their correspondence with CP processes (\Cref{thm:cp-scp}) and the fact that CP processes enjoy this same property.
This indirect proof sketch is similar to our indirect proof of subject reduction (\Cref{thm:scp-sr}).
A direct proof of progress is a natural next metatheorem to mechanize and, based on our preliminary investigations, seems to be relatively straightforward.

\section*{Data-Availability Statement}
The software containing the encoding of SCP (\Cref{sec:encoding}) and mechanization of the subject reduction proof (\Cref{sec:mechanization}) is available on Zenodo~\cite{Sano23oopsla-art}.

\begin{acks}
This work was funded by the Natural Sciences and Engineering Research
Council of Canada (grant number 206263), Fonds de recherche
du Qu{\'e}bec - Nature et Technologies (grant number 253521), a
Tomlinson Doctoral Fellowship awarded to the first author, and Postdoctoral
Fellowship from  Natural Sciences and Engineering Research
Council of Canada awarded to the second author.

 We also thank the anonymous reviewers for their valuable comments and feedback.
\end{acks}

\bibliography{cites}

\newpage
\appendix

\input{adequacy-proof}

\newpage
\input{adequacy-scp-encoding}

\end{document}

%% file: rules.tex
\defrule{Sid}{\srn{Id}}{
  \stp{\Gamma, x:A, y:A^\bot}{\sfwd{x}{y}}
}{
}

\defrule{Scut}{\srn{Cut}}{
  \stp{\Gamma}{\spcomp{x}{A}{\fns{P}}{\fns{Q}}}
}{
  \stp{\Gamma, x:A}{\fns{P}}
  &
  \lin{x}{\fns{P}}
  &
  \stp{\Gamma, x:A^\bot}{\fns{Q}}
  &
  \lin{x}{\fns{Q}}
}

\defrule{Sinl}{\srn{\(\oplus_1\)}}{
  \stp{\Gamma, x:A \oplus B}{\sinl{x}{w}{\fns{P}}}
}{
  \stp{\Gamma, x:A \oplus B, w : A}{\fns{P}}
}

\defrule{Sinr}{\srn{\(\oplus_2\)}}{
  \stp{\Gamma, x:A \oplus B}{\sinr{x}{w}{\fns{P}}}
}{
  \stp{\Gamma, x:A \oplus B, w: B}{\fns{P}}
}

\defrule{Swith}{\srn{\(\with\)}}{
  \stp{\Gamma, x:A \with B}{\schoice{x}{w}{\fns{P}}{w}{\fns{Q}}}
}{
  \stp{\Gamma, x:A \with B, w : A}{\fns{P}}
  &
  \stp{\Gamma, x:A \with B, w : B}{\fns{Q}}
}

\defrule{Sotimes}{\srn{\(\otimes\)}}{
  \stp{\Gamma, x:A \otimes B}{\sout{x}{y}{\fns{P}}{w}{\fns{Q}}}
}{
  \stp{\Gamma, x: A \otimes B, y:A}{\fns{P}}
  &
  \lin{y}{\fns{P}}
  &
  \stp{\Gamma, x : A \otimes B, w : B}{\fns{Q}}
}

\defrule{Sparr}{\srn{\(\parr\)}}{
  \stp{\Gamma, x:A \parr B}{\sinp{x}{y}{w}{\fns{P}}}
}{
  \stp{\Gamma, x:A \parr B, w:B, y:A}{\fns{P}}
  &
  \lin{y}{\fns{P}}
}

\defrule{S1}{\srn{\(1\)}}{
  \stp{\Gamma, x:1}{\close{x}}
}{
}

\defrule{Sbot}{\srn{\(\bot\)}}{
  \stp{\Gamma, x:\bot}{\wait{x}{\fns{P}}}
}{
  \stp{\Gamma}{\fns{P}}
}

\defrule{Lfwd1}{\rn{\(L_{\mt{fwd}1}\)}}{\lin{x}{\sfwd{x}{y}}}{}
\defrule{Lfwd2}{\rn{\(L_{\mt{fwd}2}\)}}{\lin{y}{\sfwd{x}{y}}}{}
\defrule{Lclose}{\rn{\(L_{\mt{close}}\)}}{\lin{x}{\sclose{x}}}{}
\defrule{Lwait}{\rn{\(L_{\mt{wait}}\)}}{
  \lin{x}{\swait{x}{\fns{P}}}
}{
  x \notin \fn(\fns{P})
}
\defrule{Lout}{\rn{\(L_{\mt{out}}\)}}{
  \lin{x}{\sout{x}{y}{\fns{P}}{w}{\fns{Q}}}
}{
  \lin{w}{\fns{Q}}
  &
  x \notin \fn(\fns{P}) \cup \fn(\fns{Q})
}
\defrule{Linp}{\rn{\(L_{\mt{inp}}\)}}{
  \lin{x}{\sinp{x}{y}{w}{\fns{P}}}
}{
  \lin{w}{\fns{P}}
  &
  x \notin \fn(\fns{P})
}
\defrule{Linl}{\rn{\(L_{\mt{inl}}\)}}{
  \lin{x}{\sinl{x}{w}{\fns{P}}}
}{
  \lin{w}{\fns{P}}
  &
  x \notin \fn(\fns{P})
}
\defrule{Linr}{\rn{\(L_{\mt{inr}}\)}}{
  \lin{x}{\sinr{x}{w}{\fns{P}}}
}{
  \lin{w}{\fns{P}}
  &
  x \notin \fn(\fns{P})
}
\defrule{Lcase}{\rn{\(L_{\mt{case}}\)}}{
  \lin{x}{\schoice{x}{w}{\fns{P}}{w}{\fns{Q}}}
}{
  \lin{w}{\fns{P}}
  &
  \lin{w}{\fns{Q}}
  &
  x \notin \fn(\fns{P}) \cup \fn(\fns{Q})
}
\defrule{Lwait2}{\rn{\(L_{\mt{wait2}}\)}}{
  \lin{z}{\swait{x}{\fns{P}}}
}{
  \lin{z}{\fns{P}}
}
\defrule{Lout2}{\rn{\(L_{\mt{out2}}\)}}{
  \lin{z}{\sout{x}{y}{\fns{P}}{w}{\fns{Q}}}
}{
  \lin{z}{\fns{P}}
  &
  z \notin \fn(\fns{Q})
}
\defrule{Lout3}{\rn{\(L_{\mt{out3}}\)}}{
  \lin{z}{\sout{x}{y}{\fns{P}}{w}{\fns{Q}}}
}{
  \lin{z}{\fns{Q}}
  &
  z \notin \fn(\fns{P})
}
\defrule{Linp2}{\rn{\(L_{\mt{inp2}}\)}}{
  \lin{z}{\sinp{x}{y}{w}{\fns{P}}}
}{
  \lin{z}{\fns{P}}
}
\defrule{Linl2}{\rn{\(L_{\text{inl2}}\)}}{
  \lin{z}{\sinl{x}{w}{\fns{P}}}
}{
  \lin{z}{\fns{P}}
}
\defrule{Linr2}{\rn{\(L_{\text{inr2}}\)}}{
  \lin{z}{\sinr{x}{w}{\fns{P}}}
}{
  \lin{z}{\fns{P}}
}
\defrule{Lcase2}{\rn{\(L_{\mt{case2}}\)}}{
  \lin{z}{\schoice{x}{w}{\fns{P}}{w}{\fns{Q}}}
}{
  \lin{z}{\fns{P}}
  &
  \lin{z}{\fns{Q}}
}
\defrule{Lpcomp1}{\rn{\(L_{\nu 1}\)}}{
  \lin{z}{\spcomp{x}{A}{\fns{P}}{\fns{Q}}}
}{
  \lin{z}{\fns{P}}
  &
  z \notin \fn(\fns{Q})
}
\defrule{Lpcomp2}{\rn{\(L_{\nu 2}\)}}{
  \lin{z}{\pcomp{x}{A}{\fns{P}}{\fns{Q}}}
}{
  \lin{z}{\fns{Q}}
  &
  z \notin \fn(\fns{P})
}

%% file: adequacy-proof.tex
\section{Proofs for Equivalence of CP and SCP}

We give the details for the proof of equivalence between SCP and SCP typing sequents.
We start by giving all cases for the encoding and decoding functions.
The function \(\escp{{-}}\) takes CP processes to SCP processes, while \(\dscp{{-}}\) maps SCP processes to CP processes.
They are recursively defined on the structure of processes:
\begin{align*}
  \escp{\fwd{x}{y}} &= \sfwd{x}{y} \\
  \escp{\pcomp{x}{A}{P}{Q}} &= \spcomp{x}{A}{\escp{P}}{\escp{Q}} \\
  \escp{\out{x}{y}{P}{Q}} &= \sout{x}{y}{\escp{P}}{x}{\escp{Q}} \\
  \escp{\inp{x}{y}{P}} &= \sinp{x}{y}{x}{\escp{P}} \\
  \escp{\inl{x}{P}} &= \sinl{x}{x}{\escp{P}} \\
  \escp{\inr{x}{P}} &= \sinr{x}{x}{\escp{P}} \\
  \escp{\choice{x}{P}{Q}} &= \schoice{x}{x}{\escp{P}}{x}{\escp{Q}} \\
  \escp{\close{x}} &= \sclose{x} \\
  \escp{\wait{x}{P}} &= \swait{x}{\escp{P}} \\
  \dscp{\sfwd{x}{y}} &= \fwd{x}{y} \\
  \dscp{\spcomp{x}{A}{\fns{P}}{\fns{Q}}} &= \pcomp{x}{A}{\dscp{\fns{P}}}{\dscp{\fns{Q}}} \\
  \dscp{\sout{x}{y}{\fns{P}}{w}{\fns{Q}}} &= \out{x}{y}{\dscp{\fns{P}}}{\subst{x}{w}{\dscp{\fns{Q}}}} \\
  \dscp{\sinp{x}{y}{w}{\fns{P}}} &= \inp{x}{y}{\subst{x}{w}{\dscp{\fns{P}}}} \\
  \dscp{\sinl{x}{w}{\fns{P}}} &= \inl{x}{\subst{x}{w}{\dscp{\fns{P}}}} \\
  \dscp{\sinr{x}{w}{\fns{P}}} &= \inr{x}{\subst{x}{w}{\dscp{\fns{P}}}} \\
  \dscp{\schoice{x}{w}{\fns{P}}{w}{\fns{Q}}} &= \choice{x}{\subst{x}{w}{\dscp{\fns{P}}}}{\subst{x}{w}{\dscp{\fns{Q}}}} \\
  \dscp{\sclose{x}} &= \close{x} \\
  \dscp{\swait{x}{\fns{P}}} &= \wait{x}{\dscp{\fns{P}}}
\end{align*}

We next state some structural lemmas:

\begin{lemma}[Weakening]
  If\/ $\stp{\Gamma}{\fns{P}}$, then $\stp{\Gamma, x:A}{\fns{P}}$.
\end{lemma}

\begin{proof}
  By induction on the derivation of $\stp{\Gamma}{\fns{P}}$.
\end{proof}

\begin{lemma}[Strengthening]
  If\/ $\stp{\Gamma, x:A}{\fns{P}}$ and $x \notin \fn(\fns{P})$, then $\stp{\Gamma}{\fns{P}}$.
\end{lemma}

\begin{proof}
  By induction on the derivation of $\stp{\Gamma, x : A}{\fns{P}}$.
\end{proof}

\begin{lemma}[Free Names are Typed]
  \label{lemma:main_writeup:2}
  \leavevmode
  \begin{enumerate}
  \item If \(\tp{\Delta}{P}\), then \(x \in \fn(P)\) if and only if \(x \in \dom(\Delta)\).
  \item If\/ \(\stp{\Gamma}{\fns{P}}\) and \(x \in \fn(\fns{P})\), then \(x \in \dom(\Gamma)\).
  \end{enumerate}
\end{lemma}

\begin{proof}
  By induction on the derivations of \(\tp{\Delta}{P}\) and of \(\stp{\Gamma}{\fns{P}}\).
\end{proof}

\begin{lemma}[Genericity of Linearity]
  \label{lemma:main_writeup:3}
  If\/ \(\lin{x}{P}\) and \(w \notin \fn(P)\), then \(\lin{x}{\subst{w}{y}{P}}\) for~any~\(y\).
\end{lemma}

\begin{proof}
  By induction on the derivation of \(\lin{x}{P}\).
\end{proof}

\begin{lemma}
  \label{lemma:adequacy-proof:1}
  If \(\tp{\Delta, z : A}{P}\) and \(w \notin \fn(P)\), then \(\subst{w}{z}{\escp{P}} = \escp{\subst{w}{z}{P}}\).
\end{lemma}

\begin{proof}
  By induction on the derivation of \(\tp{\Delta, z : A}{P}\).
\end{proof}

\begin{lemma}
  \label{lemma:main_writeup:6}
  If\/ \(\lin{z}{\fns{P}}\) and \(w \notin \fn(\fns{P})\), then \(\subst{w}{z}{\dscp{\fns{P}}} = \dscp{\subst{w}{z}{\fns{P}}}\).
\end{lemma}

\begin{proof}
  By induction on the derivation \(\lin{z}{\fns{P}}\).
  Assume without loss of generality that all bound names are chosen distinct from free names.
  The principal cases, \ie, \getrn{Lfwd1}, \getrn{Lfwd2}, \getrn{Lclose}, \getrn{Lwait}, \getrn{Lout}, \getrn{Linp}, \getrn{Linl}, \getrn{Linr}, and \getrn{Lcase}, are all immediate.
  \begin{proofcases}
  \item[\getrn{Lwait2}] Follows easily by the induction hypothesis:
    \begin{align*}
      &\subst{w}{z}{(\dscp{\swait{x}{\fns{P}}})}\\
      &= \subst{w}{z}{(\wait{x}{\dscp{\fns{P}}})}\\
      &= \wait{x}{\subst{w}{z}{\dscp{\fns{P}}}}\\
      &= \wait{x}{\dscp{\subst{w}{z}{\fns{P}}}}\\
      &= \dscp{\swait{x}{\subst{w}{z}{\fns{P}}}}\\
      &= \dscp{\subst{w}{z}{(\swait{x}{\fns{P}})}}.
    \end{align*}
  \item[\getrn{Lout2}] Follows easily by the induction hypothesis:
    \begin{align*}
      &\subst{w}{z}(\dscp{\sout{x}{y}{\fns{P}}{u}{\fns{Q}}})\\
      &= \subst{w}{z}{(\out{x}{y}{\dscp{\fns{P}}}{\subst{x}{u}{\dscp{\fns{Q}}}})}\\
      &= \out{x}{y}{\subst{w}{z}{\dscp{\fns{P}}}}{\subst{w, x}{z, u}{\dscp{\fns{Q}}}}\\
      \shortintertext{but \(z \notin \fn{\fns{Q}}\)}
      &= \out{x}{y}{\subst{w}{z}{\dscp{\fns{P}}}}{\subst{x}{u}{\dscp{\fns{Q}}}}\\
      \shortintertext{by the induction hypothesis}
      &= \out{x}{y}{\dscp{\subst{w}{z}{\fns{P}}}}{\subst{x}{u}{\dscp{\fns{Q}}}}\\
      &= \dscp{\sout{x}{y}{\subst{w}{z}{\fns{P}}}{u}{\fns{Q}}}\\
      \shortintertext{again because \(z \notin \fn{\fns{Q}}\)}
      &= \dscp{\sout{x}{y}{\subst{w}{z}{\fns{P}}}{u}{\subst{w}{z}{\fns{Q}}}}\\
      &= \dscp{\subst{w}{z}{(\sout{x}{y}{\fns{P}}{u}{\fns{Q}})}}
    \end{align*}
  \item[\getrn{Linp2}] Follows easily by the induction hypothesis:
    \begin{align*}
      &\subst{w}{z}{(\dscp{\sinp{x}{y}{u}{\fns{P}}})}\\
      &= \subst{w}{z}{(\inp{x}{y}{\subst{x}{u}{\dscp{\fns{P}}}})}\\
      &= \inp{x}{y}{\subst{w, x}{z, u}{\dscp{\fns{P}}}}\\
      \shortintertext{by the induction hypothesis}
      &= \inp{x}{y}{\subst{x}{u}{\dscp{\subst{w}{z}{\fns{P}}}}}\\
      &= \dscp{\sinp{x}{y}{u}{\subst{w}{z}{\fns{P}}}}\\
      &= \dscp{\subst{w}{z}{(\sinp{x}{y}{u}{\fns{P}})}}
    \end{align*}
  \end{proofcases}
  The remaining cases are analogous.
\end{proof}

\begin{lemma}[Linear Names are Free]
  \label{lemma:main_writeup:4}
  If\/ \(\lin{x}{\fns{P}}\), then \(x \in \fn(\fns{P})\).
  Consequently, if\/ \(\lin{\Gamma}{\fns{P}}\) and \(\stp{\Gamma}{\fns{P}}\), then \(\dom(\Gamma) = \fn(\fns{P})\).
\end{lemma}

\begin{proof}
  The first claim is by induction on the derivation of \(\lin{x}{\fns{P}}\).
  To show the second claim, assume \(\lin{\Gamma}{\fns{P}}\) and \(\stp{\Gamma}{\fns{P}}\).
  By the first claim, \(\dom(\Gamma) \subseteq \fn(\fns{P})\), and \(\fn(\fns{P}) \subseteq \dom(\Gamma)\) by \cref{lemma:main_writeup:2}.
  It follows that \(\dom(\Gamma) = \fn(\fns{P})\).
\end{proof}

\begin{lemma}[Syntax-Directedness]
  \label{lemma:main_writeup:1}
  \leavevmode
  Typing judgments and linearity predicates are syntax-directed:
  \begin{enumerate}
  \item For all \(\Delta\) and \(P\), there exists at most one derivation of \(\tp{\Delta}{P}\).
  \item For all \(\Delta\) and \(\fns{P}\), there exists at most one derivation of \(\stp{\Delta}{\fns{P}}\).
  \item For all \(x\) and \(\fns{P}\), there exists at most one derivation of\/ \(\lin{x}{\fns{P}}\).
  \end{enumerate}
\end{lemma}

\begin{proof}
  By induction on the derivation, using the observation that each judgment appears as the conclusion of at most one rule.
\end{proof}

\begin{theorem}[Adequacy]
  The function \(\delta\) is left inverse to \(\varepsilon\), \ie, \(\dscp{\escp{P}} = P\) for all CP processes \(P\).
  Their syntax-directed nature induces functions \(\varepsilon\) and \(\delta\) between CP typing derivations and typing derivations of linear SCP processes:
  \begin{enumerate}
  \item If \(\mc{D}\) is a derivation of\/ \(\tp{\Delta}{P}\), then there exists a derivation \(\escp{\mc{D}}\) of \(\stp{\Delta}{\escp{P}}\), and $\lin{\Delta}{\escp{P}}$ and \(\dscp{\escp{\mc{D}}} = \mc{D}\).
  \item If \(\mc{D}\) is a derivation of \(\stp{\Gamma, \Delta}{\fns{P}}\) where \(\fn(\fns{P}) = \dom(\Delta)\) and \(\lin{\Delta}{\fns{P}}\), then there exists a derivation \(\dscp{\mc{D}}\) of\/ \(\tp{\Delta}{\dscp{\fns{P}}}\), and \(\escp{\dscp{\fns{P}}} = \fns{P}\).
    Moreover, \(\mc{D}\) is the result of weakening the derivation \(\escp{\dscp{\mc{D}}}\) of\/ \(\stp{\Delta}{\fns{P}}\) by \(\Gamma\).
  \end{enumerate}
\end{theorem}
\begin{proof}
  We show that \(\varepsilon\) is a section by induction on the structure of \(P\).
  Nearly all cases follow immediately by the induction hypothesis.
  The interesting cases involve binding.
  Where each final equality is given by the respective induction hypothesis, they are
  \begin{align*}
    &\dscp{\escp{\inl{x}{P}}} \\
    &= \dscp{\sinl{x}{x}{\escp{P}}} \\
    &= \inl{x}{\subst{x}{x}{\left(\dscp{\escp{P}}\right)}} \\
    &= \inl{x}{P}
  \end{align*}
  and
  \begin{align*}
    &\dscp{\escp{\out{x}{y}{P}{Q}}} \\
    &= \dscp{\sout{x}{y}{\escp{P}}{x}{\escp{Q}}} \\
    &= \out{x}{y}{\dscp{\escp{P}}}{\subst{x}{x}{\left(\dscp{\escp{Q}}\right)}} \\
    &= \out{x}{y}{P}{Q}.
  \end{align*}
  The remaining cases are analogous.

  Next, we show that \(\varepsilon\) induces a mapping from CP typing derivations to typing derivations of linear SCP processes.
  Assume that \(\tp{\Delta}{P}\).
  We show that \(\stp{\Delta}{\escp{P}}\) and \(\lin{\Delta}{\escp{P}}\) by induction on the derivation of \(\tp{\Delta}{P}\).
  We will show that \(\dscp{\escp{\mc{D}}} = \mc{D}\) later, once we have defined the action of \(\delta\) on derivations.
  \begin{proofcases}
  \item[\getrn{Cid}] The derivation is:
    \[
      \getrule{Cid}
    \]
    Let \(\escp{\mc{D}}\) be given by \getrn{Sid}.
    The rules \getrn{Lfwd1} and \getrn{Lfwd2} imply the desired linearity predicate.

  \item[\getrn{Ccut}] Assume \(\mc{D}\) is given by
    \[
      \infer[\getrn{Ccut}]{
        \getrc{Ccut}
      }{
        \deduce{
          \getrh{Ccut}{1}
        }{
          \mc{D}_1
        }
        &
        \deduce{
          \getrh{Ccut}{2}
        }{
          \mc{D}_2
        }
      }
    \]
    There exist derivations
    \begin{enumerate}
    \item a derivation \(\escp{\mc{D}_1}\) of \(\stp{\Delta_1, x:A}{\escp{P}}\) \hfill (induction hypothesis) \label{item:adequacy-proof:1}
    \item a derivation \(\escp{\mc{D}_2}\) of \(\stp{\Delta_2, x:A^\bot}{\escp{Q}}\) \hfill (induction hypothesis) \label{item:adequacy-proof:2}
    \item a derivation \(\mc{L}_u\) of \(\lin{u}{\escp{P}}\) for each \(u \in \dom(\Delta_1, x : A)\) \hfill (induction hypothesis)
    \item a derivation \(\mc{L}'_u\) of \(\lin{u}{\escp{Q}}\) for each \(u \in \dom(\Delta_2, x : A^\bot)\) \hfill (induction hypothesis)
    \item a derivation \(\mc{W}_1\) of \(\stp{\Delta_1, \Delta_2, x : A}{\escp{P}}\) \hfill (\cref{lem:scp-wk} and \ref{item:adequacy-proof:1})
    \item a derivation \(\mc{W}_2\) of \(\stp{\Delta_1, \Delta_2, x : A^\bot}{\escp{Q}}\) \hfill (\cref{lem:scp-wk} and \ref{item:adequacy-proof:2})
    \end{enumerate}
    Let \(\escp{\mc{D}}\) be given by
    \[
      \infer[\getrn{Scut}]{
        \tp{\Delta_1, \Delta_2}{\pcomp{x}{A}{\escp{P}}{\escp{Q}}}
      }{
        \deduce{
          \stp{\Delta_1, \Delta_2, x:A}{\escp{P}}
        }{
          \mc{W}_1
        }
        &
        \deduce{
          \lin{x}{\escp{P}}
        }{
          \mc{L}_x
        }
        &
        \deduce{
          \stp{\Delta_1, \Delta_2, x:A^\bot}{\escp{Q}}
        }{
          \mc{W}_2
        }
        &
        \deduce{
          \lin{x}{\escp{Q}}
        }{
          \mc{L}'_x
        }
      }
    \]
    To deduce \(\lin{\Delta_1, \Delta_2}{\pcomp{x}{A}{\escp{P}}{\escp{Q}}}\), observe that \(\Delta_1\) and \(\Delta_2\) type disjoint sets of names by the well-formedness of \(\getrc{Ccut}\).
    \Cref{lemma:main_writeup:2} then implies that each free name in \(\pcomp{x}{A}{\escp{P}}{\escp{Q}}\) appears in either \(\escp{P}\) or \(\escp{Q}\), but not both.
    We are then done by \getrn{Lpcomp1} and \getrn{Lpcomp2} using the derivations \(\mc{L}_u\) and \(\mc{L}'_u\).

  \item[\getrn{Cotimes}] Assume \(\mc{D}\) is given by
    \[
      \infer[\getrn{Cotimes}]{
        \getrc{Cotimes}
      }{
        \deduce{
          \getrh{Cotimes}{1}
        }{
          \mc{D}_1
        }
        &
        \deduce{
          \getrh{Cotimes}{2}
        }{
          \mc{D}_2
        }
      }
    \]
    Let \(w \notin \dom(\Delta_1, \Delta_2, x : A \otimes B)\) be a fresh channel name.
    There exist derivations
    \begin{proofenum}
    \item \(\escp{\mc{D}_1}\) of \(\stp{\Delta_1, y:A}{\escp{P}}\) \hfill (induction hypothesis) \label{item:main_writeup:2}
    \item \(\escp{\mc{D}_2}\) of \(\stp{\Delta_2, x:B}{\escp{Q}}\) \hfill (induction hypothesis) \label{item:main_writeup:1}
    \item \(\mc{L}_u\) of \(\lin{u}{\escp{P}}\) for each \(u \in \dom(\Delta_1, y : A)\) \hfill (induction hypothesis)
    \item \(\mc{L}'_u\) of \(\lin{u}{\escp{Q}}\) for each \(u \in \dom(\Delta_2, x : B)\)  \hfill (induction hypothesis) \label{item:main_writeup:8}
    \item \(\subst{w}{x}{\escp{\mc{D}_2}}\) of \(\stp{\Delta_2, w:B}{\subst{w}{x}{\escp{Q}}}\) \hfill (\labelcref{item:main_writeup:1} and genericity)
    \item \(\mc{W}_1\) of \(\stp{\Delta_1, \Delta_2, x : A \otimes B, y : A}{\escp{P}}\) \hfill (\cref{lem:scp-wk} and \labelcref{item:main_writeup:2})
    \item \(\mc{W}_2\) of \(\stp{\Delta_1, \Delta_2, x : A \otimes B, w : B}{\subst{w}{x}{\escp{Q}}}\) \hfill (\cref{lem:scp-wk} and \labelcref{item:main_writeup:1})
    \item  \(\subst{w}{x}{\mc{L}'_u}\) of \(\lin{u}{\subst{w}{x}{\escp{Q}}}\) for each \(u \in \dom(\Delta_2, x : B)\) \hfill (\cref{lemma:main_writeup:3} and \ref{item:main_writeup:8})
    \end{proofenum}

    Let \(\escp{\mc{D}}\) be given by
    \[
      \infer[\getrn{Sotimes}]{
        \stp{\Delta_1, \Delta_2, x : A \otimes B}{\sout{x}{y}{\escp{P}}{w}{\subst{w}{x}{\escp{Q}}}}
      }{
        \deduce{
          \stp{\Delta_1, \Delta_2, x : A \otimes B, y : A}{\escp{P}}
        }{
          \mc{W}_1
        }
        &
        \deduce{
          \lin{y}{P}
        }{
          \mc{L}_y
        }
        &
        \deduce{
          \stp{\Delta_1, \Delta_2, x : A \otimes B, w : B}{\subst{w}{x}{\escp{Q}}}
        }{
          \mc{W}_2
        }
      }
    \]
    This derivation has the correct conclusion: \(w \notin \fn(\escp{Q})\), so by \cref{lemma:main_writeup:2} and \labelcref{item:main_writeup:1},
    \begin{align*}
      &\escp{\getrc{Cotimes}}\\
      &= \sout{x}{y}{\escp{P}}{x}{\escp{Q}}\\
      &\equiv_\alpha \sout{x}{y}{\escp{P}}{w}{\subst{w}{x}{\escp{Q}}}
    \end{align*}
    are \(\alpha\)-equivalent SCP processes.

    To establish \(\lin{u}{\sout{x}{y}{\escp{P}}{w}{\subst{w}{x}{\escp{Q}}}}\) for \(u \in \dom(\Delta_1, \Delta_2, x : A \otimes B)\), we proceed by case analysis on \(u\).
    \begin{itemize}
    \item If \(u \in \dom(\Delta_1)\), then we are done by \(\mc{L}_u\) and \getrn{Lout2}.
    \item If \(u \in \dom(\Delta_2)\), then we are done by \(\subst{w}{x}{\mc{L}'_u}\) and \getrn{Lout3}.
    \item Assume \(u = x\). We know that \(u \notin \fn(\escp{P})\) by \cref{lemma:main_writeup:2} and \(\stp{\Delta_1, y:A}{\escp{P}}\), and \(u \notin \fn(\subst{w}{x}{\escp{Q}})\) by definition of substitution.
      We are done by \getrn{Lout}.
    \end{itemize}

  \item[\getrn{Cparr}]
    Assume \(\mc{D}\) is given by
    \[
      \infer[\getrn{Cparr}]{
        \getrc{Cparr}
      }{
        \deduce{
          \getrh{Cparr}{1}
        }{
          \mc{D}_1
        }
      }
    \]
    Let \(w \notin \dom(\Delta, x : A \parr B)\) be a fresh channel name.
    There exist derivations
    \begin{proofenum}
    \item \(\escp{\mc{D}_1}\) of \(\stp{\Delta, x : B, y:A}{\escp{P}}\) \hfill (induction hypothesis) \label{item:main_writeup:5}
    \item \(\mc{L}_u\) of \(\lin{u}{\escp{P}}\) for each \(u \in \dom(\Delta, x : B, y : A)\) \hfill (induction hypothesis)  \label{item:main_writeup:7}
    \item \(\subst{w}{x}{\escp{\mc{D}_1}}\) of \(\stp{\Delta, w : B, y : A}{\subst{w}{x}{\escp{P}}}\) \hfill (\labelcref{item:main_writeup:5} and genericity)  \label{item:main_writeup:6}
    \item \(\mc{W}\) of \(\stp{\Delta, x : A \parr B, w : B, y : A}{\subst{w}{x}{\escp{P}}}\) \hfill (\cref{lem:scp-wk} and \ref{item:main_writeup:6})
    \item \(\subst{w}{x}{\escp{\mc{D}_1}}\) of \(\stp{\Delta, w : B, y : A}{\subst{w}{x}{\escp{P}}}\) \hfill (\labelcref{item:main_writeup:5} and genericity)  \label{item:main_write up:6}
    \item  \(\subst{w}{x}{\mc{L}_u}\) of \(\lin{u}{\subst{w}{x}{\escp{P}}}\) for each \(u \in \dom(\Delta, x : B, y : A)\) \hfill (\cref{lemma:main_writeup:3} and \ref{item:main_writeup:7})
    \end{proofenum}

    Let \(\escp{\mc{D}}\) be given by
    \[
      \infer[\getrn{Sparr}]{
        \stp{\Delta, x : A \parr B}{\sinp{x}{y}{w}{\subst{w}{x}{\escp{P}}}}
      }{
        \deduce{
          \stp{\Delta, x : A \parr B, w : B, y : A}{\subst{w}{x}{\escp{P}}}
        }{
          \mc{W}
        }
        &
        \deduce{
          \lin{y}{\subst{w}{x}{P}}
        }{
          \subst{w}{x}{\mc{L}_y}
        }
      }
    \]
    This derivation has the correct conclusion by \(\alpha\)-equivalence.

    To establish \(\lin{u}{\sinp{x}{y}{w}{\subst{w}{x}{\escp{P}}}}\) for \(u \in \dom(\Delta, x : A \parr B)\), we proceed by case analysis on \(u\).
    \begin{itemize}
    \item If \(u \in \dom(\Delta)\), then we are done by \(\subst{w}{x}{\mc{L}_u}\) and \getrn{Linp2}.
    \item If \(u = x\), then we are done by \(\subst{w}{x}{\mc{L}_x}\) and \getrn{Linp}.
    \end{itemize}

  \item[\getrn{Cinl}]
    Assume \(\mc{D}\) is given by
    \[
      \infer[\getrn{Cinl}]{
        \getrc{Cinl}
      }{
        \deduce{
          \getrh{Cinl}{1}
        }{
          \mc{D}_1
        }
      }
    \]
    Let \(w \notin \dom(\Delta, x : A \oplus B)\) be a fresh channel name.
    There exist derivations
    \begin{proofenum}
    \item \(\escp{\mc{D}_1}\) of \(\stp{\Delta, x : A}{\escp{P}}\) \hfill (induction hypothesis) \label{item:main_writeup:25}
    \item \(\mc{L}_u\) of \(\lin{u}{\escp{P}}\) for each \(u \in \dom(\Delta, x : A)\) \hfill (induction hypothesis)  \label{item:main_writeup:27}
    \item \(\subst{w}{x}{\escp{\mc{D}_1}}\) of \(\stp{\Delta, w : A}{\subst{w}{x}{\escp{P}}}\) \hfill (\labelcref{item:main_writeup:25} and genericity)  \label{item:main_writeup:26}
    \item \(\mc{W}\) of \(\stp{\Delta, x : A \oplus B, w : A}{\subst{w}{x}{\escp{P}}}\) \hfill (\cref{lem:scp-wk} and \ref{item:main_writeup:26})
    \item  \(\subst{w}{x}{\mc{L}_u}\) of \(\lin{u}{\subst{w}{x}{\escp{P}}}\) for each \(u \in \dom(\Delta, w : A)\) \hfill (\cref{lemma:main_writeup:3} and \ref{item:main_writeup:27})
    \end{proofenum}

    Let \(\escp{\mc{D}}\) be given by
    \[
      \infer[\getrn{Sinl}]{
        \sinl{x}{w}{\subst{w}{x}{\escp{P}}}
      }{
        \deduce{
          \stp{\Delta, x : A \oplus B, w : A}{\subst{w}{x}{\escp{P}}}
        }{
          \subst{w}{x}{\escp{\mc{D}_1}}
        }
      }
    \]
    This derivation has the correct conclusion by \(\alpha\)-equivalence.

    Linearity \(\lin{u}{\sinl{x}{w}{\subst{w}{x}{\escp{P}}}}\) for \(u \in \dom(\Delta, x : A \oplus B)\) is given by \getrn{Linl} and \(\subst{w}{x}{\mc{L}_x}\) if \(u = x\), and by \getrn{Linl2} and \(\subst{w}{x}{\mc{L}_u}\) otherwise.
  \item[\getrn{Cinr}]
    This case is analogous to the case \getrn{Cinl}.
  \item[\getrn{Cwith}]
    Assume \(\mc{D}\) is given by
    \[
      \infer[\getrn{Cwith}]{
        \getrc{Cwith}
      }{
        \deduce{
          \getrh{Cwith}{1}
        }{
          \mc{D}_1
        }
        &
        \deduce{
          \getrh{Cwith}{2}
        }{
          \mc{D}_2
        }
      }
    \]
    Let \(w \notin \dom(\Delta, x : A \with B)\) be a fresh channel name.
    There exist derivations
    \begin{proofenum}
    \item \(\escp{\mc{D}_1}\) of \(\stp{\Delta, x : A}{\escp{P}}\) \hfill (induction hypothesis) \label{item:main_writeup:35}
    \item \(\escp{\mc{D}_2}\) of \(\stp{\Delta, x : B}{\escp{Q}}\) \hfill (induction hypothesis) \label{item:main_writeup:33}
    \item \(\mc{L}_u\) of \(\lin{u}{\escp{P}}\) for each \(u \in \dom(\Delta, x : A)\) \hfill (induction hypothesis)  \label{item:main_writeup:37}
    \item \(\mc{L}'_u\) of \(\lin{u}{\escp{Q}}\) for each \(u \in \dom(\Delta, x : B)\) \hfill (induction hypothesis)  \label{item:main_writeup:38}
    \item \(\subst{w}{x}{\escp{\mc{D}_1}}\) of \(\stp{\Delta, w : A}{\subst{w}{x}{\escp{P}}}\) \hfill (\labelcref{item:main_writeup:35} and genericity)  \label{item:main_writeup:36}
    \item \(\subst{w}{x}{\escp{\mc{D}_2}}\) of \(\stp{\Delta, w : B}{\subst{w}{x}{\escp{Q}}}\) \hfill (\labelcref{item:main_writeup:33} and genericity)  \label{item:main_writeup:39}
    \item \(\mc{W}_1\) of \(\stp{\Delta, x : A \oplus B, w : A}{\subst{w}{x}{\escp{P}}}\) \hfill (\cref{lem:scp-wk} and \ref{item:main_writeup:36})
    \item \(\mc{W}_2\) of \(\stp{\Delta, x : A \oplus B, w : B}{\subst{w}{x}{\escp{Q}}}\) \hfill (\cref{lem:scp-wk} and \ref{item:main_writeup:39})
    \item \(\subst{w}{x}{\mc{L}_u}\) of \(\lin{u}{\subst{w}{x}{\escp{P}}}\) for each \(u \in \dom(\Delta, w : A)\) \hfill (\cref{lemma:main_writeup:3} and \ref{item:main_writeup:37})
    \item \(\subst{w}{x}{\mc{L}'_u}\) of \(\lin{u}{\subst{w}{x}{\escp{Q}}}\) for each \(u \in \dom(\Delta, w : B)\) \hfill (\cref{lemma:main_writeup:3} and \ref{item:main_writeup:38})
    \end{proofenum}

    Let \(\escp{\mc{D}}\) be given by
    \[
      \infer[\getrn{Swith}]{
        \schoice{x}{w}{\subst{w}{x}{\escp{P}}}{w}{\subst{w}{x}{\escp{Q}}}
      }{
        \deduce{
          \stp{\Delta, x : A \oplus B, w : A}{\subst{w}{x}{\escp{P}}}
        }{
          \mc{W}_1
        }
        &
        \deduce{
          \stp{\Delta, x : A \oplus B, w : B}{\subst{w}{x}{\escp{Q}}}
        }{
          \mc{W}_2
        }
      }
    \]
    This derivation has the correct conclusion by \(\alpha\)-equivalence.

    Linearity \(\lin{u}{\schoice{x}{w}{\subst{w}{x}{\escp{P}}}{w}{\subst{w}{x}{\escp{Q}}}}\) for \(u \in \dom(\Delta, x : A \with B)\) is given by case analysis on \(u\):
    \begin{proofcases}
    \item[\(u = x\)] We are done by \(\subst{w}{x}{\mc{L}_x}\),  \(\subst{w}{x}{\mc{L}'_x}\), and \getrn{Lcase}.
    \item[\(u \in \dom(\Delta)\)] We are done by \(\subst{w}{x}{\mc{L}_x}\),  \(\subst{w}{x}{\mc{L}'_x}\), and \getrn{Lcase2}.
    \end{proofcases}

  \item[\getrn{C1}]
    Assume \(\mc{D}\) is given by
    \[
      \getrule{C1}
    \]
    Let \(\escp{\mc{D}}\) be given by \getrn{S1}.
    Linearity \(\lin{x}{\escp{\sclose{x}}}\) is given by \getrn{Lclose}.

  \item[\getrn{Cbot}]
    Assume \(\mc{D}\) is given by
    \[
      \infer[\getrn{Cbot}]{
        \getrc{Cbot}
      }{
        \deduce{
          \getrh{Cbot}{1}
        }{
          \mc{D}_1
        }
      }
    \]
    There exist derivations
    \begin{proofenum}
    \item \(\escp{\mc{D}_1}\) of \(\stp{\Delta}{\escp{P}}\) \hfill (induction hypothesis)
    \item \(\mc{L}_u\) of \(\lin{u}{\escp{P}}\) for each \(u \in \dom(\Delta)\) \hfill (induction hypothesis)
    \end{proofenum}
    Let \(\escp{\mc{D}}\) be given by
    \[
      \infer[\getrn{Sbot}]{
        \stp{\Delta, x : \bot}{\swait{x}{\escp{P}}}
      }{
        \deduce{
          \stp{\Delta}{\escp{P}}
        }{
          \escp{\mc{D}_1}
        }
      }
    \]
    Linearity \(\lin{u}{\swait{x}{\escp{P}}}\) for \(u \in \dom(\Delta, x : \bot)\) is given by \getrn{Lwait} if \(u = x\), and by \getrn{Lwait2} and \(\mc{L}_u\) otherwise.
  \end{proofcases}

  We now show the converse, namely, that if \(\mc{D}\) is a derivation of \(\stp{\Gamma, \Delta}{\fns{P}}\) where \(\fn(\fns{P}) = \dom(\Delta)\) and \(\lin{\Delta}{\fns{P}}\), then there exists a derivation \(\dscp{\mc{D}}\) of \(\tp{\Delta}{\dscp{\fns{P}}}\) and \(\escp{\dscp{\fns{P}}} = \fns{P}\).
  We will repeatedly use the following fact: if \(\stp{\Gamma}{\fns{P}}\), then by \cref{lemma:main_writeup:2} there exists a \(\Delta \subseteq \Gamma\) such that \(\dom(\Delta) = \fn(\fns{P})\)
  We proceed by induction on the derivation \(\mc{D}\) of \(\stp{\Gamma, \Delta}{\fns{P}}\).

  \begin{proofcases}
  \item[\getrn{Sid}]
    Assume \(\mc{D}\) is given by
    \[
      \getrule{Sid}
    \]
    Then \(\Delta = x : A, y : A^\bot\).
    Let \(\dscp{\mc{D}}\) be given by
    \[
      \getrule{Cid}
    \]
    It is clear that \(\escp{\dscp{\sfwd{x}{y}}} = \sfwd{x}{y}\).

  \item[\getrn{Scut}]
    Assume \(\mc{D}\) is given by
    \[
      \infer[\getrn{Scut}]{
        \getrc{Scut}
      }{
        \deduce{
          \getrh{Scut}{1}
        }{
          \mc{D}_1
        }
        &
        \deduce{
          \getrh{Scut}{2}
        }{
          \mc{L}_1
        }
        &
        \deduce{
          \getrh{Scut}{3}
        }{
          \mc{D}_2
        }
        &
        \deduce{
          \getrh{Scut}{4}
        }{
          \mc{L}_2
        }
      }
    \]

    We start by showing that \(\linp{\fns{P}}\).
    Observe that \(x \in \fn(\fns{P})\) by \cref{lemma:main_writeup:4} and \(\getrh{Scut}{2}\).
    Let \(\Delta_1 \subseteq \Gamma\) be such that \(\dom(\Delta_1, x : A) = \fn(\fns{P})\).
    Showing \(\linp{\fns{P}}\) thus requires showing \(\lin{\Delta_1, x : A}{\fns{P}}\).
    By inversion on \getrn{Lpcomp1}, it follows that \(\lin{z}{\fns{P}}\) for all \(z \in \dom(\Delta_1)\).
    We have \(\lin{x}{\fns{P}}\) by assumption.
    This gives \(\linp{\fns{P}}\) as desired.
    By the induction hypothesis, there exists a derivation \(\dscp{\mc{D}_1}\) of \(\tp{\Delta_1, x : A}{\fns{P}}\).
    An identical argument produces a derivation \(\dscp{\mc{D}_2}\) of \(\tp{\Delta_2, x : A^\bot}{\fns{Q}}\).

    Let \(\dscp{\mc{D}}\) be given by
    \[
      \infer[\getrn{Ccut}]{
        \tp{\Delta_1, \Delta_2}{\pcomp{x}{A}{\dscp{\fns{P}}}{\dscp{\fns{Q}}}}
      }{
        \deduce{
          \tp{\Delta_1, x : A}{\dscp{\fns{P}}}
        }{
          \dscp{\mc{D}_1}
        }
        &
        \deduce{
          \tp{\Delta_2, x : A^\bot}{\dscp{\fns{Q}}}
        }{
          \dscp{\mc{D}_2}
        }
      }
    \]

    Finally, we show that \(\escp{\dscp{\spcomp{x}{A}{\dscp{\fns{P}}}{\dscp{\fns{Q}}}}} = \spcomp{x}{A}{\fns{P}}{\fns{Q}}\).
    By the induction hypothesis, \({\escp{\dscp{\fns{P}}} = \fns{P}}\) and analogously for \(\fns{Q}\).
    Using this, we compute:
    \begin{align*}
      &\escp{\dscp{\spcomp{x}{A}{\dscp{\fns{P}}}{\dscp{\fns{Q}}}}}\\
      &= \escp{\pcomp{x}{A}{\dscp{\fns{P}}}{\dscp{\fns{Q}}}}\\
      &= \spcomp{x}{A}{\escp{\dscp{\fns{P}}}}{\escp{\dscp{Q}}}\\
      &= \spcomp{x}{A}{\fns{P}}{\fns{Q}}.
    \end{align*}

  \item[\getrn{Sotimes}]
    Assume \(\mc{D}\) is given by
    \[
      \infer[\getrn{Sotimes}]{
        \getrc{Sotimes}
      }{
        \deduce{
          \getrh{Sotimes}{1}
        }{
          \mc{D}_1
        }
        &
        \deduce{
          \getrh{Sotimes}{2}
        }{
          \mc{L}
        }
        &
        \deduce{
          \getrh{Sotimes}{3}
        }{
          \mc{D}_2
        }
      }
    \]
    We follow a similar approach as in the case \getrn{Scut}.

    We first show \(\linp{\fns{P}}\), which requires checking the linearity of each free name in \(\fns{P}\).
    By \cref{lemma:main_writeup:2}, \(\fn(\fns{P}) \subseteq \dom(\Gamma, x : A \otimes B, y : A)\).
    We claim that \(\fn(\fns{P}) = \dom(\Delta_1, y : A)\) for some \(\Delta_1 \subseteq \Gamma\):
    \begin{itemize}
    \item \(y \in \fn(\fns{P})\) by \(\getrh{Sotimes}{2}\) and \cref{lemma:main_writeup:4};
    \item \(x \notin \fn(\fns{P})\) by inversion on \(\lin{x}{\sout{x}{y}{\fns{P}}{w}{\fns{Q}}}\) and \getrn{Lout};
    \end{itemize}
    This establishes that \(\fn(\fns{P}) = \dom(\Delta_1, y : A)\) for some \(\Delta_1 \subseteq \Gamma\).

    Having established the set of free names that must be linear, we check \(\linp{P}\).
    To do so, we rely on the fact that the sets \(\fn(\fns{P})\) and \(\fn(\fns{Q})\) are disjoint by inversion on \getrn{Lout}, \getrn{Lout2}, \getrn{Lout3} using \(\lin{\Gamma, x : A \otimes B}{\sout{x}{y}{\fns{P}}{w}{\fns{Q}}}\).
    \begin{itemize}
    \item \(\lin{y}{\fns{P}}\) by assumption;
    \item \(\lin{z}{\fns{P}}\) for all \(z \in \dom(\Delta_1)\) by inversion on \getrn{Lout2} and linearity of \(\sout{x}{y}{\fns{P}}{w}{\fns{Q}}\).
      Indeed, \getrn{Lout2} is the only rule that could have been applied for \(z \in \dom(\Delta_1)\) because \(\fns{P}\) and \(\fns{Q}\) have disjoint sets of free names.
    \end{itemize}
    We conclude \(\linp{P}\).
    It follows by the induction hypothesis that there then exists a derivation \(\dscp{\mc{D}_1}\) of \(\tp{\Delta_1, y : A}{\dscp{\fns{P}}}\).

    A similar argument implies that \(\linp{\fns{Q}}\).
    The induction hypothesis produces a derivation \(\dscp{\mc{D}_2}\) of \(\tp{\Delta_2, w : B}{\dscp{\fns{Q}}}\).
    By genericity, it follows that there exists a derivation \(\subst{x}{w}{\dscp{\mc{D}_2}}\) of \(\tp{\Delta_2, x : B}{\subst{x}{w}{\dscp{\fns{Q}}}}\).

    Let the derivation \(\dscp{\mc{D}}\) be given by
    \[
      \infer[\getrn{Cotimes}]{
        \tp{\Delta_1, \Delta_2, x : A \otimes B}{\out{x}{y}{\dscp{\fns{P}}}{\subst{x}{w}{\dscp{Q}}}}
      }{
        \deduce{
          \tp{\Delta_1, y : A}{\dscp{\fns{P}}}
        }{
          \dscp{\mc{D}_1}
        }
        &
        \deduce{
          \tp{\Delta_2, x : B}{\subst{x}{w}{\dscp{\fns{Q}}}}
        }{
          \subst{x}{w}{\dscp{\mc{D}_2}}
        }
      }
    \]

    It remains to show that \(\escp{\dscp{\sout{x}{y}{\fns{P}}{w}{\fns{Q}}}} = \sout{x}{y}{\fns{P}}{w}{\fns{Q}}\).
    Because \(x \notin \fn(\fns{Q})\), it follows that
    \[
      \sout{x}{y}{\fns{P}}{w}{\fns{Q}} \equiv_\alpha \sout{x}{y}{\fns{P}}{x}{\subst{x}{w}{\fns{Q}}}
    \]
    are \(\alpha\)-equivalent processes.
    By the induction hypothesis, we know that \(\escp{\dscp{\fns{P}}} = \fns{P}\) and \(\escp{\dscp{\fns{Q}}} = \fns{Q}\).
    We compute:
    \begin{align*}
      &\escp{\dscp{\sout{x}{y}{\fns{P}}{w}{\fns{Q}}}}\\
      &= \escp{\dscp{\sout{x}{y}{\fns{P}}{x}{\subst{x}{w}{\fns{Q}}}}}\\
      &= \escp{\out{x}{y}{\dscp{\fns{P}}}{\subst{x}{x}{\dscp{\subst{x}{w}{\fns{Q}}}}}}\\
      &= \escp{\out{x}{y}{\dscp{\fns{P}}}{\dscp{\subst{x}{w}{\fns{Q}}}}}\\
      &= \sout{x}{y}{\escp{\dscp{\fns{P}}}}{x}{\escp{\dscp{\subst{x}{w}{\fns{Q}}}}}\\
      \shortintertext{which by applying \cref{lemma:main_writeup:6,lemma:main_writeup:3} to \(\linp{\fns{Q}}\),}
      &= \sout{x}{y}{\escp{\dscp{\fns{P}}}}{x}{\escp{\subst{x}{w}{(\dscp{\fns{Q}})}}}\\
      \shortintertext{which by applying \cref{lemma:adequacy-proof:1} to \(\tp{\Delta_2, x : B}{\subst{x}{w}{\dscp{\fns{Q}}}}\)}
      &= \sout{x}{y}{\escp{\dscp{\fns{P}}}}{x}{\subst{x}{w}{(\escp{\dscp{\fns{Q}}})}}\\
      &= \sout{x}{y}{\escp{\dscp{\fns{P}}}}{w}{\escp{\dscp{\fns{Q}}})}\\
      &= \sout{x}{y}{\fns{P}}{w}{\fns{Q}}
    \end{align*}
    This completes the case.

  \item[\getrn{Sparr}]
    Assume \(\mc{D}\) is given by
    \[
      \infer[\getrn{Sparr}]{
        \getrc{Sparr}
      }{
        \deduce{
          \getrh{Sparr}{1}
        }{
          \mc{D}_1
        }
        &
        \deduce{
          \getrh{Sparr}{2}
        }{
          \mc{L}_y
        }
      }
    \]

    We show that \(\linp{\fns{P}}\).
    We start by showing that \(\fn{\fns{P}} = \dom(\Delta, w : B, y : A)\) for some \(\Delta \subseteq \Gamma\).
    By inversion on \(\getrh{Sparr}{1}\) with \getrn{Linp} and \getrn{Linp2}, we deduce \(\lin{u}{\fns{P}}\) for all \(u \in \dom(\Delta, w : B)\).
    We know \(\lin{y}{\fns{P}}\) by hypothesis.
    We deduce \(\linp{\fns{P}}\).
    By the induction hypothesis, there exists a derivation \(\dscp{\mc{D}_1}\) of \(\tp{\Delta, w : B, y : A}{\dscp{\fns{P}}}\).
    By genericity, there exists a derivation \(\subst{x}{w}{\dscp{\mc{D}_1}}\)  of \(\tp{\Delta, x : B, y : A}{\subst{x}{w}{\dscp{\fns{P}}}}\).

    Let the derivation \(\dscp{\mc{D}}\) be given by:
    \[
      \infer[\getrn{Cparr}]{
        \tp{\Delta, x : A \parr B}{\inp{x}{y}{\subst{x}{w}{\dscp{\fns{P}}}}}
      }{
        \deduce{
          \tp{\Delta, x : B, y : A}{\subst{x}{w}{(\dscp{\fns{P}})}}
        }{
          \subst{x}{w}{\dscp{\mc{D}_1}}
        }
      }
    \]

    It remains to show that \(\escp{\dscp{\sinp{x}{y}{w}{\fns{P}}}} = \sinp{x}{y}{w}{\fns{P}}\).
    We compute:
    \begin{align*}
      &\escp{\dscp{\sinp{x}{y}{w}{\fns{P}}}}\\
      &= \escp{\inp{x}{y}{\subst{x}{w}{\dscp{\fns{P}}}}}\\
      \shortintertext{which by \cref{lemma:adequacy-proof:1}}
      &= \sinp{x}{y}{x}{\escp{\subst{x}{w}{\dscp{\fns{P}}}}}\\
      &= \sinp{x}{y}{x}{\subst{x}{w}{(\escp{\dscp{\fns{P}}})}}\\
      \shortintertext{which by the induction hypothesis}
      &= \sinp{x}{y}{x}{\subst{x}{w}{\fns{P}}}\\
      &= \sinp{x}{y}{w}{\fns{P}}
    \end{align*}

  \item[\getrn{Sinl}]
    Assume \(\mc{D}\) is given by
    \[
      \infer[\getrn{Sinl}]{
        \getrc{Sinl}
      }{
        \deduce{
          \getrh{Sinl}{1}
        }{
          \mc{D}_1
        }
      }
    \]

    We show that \(\linp{\fns{P}}\) to be able to apply the induction hypothesis.
    By inversion on the assumption \(\linp{\sinl{x}{w}{\fns{P}}}\) with \getrn{Linl}, we know that \(x \notin \fn(\fns{P})\) and \(w \in \fn(\fns{P})\).
    Let \(\Delta \subseteq \Gamma\) be such that \(\fn(\fns{P}) = \dom(\Delta, w : A)\).
    Then by the induction hypothesis, there exists a derivation \(\dscp{\mc{D}_1}\) of \(\tp{\Delta, w : A}{\dscp{\fns{P}}}\).
    By genericity, there exists a derivation \(\subst{x}{w}{\dscp{\mc{D}_1}}\) of \(\tp{\Delta, x : A}{\subst{x}{w}{\dscp{\fns{P}}}}\).

    Let the derivation \(\stp{\Delta, x : A \oplus B}{\inl{x}{P}}\) be given by
    \[
      \infer[\getrn{Cinl}]{
        \getrc{Cinl}
      }{
        \deduce{
          \tp{\Delta, x : A}{\subst{x}{w}{\dscp{\fns{P}}}}
        }{
          \subst{x}{w}{\dscp{\mc{D}_1}}
        }
      }
    \]

    It remains to show that \(\escp{\dscp{\sinl{x}{w}{\fns{P}}}} = \sinl{x}{w}{\fns{P}}\).
    We compute:
    \begin{align*}
      &\escp{\dscp{\sinl{x}{w}{\fns{P}}}}\\
      &= \escp{\inl{x}{\subst{x}{w}{\dscp{\fns{P}}}}}\\
      &= \sinl{x}{x}{\escp{\subst{x}{w}{\dscp{\fns{P}}}}}\\
      \shortintertext{which by \cref{lemma:adequacy-proof:1}}
      &= \sinl{x}{x}{\subst{x}{w}{(\escp{\dscp{\fns{P}}})}}\\
      \shortintertext{which by the induction hypothesis}
      &= \sinl{x}{x}{\subst{x}{w}{\fns{P}}}\\
      &= \sinl{x}{w}{\fns{P}}
    \end{align*}

  \item[\getrn{Sinr}] This case in analogous to the case \getrn{Sinl}.
  \item[\getrn{Swith}]
    Assume \(\mc{D}\) is given by
    \[
      \infer[\getrn{Swith}]{
        \getrc{Swith}
      }{
        \deduce{
          \getrh{Swith}{1}
        }{
          \mc{D}_1
        }
        &
        \deduce{
          \getrh{Swith}{2}
        }{
          \mc{D}_2
        }
      }
    \]

    We show \(\linp{\fns{P}}\).
    By inversion on \(\linp{\schoice{x}{w}{\fns{P}}{w}{\fns{Q}}}\) with \getrn{Lcase} and \getrn{Lcase2}, we deduce \(x \notin \fns{P}\) and \(w \in \fns{P}\).
    Let \(\Delta \subseteq \Gamma\) be such that \(\dom(\Delta, w : A) = \fn(\fns{P})\).
    By the induction hypothesis, there exists a derivation \(\dscp{\mc{D}_1}\) of \(\tp{\Delta, w : A}{\dscp{\fns{P}}}\).
    By genericity, there exists a derivation \(\subst{x}{w}{\dscp{\mc{D}_1}}\) of \(\tp{\Delta, x : A}{\subst{x}{w}{\fns{P}}}\).

    An analogous argument gives \(\linp{\fns{Q}}\) and a derivation a derivation \(\subst{x}{w}{\dscp{\mc{D}_2}}\) of \(\tp{\Delta, x : B}{\subst{x}{w}{\fns{Q}}}\).

    Let \(\dscp{\mc{D}}\) be given by
    \[
      \infer[\getrn{Cwith}]{
        \getrc{Cwith}
      }{
        \deduce{
          \tp{\Delta, x : A}{\subst{x}{w}{\fns{P}}}
        }{
          \subst{x}{w}{\dscp{\mc{D}_1}}
        }
        &
        \deduce{
          \tp{\Delta, x : B}{\subst{x}{w}{\fns{Q}}}
        }{
          \subst{x}{w}{\dscp{\mc{D}_2}}
        }
      }
    \]

    The induction hypothesis and \cref{lemma:adequacy-proof:1} imply \(\escp{\dscp{\schoice{x}{w}{\fns{P}}{w}{\fns{Q}}}} = \schoice{x}{w}{\fns{P}}{w}{\fns{Q}}\).

  \item[\getrn{S1}]
    Assume \(\mc{D}\) is given by
    \[
      \getrule{S1}
    \]
    Let \(\dscp{\mc{D}}\) be given by
    \[
      \getrule{C1}
    \]
    It is clear that \(\escp{\dscp{\sclose{x}}} = \sclose{x}\).
  \item[\getrn{Sbot}]
    Assume \(\mc{D}\) is given by
    \[
      \infer[\getrn{Sbot}]{
        \getrc{Sbot}
      }{
        \deduce{
          \getrh{Sbot}{1}
        }{
          \mc{D}_1
        }
      }
    \]

    By inversion on \(\linp{\swait{x}{\fns{P}}}\) and \getrn{Lwait}, \(x \notin \fn(\fns{P})\).
    Let \(\Delta \subseteq \Gamma\) be such that \(\fn(\fns{P}) = \dom(\Delta)\).
    By the induction hypothesis, there exists a derivation \(\dscp{\mc{D}}\) of \(\tp{\Delta}{\dscp{\fns{P}}}\).

    Let the derivation \(\dscp{\mc{D}}\) be given by
    \[
      \infer[\getrn{Cbot}]{
        \getrc{Cbot}
      }{
        \deduce{
          \getrh{Cbot}{1}
        }{
          \dscp{\mc{D}}
        }
      }
    \]

    It is easy to check using the induction hypothesis that \(\escp{\dscp{\swait{x}{\fns{P}}}} = \swait{x}{\fns{P}}\).
  \end{proofcases}

  Finally, we show the desired identities for compositions of \(\varepsilon\) and \(\delta\) on derivations.
  When \(\mc{D}\) is a CP typing derivation, \cref{lemma:main_writeup:1} implies that \(\dscp{\escp{\mc{D}}} = \mc{D}\).
  Indeed, if \(\mc{D}\) is a derivation of \(\tp{\Delta}{P}\), then so is \(\dscp{\escp{\mc{D}}}\) by the above.
  But there exists at most one derivation of \(\tp{\Delta}{P}\), so we conclude \(\mc{D} = \dscp{\escp{\mc{D}}}\).

  If \(\mc{D}\) is a derivation of \(\stp{\Gamma, \Delta}{\fns{P}}\) where \(\dom(\Delta) = \fn(\fns{P})\) and \(\linp{\fns{P}}\), then \(\escp{\dscp{\mc{D}}}\) is a derivation of \(\stp{\Delta}{\fns{P}}\).
  Weakening this derivation by \(\Gamma\) gives a derivation \(\mc{W}\) of \(\stp{\Gamma, \Delta}{\fns{P}}\).
  But \cref{lemma:main_writeup:1} then implies that \(\mc{W} = \mc{D}\).
  This is what we wanted to show.
\end{proof}


%% file: adequacy-scp-encoding.tex
\section{Adequacy proof of the LF encoding of SCP}

\subsection{Proof of \Cref{lem:adeq-tp}}
\begin{lemma}[Adequacy of {\color{magenta}tp}]
  There exists a bijection between the set of session types and canonical LF terms $M$ such that $\lftp{M}$.
\end{lemma}
\begin{proof}
We define the encoding $\etp{-}$ and decoding $\dtp{-}$ of types in SCP as follows:
\begin{align*}
  \etp{1} &= 1 & \etp{\bot} &= \bot \\
  \etp{A \otimes B} &= \etp{A} \otimes \etp{B} & \etp{A \parr B} &= \etp{A} \parr \etp{B} \\
  \etp{A \oplus B} &= \etp{A} \oplus \etp{B} & \etp{A \with B} &= \etp{A} \with \etp{B} \\
  \dtp{1} &= 1 & \dtp{\bot} &= \bot \\
  \dtp{A \otimes B} &= \dtp{A} \otimes \dtp{B} & \dtp{A \parr B} &= \dtp{A} \parr \dtp{B} \\
  \dtp{A \oplus B} &= \dtp{A} \oplus \dtp{B} & \dtp{A \with B} &= \dtp{A} \with \dtp{B}
\end{align*}
$\etp{-}$ and $\dtp{-}$ are clearly inverses of each other and satisfies adequacy.
\end{proof}

\subsection{Proof of \Cref{lem:adeq-dual}}
\begin{lemma}[Adequacy of {\color{magenta}dual}]
  \leavevmode
  \begin{enumerate}
    \item For any session type $A$, there exists a unique LF canonical form $D$ such
      that $\lfdual{D}{\etp{A}}{\etp{A^\bot}}$
    \item For any LF canonical form $D$ such that $\lfdual{D}{\etp{A}}{\etp{A'}}$, $\;A' = A^\bot$.
  \end{enumerate}
\end{lemma}
\begin{proof}
We show some cases for both parts. For (1), by induction on $A$.
\setcounter{case}{0}
\begin{case}
  $A = 1$, then $A^\bot = \bot$, and therefore $\etp{A} = 1$ and $\etp{A^\bot} = \bot$. So we can use the LF constructor
  @D1 : dual 1 bot@. Uniqueness follows by inspecting other constructors.
\end{case}
\begin{case}
  $A = B \otimes C$, then $A^\bot = B^\bot \parr C^\bot$, and therefore $\etp{A} = \etp{B} \otimes \etp{C}$
  and $\etp{A^\bot} = \etp{B^\bot} \parr \etp{C^\bot}$. \\
  Then by induction hypothesis there exist unique LF derivations $D$ and $E$ such that ${\lfdual{D}{\etp{B}}{\etp{B^\bot}}}$ and
  ${\lfdual{E}{\etp{C}}{\etp{C}}}$. We now use the constructor @D⊗@, and uniqueness follows by inspecting other
  constructors.
\end{case}
For (2), by induction on the dervation of $D$ such that ${\lfdual{D}{\etp{A}}{\etp{A'}}}$.
\setcounter{case}{0}
\begin{case}
  @D1@, then $\etp{A} = 1$ and $\etp{A'} = \bot$, so $A = 1$ and $A' = \bot = A^\bot$.
\end{case}
\begin{case}
  @D⊗@, then $\etp{A} = \etp{B} \otimes \etp{C}$ and $\etp{A'} = \etp{B'} \otimes \etp{C'}$ for some $B, B', C, C'$ and
  there are LF forms $D$ and $E$ such that ${\lfdual{D}{\etp{B}}{\etp{B'}}}$ and ${\lfdual{E}{\etp{C}}{\etp{C'}}}$. By
  induction hypothesis, $B' = B^\bot$ and $C' = C^\bot$, so $\etp{A} = \etp{A^\bot}$.
\end{case}
\end{proof}

\subsection{Proof of \Cref{lem:adeq-proc}}
We first define the encoding $\eproc{-}$ and decoding $\dproc{-}$ of processes in SCP as follows.
For the encoding, the idea is to represent all bindings as intuitionistic functions.
For the decoding, the idea is to perform function application on relevant substructures with freshly bound channels.
\begin{align*}
  \eproc{\sfwd{x}{y}} &= \lffwd{x}{y} & \eproc{\spcomp{x}{A}{\fns{P}}{\fns{Q}}} &= \lfpcomp{x}{\etp{A}}{\eproc{\fns{P}}}{\eproc{\fns{Q}}} \\
  \eproc{\sclose{x}} &= \lfclose{x} & \eproc{\swait{x}{\fns{P}}} &= \lfwait{x}{\eproc{\fns{P}}} \\
  \eproc{\sout{x}{y}{\fns{P}}{w}{\fns{Q}}} &= \lfout{x}{(\lambda y. \eproc{\fns{P}})}{(\lambda w. \eproc{\fns{Q}})}
& \eproc{\sinp{x}{y}{w}{\fns{P}}} &= \lfinp{x}{(\lambda w. \lambda y. \eproc{\fns{P}})} \\
  \eproc{\sinl{x}{w}{\fns{P}}} &= \lfinl{x}{(\lambda w. \eproc{\fns{P}})} & \eproc{\sinr{x}{w}{\fns{P}}} &= \lfinr{x}{(\lambda w. \eproc{\fns{P}})} \\
  \eproc{\schoice{x}{w}{\fns{P}}{w}{\fns{Q}}} &= \lfchoice{x}{(\lambda w. \eproc{\fns{P}})}{(\lambda w. \eproc{\fns{Q}})}
\end{align*}
\begin{align*}
  \dproc{\lffwd{x}{y}} &= \sfwd{x}{y} & \dproc{\lfpcompd{T}{M}{N}} &= \spcomp{x}{\dtp{T}}{\dproc{M\;x}}{\dproc{N\;x}} \\
  \dproc{\lfclose{x}} &= \sclose{x} & \dproc{\lfwait{x}{M}} &= \swait{x}{\dproc{M}} \\
  \dproc{\lfout{x}{M}{N}} &= \sout{x}{y}{\dproc{M\;y}}{w}{\dproc{N\;w}} & \dproc{\lfinp{x}{M}} &= \sinp{x}{y}{w}{\dproc{M\;w\;y}} \\
  \dproc{\lfinl{x}{M}} &= \sinl{x}{w}{\dproc{M\;w}} & \dproc{\lfinr{x}{M}} &= \sinr{x}{w}{\dproc{M\;w}} \\
  \dproc{\lfchoice{x}{M}{N}} &= \schoice{x}{w}{\dproc{M\;w}}{w}{\dproc{N\;w}}
\end{align*}

\begin{lemma}[Adequacy of {\color{magenta}proc}]
  For each SCP processes \(\fns{P}\), there exists a unique canonical LF derivation \(\lfproc{\enc{\fn(\fns{P})}}{\enc{\fns{P}}}\) and \(\dec{\enc{\fns{P}}} = \fns{P}\).
  Conversely, if\/ \(\lfproc{\Gamma}{M}\) is a canonical LF derivation, then \(\dec{M}\) is an SCP process, \(\enc{\dec{M}} = M\), and \(\enc{\fn(\dec{M})} \subseteq \Gamma\).
\end{lemma}
\begin{proof}
  For the forward direction, by induction on $\fns{P}$. We show two cases. For all cases, uniqueness follows from inspecting constructors to observe that every process construct in SCP has a corresponding constructor in our encoding. Left invertibility of $\eproc{-}$ follows from simple computation of each case.
\setcounter{case}{0}
\begin{case}
  $\fns{P} = \sfwd{x}{y}$\\
  Then $\fn(\fns{P}) = \{x, y\}$, and indeed, $\lfproc{\lfn{x}, \lfn{y}}{\lffwd{x}{y}}$.
\end{case}
\begin{case}
  $\fns{P} = \sinl{x}{w}{\fns{P'}}$\\
  Then $\eproc{\fns{P}} = \lfinl{x}{(\lambda w. \eproc{\fns{P'}})}$. Therefore $\lfproc{\enc{\fn(\fns{P'})}}{\eproc{\fns{P'}}}$ by induction hypothesis.
  Since $\fn(\fns{P}) = (\fn(\fns{P'}) \setminus{w}) \cup \{x\}$, we have $\lfproc{\enc{\fn(\fns{P})}}{\eproc{\fns{P}}}$ by the LF constructor @inl@.
\end{case}
For the reverse direction, by induction on the derivation $\lfproc{\Gamma}{M}$.
\setcounter{case}{0}
\begin{case}
  $\lfproc{\Gamma}{\lffwd{x}{y}}$\\
  Then $\dproc{M} = \sfwd{x}{y}$ and invertibility is obvious.
  Moreover, $\lfn{x}$ and $\lfn{y}$ must appear in $\Gamma$ since there are no constructors of the LF type @name@.
  Indeed, $\fn(\sfwd{x}{y}) = \{\lfn{x}, \lfn{y}\}$ which is a subset of $\Gamma$.
\end{case}
\begin{case}
  $\lfproc{\Gamma}{\lfinl{x}{M'}}$\\
  Then $\dproc{M} = \sinl{x}{w}{\dproc{M'\;w}}$.
  For invertibility, we have
  \begin{align*}
    \eproc{\dproc{\lfinl{x}{M'}}} &= \eproc{\sinl{x}{w}{\dproc{M'\;w}}} \\
                                  &= \lfinl{x}{(\lambda w. \eproc{\dproc{M'\;w}})} \\
                                  &= \lfinl{x}{(\lambda w. M'\;w)} \\
                                  &= \lfinl{x}{M'}
  \end{align*}
  For the context condition, we have $\enc{\fn(\dproc{M'\ w})} \subseteq \Gamma, \lfn{w}$ by induction hypothesis using some fresh $\lfn{w}$.
  Since $\fn(\dproc{M}) = (\fn(\dproc{M'\ w}) \setminus w) \cup \{x\}$, we obtain
  \[
    \enc{\fn(\dproc{M}} \subseteq \Gamma
  \]
  since $\lfn{x} \in \Gamma$ due to it being used in $\lfinl{x}{M'}$.
\end{case}
\end{proof}

\subsection{Proof of \Cref{lem:adeq-linear}}
\begin{lemma}[Adequacy of {\color{magenta}linear}]
  For each derivation $\mc{D}$ of \/ $\lin{x}{\fns{P}}$, there exists a unique canonical LF derivation $L = \elin{\mc{D}}$ such that
  $\lflind{\enc{\fn(\fns{P}) \setminus x}}{L}{x}{\eproc{\fns{P}}}$ and $\dlin{L} = \mc{D}$.
  Conversely, if \/ $\lflinm{\Gamma}{L}{M}$ is a canonical LF derivation, then $\dlin{L}$ is a derivation of \/ $\lin{x}{\dproc{M\ x}}$
  and $\lflinm{\enc{\fn(\dproc{M\ x}) \setminus x}}{\elin{\dlin{L}}}{M}$ where $\enc{\fn(\dproc{M\ x})} \subseteq \Gamma$.
\end{lemma}
\begin{proof}
Since there is an exact correspondence between inference rules for the linearity predicate in SCP and the constructors for the LF type family @linear@, we informally case on a few inference rules and prove both directions.
The encodings and decodings are modular, so we present the result of encoding and decoding as needed in the proof.
We omit the verification of the invertibility statements since they are obvious.
\setcounter{case}{0}
\begin{case}[\getrn{Lfwd1}]
  We start with the forward direction. Suppose $\mc{D}$ is
  \[
    \getrule{Lfwd1}
  \]
  Where $\elin{\mc{D}} = \lfLfwdAn$. Then indeed,
  $$\lflind{\enc{\fn(\sfwd{x}{y}) \setminus x}}{\lfLfwdAn}{x}{\eproc{\sfwd{x}{y}}}$$
  because $\eproc{\sfwd{x}{y}} = \lffwd{x}{y}$ and $\enc{\fn(\sfwd{x}{y}) \setminus x} = \lfn{y}$.
  The reverse direction uses a similar argument.
\end{case}

\begin{case}[\getrn{Linl}]
  Starting with the forward direction, where $\mc{D}$ is
  \[
    \getrule{Linl}
  \]
  Then $\elin{\mc{D}} = \lfLinln : \lfL{w}{\eproc{\fns{P}}} \to \lfL{x}{(\lfinl{x}{\lambda w. \eproc{\fns{P}}}}$.
  \\First, we have $\lflind{\enc{\fn(\fns{P}) \setminus w}}{L}{w}{\eproc{\fns{P}}}$ by induction hypothesis, so it suffices to show
  \[
    \lflind{\enc{\fn(\sinl{x}{w}{\fns{P}}) \setminus x}}{(\lfLinln\ L)}{x}{\eproc{\sinl{x}{w}{\fns{P}}}}
  \]
  First, we observe that $\eproc{\sinl{x}{w}{\fns{P}}} = \lfinl{x}{\lambda w. \eproc{\fns{P}}}$ by definition of $\eproc{-}$, so
  the typing matches.
  Next, to show that the context matches, we must show $\enc{\fn(\fns{P}) \setminus w} = \enc{\fn(\sinl{x}{w}{\fns{P}}) \setminus x}$, which follows from the side condition $\getrh{Linl}{2}$.
  
  Consider the converse next where we have an LF derivation of form
  \[
    \lf{\Gamma}{(\lfLinln\ L) : \lfL{x}{(\lfinl{x}{M}}}
  \]
  such that $\lf{\Gamma}{L : \lfLo{M}}$.
  \\Then by induction hypothesis, we have $\dlin{L}$ is a derivation of $\lin{w}{\dproc{M\ w}}$.
  Furthermore, $M$ cannot depend on $x$ since it is a metavariable and therefore must be independent of the internally bound $x$.
  Therefore, $x \notin \fn(\dproc{M\ w})$. And therefore, we can apply \getrn{Linl}.
 \end{case}
\end{proof}

\subsection{Proof of \Cref{lem:adeq-wtp}}
\begin{lemma}[Adequacy of {\color{magenta}wtp}]
  There exists a bijection between typing derivations in SCP of form
  $\stp{\Gamma}{\fns{P}}$ and LF canonical forms $D$ such that
  $\lfwtp{\enc{\Gamma}}{D}{\enc{\fns{P}}}$
\end{lemma}
\begin{proof}
  Just like in the proof sketch for adequacy on linearity, we informally case on a few typing rules and prove both directions since there is an exact correspondence between typing rules in SCP and the constructors for the LF type family @wtp@.
  The encodings and decodings are modular, so we informally present the result of encoding and decoding as needed in the proof. 
  We omit verifying invertibility statements since they are obvious.
\setcounter{case}{0}
\begin{case}[\getrn{Sid}]
\[
  \getrule{Sid}
\]
This rule corresponds to  
\newline @wtp_fwd : dual T T' → {X:name}hyp X T → {Y:name}hyp Y T' → wtp (fwd X Y)@.
\newline By \Cref{lem:adeq-dual}, we have a unique derivation $\lfdual{D}{\etp{A}}{\etp{A^\bot}}$, so we set $A = T$ and $A^\bot = T'$.
By expanding the context encoding, $\enc{\Gamma, \h{x}{A}, \h{y}{A^\bot}} = \enc{\Gamma}, \lfb{hx}{x}{A}, \lfb{hy}{y}{A^\bot}$,
allowing usage of @wtp_fwd@.
\\
We use a similar argument for the reverse direction, in particular, we apply \Cref{lem:adeq-dual} to infer $A' = A^\bot$ from $\lfdual{D}{\etp{A}}{\etp{A'}}$.
We then prove that the decoding of the LF context yields a context of form $\Gamma, \h{x}{A}, \h{y}{A'}$ with $A' = A^\bot$.
\end{case}
\begin{case}[\getrn{Scut}]
\[
  \getrule{Scut}
\]
This rule corresponds to
\begin{lstlisting}
wtp_pcomp : dual T T' → ({x:name} hyp x T → wtp (M x)) → ({x:name} hyp x T' → wtp (N x))
          → linear M → linear N
          → wtp (pcomp T M N)
\end{lstlisting}
First, $\etp{A} = T$ and $\etp{A^\bot} = T'$ by \Cref{lem:adeq-dual}. 
By induction hypothesis, we have unique derivations 
$\lfwtp{\enc{\Gamma, \h{x}{A}}}{\mc{D}_1}{\eproc{\fns{P}}}$ and
$\lfwtp{\enc{\Gamma, \h{x}{A^\bot}}}{\mc{D}_2}{\eproc{\fns{Q}}}$. 
\\
The linearity predicates follow from \Cref{lem:adeq-linear}. For example, we can encode the derivation of the predicate $\lin{x}{\fns{P}}$ to some $L$ such that
$\lflind{\enc{\fn(\fns{P}) \setminus x}}{L}{x}{\eproc{\fns{P}}}$. Moreover, $\enc{\fn(\fns{P})} \subseteq \enc{\Gamma}$, so by weakening, we have
${\lflind{\enc{\Gamma}}{L}{x}{\enc{\fns{P}}}}$.
We apply the same reasoning to obtain the corresponding derivation to $\lin{x}{\fns{Q}}$, thereby enabling use of @wtp_pcomp@ as desired.
\\
For the reverse direction, we have a derivation
\[
  \lfwtp{\Gamma}{\lfwtppcomp\ D\ \mc{D}_1\ \mc{D}_2\ \mc{L}_1\ \mc{L}_2}{\lfpcompd{T}{M}{N}}
\]
where:
\begin{align*}
  \lfdual{D}{T}{T'} \\
  \lfnhwtp{\Gamma}{\mc{D}_1}{M\ x}{x}{T} \\
  \lfnhwtp{\Gamma}{\mc{D}_2}{N\ x}{x}{T'} \\
  \lflinm{\Gamma}{\mc{L}_1}{M} \\
  \lflinm{\Gamma}{\mc{L}_2}{N}
\end{align*}
First, we let $A = \dtp{T}$, or $\enc{A} = T$, and by \Cref{lem:adeq-dual}, $\enc{A^\bot} = T'$.
We also obtain two derivations of @wtp@ by extending the context on $\mc{D}_1$ and $\mc{D}_2$:
\begin{align*}
  \lfwtp{\Gamma, \lfb{h}{x}{A}}{(\mc{D}_1\ x\ h)}{M\ x} \\
  \lfwtp{\Gamma, \lfb{h}{x}{A^\bot}}{(\mc{D}_2\ x\ h)}{N\ x}
\end{align*}

By induction hypotheses on $(\mc{D}_1\ x\ h)$ and $(\mc{D}_2\ x\ h)$, we obtain two SCP type derivations.
\[
  \stp{\Delta, \h{x}{A}}{\dproc{M\ x}}
  \quad
  \stp{\Delta, \h{x}{A^\bot}}{\dproc{N\ x}}
\]
where $\Delta$ is an SCP typing context such that $\enc{\Delta} = \Gamma$.

Next, by \Cref{lem:adeq-linear} on $\mc{L}_1$ and $\mc{L}_2$, we obtain two SCP linear predicates:
\[
  \lin{x}{\dproc{M\ x}}
  \quad
  \lin{x}{\dproc{N\ x}}
\]
And finally, we apply \getrn{Scut}:
\[
  \infer[\getrn{Scut}]{\stp{\Gamma}{\spcomp{x}{A}{\dproc{M\ x}}{\dproc{N\ x}}}}
  {\stp{\Gamma, \h{x}{A}}{\dproc{M\ x}} & \lin{x}{\dproc{M\ x}}
  &\stp{\Gamma, \h{x}{A^\bot}}{\dproc{N\ x}} & \lin{x}{\dproc{N\ x}}}
\]
Finally, we verify that $\eproc{\lfpcompd{T}{M}{N}} = \spcomp{x}{A}{\dproc{M\ x}}{\dproc{N\ x}}$.
\end{case}
\end{proof}

\subsection{Proof of \Cref{lem:adeq-step} and \Cref{lem:adeq-equiv}}
\begin{lemma}[Adequacy of {\color{magenta}step}]
  For each SCP reduction $S$ of $\fns{P} \sst \fns{Q}$, there exists a unique canonical LF derivation
  $\lfstep{\enc{\fn(\fns{P})}}{D}{\enc{\fns{P}}}{\enc{\fns{Q}}}$ and $\dec{\enc{S}} = S$.
  Conversely, if\/ \(\lfstep{\Gamma}{D}{M}{N}\) is a canonical LF derivation, then \(\dec{D}\) is a derivation of
  a reduction $\dproc{M} \sst \dproc{N}$, \(\enc{\dec{D}} = D\), and \(\enc{\fn(\dec{M})} \subseteq \Gamma\).
\end{lemma}
\begin{lemma}[Adequacy of {\color{magenta}equiv}]
  For each SCP structural equivalence $S$ of $\fns{P} \equiv \fns{Q}$, there exists a unique canonical LF derivation
  $\lfequiv{\enc{\fn(\fns{P})}}{D}{\enc{\fns{P}}}{\enc{\fns{Q}}}$ and $\dec{\enc{S}} = S$.
  Conversely, if\/ \(\lfequiv{\Gamma}{D}{M}{N}\) is a canonical LF derivation, then \(\dec{D}\) is a derivation of
  a reduction $\dproc{M} \equiv \dproc{N}$, \(\enc{\dec{D}} = D\), and \(\enc{\fn(\dec{M})} \subseteq \Gamma\).
\end{lemma}
\begin{proof}
  Both are easy to prove. 
  For both directions, all axiom cases can be shown by appealing to \Cref{lem:adeq-proc} on both $\fns{P}$ and $\fns{Q}$ (or $M$ and $N$ for the reverse direction).
  The two congruence cases in @step@ for @pcomp@ is a straightforward application of the induction hypothesis.
\end{proof}

%% file: main_writeup.bbl

\begin{thebibliography}{34}


\ifx \showCODEN    \undefined \def \showCODEN     #1{\unskip}     \fi
\ifx \showDOI      \undefined \def \showDOI       #1{#1}\fi
\ifx \showISBNx    \undefined \def \showISBNx     #1{\unskip}     \fi
\ifx \showISBNxiii \undefined \def \showISBNxiii  #1{\unskip}     \fi
\ifx \showISSN     \undefined \def \showISSN      #1{\unskip}     \fi
\ifx \showLCCN     \undefined \def \showLCCN      #1{\unskip}     \fi
\ifx \shownote     \undefined \def \shownote      #1{#1}          \fi
\ifx \showarticletitle \undefined \def \showarticletitle #1{#1}   \fi
\ifx \showURL      \undefined \def \showURL       {\relax}        \fi
\providecommand\bibfield[2]{#2}
\providecommand\bibinfo[2]{#2}
\providecommand\natexlab[1]{#1}
\providecommand\showeprint[2][]{arXiv:#2}

\bibitem[\protect\citeauthoryear{Balzer and Pfenning}{Balzer and
  Pfenning}{2017}]%
        {Balzer17icfp}
\bibfield{author}{\bibinfo{person}{Stephanie Balzer} {and}
  \bibinfo{person}{Frank Pfenning}.} \bibinfo{year}{2017}\natexlab{}.
\newblock \showarticletitle{Manifest Sharing with Session Types}. In
  \bibinfo{booktitle}{\emph{International Conference on Functional Programming
  (ICFP)}}. \bibinfo{publisher}{ACM}, \bibinfo{pages}{37:1--37:29}.
\newblock
\newblock
\shownote{Extended version available as Technical Report
  \href{http://www.cs.cmu.edu/~fp/papers/CMU-CS-17-106R.pdf}{CMU-CS-17-106R},
  June 2017.}


\bibitem[\protect\citeauthoryear{Castro{-}Perez, Ferreira, and
  Yoshida}{Castro{-}Perez et~al\mbox{.}}{2020}]%
        {Castro-Perez20tacas}
\bibfield{author}{\bibinfo{person}{David Castro{-}Perez},
  \bibinfo{person}{Francisco Ferreira}, {and} \bibinfo{person}{Nobuko
  Yoshida}.} \bibinfo{year}{2020}\natexlab{}.
\newblock \showarticletitle{{EMTST:} Engineering the Meta-theory of Session
  Types}. In \bibinfo{booktitle}{\emph{Tools and Algorithms for the
  Construction and Analysis of Systems - 26th International Conference, {TACAS}
  2020, Held as Part of the European Joint Conferences on Theory and Practice
  of Software, {ETAPS} 2020, Dublin, Ireland, April 25-30, 2020, Proceedings,
  Part {II}}} \emph{(\bibinfo{series}{Lecture Notes in Computer Science})},
  \bibfield{editor}{\bibinfo{person}{Armin Biere} {and} \bibinfo{person}{David
  Parker}} (Eds.), Vol.~\bibinfo{volume}{12079}. \bibinfo{publisher}{Springer},
  \bibinfo{pages}{278--285}.
\newblock
\urldef\tempurl%
\url{https://doi.org/10.1007/978-3-030-45237-7\_17}
\showDOI{\tempurl}


\bibitem[\protect\citeauthoryear{Cave and Pientka}{Cave and Pientka}{2012}]%
        {Cave:POPL12}
\bibfield{author}{\bibinfo{person}{Andrew Cave} {and} \bibinfo{person}{Brigitte
  Pientka}.} \bibinfo{year}{2012}\natexlab{}.
\newblock \showarticletitle{Programming with binders and indexed data-types}.
  In \bibinfo{booktitle}{\emph{{39th {ACM} SIGPLAN-SIGACT Symposium on
  Principles of Programming Languages (POPL'12)}}}. \bibinfo{pages}{413--424}.
\newblock


\bibitem[\protect\citeauthoryear{Cervesato, Pfenning, Walker, and
  Watkins}{Cervesato et~al\mbox{.}}{2002}]%
        {Cervesato02tr}
\bibfield{author}{\bibinfo{person}{Iliano Cervesato}, \bibinfo{person}{Frank
  Pfenning}, \bibinfo{person}{David Walker}, {and} \bibinfo{person}{Kevin
  Watkins}.} \bibinfo{year}{2002}\natexlab{}.
\newblock \bibinfo{booktitle}{\emph{A Concurrent Logical Framework {II}:
  Examples and Applications}}.
\newblock \bibinfo{type}{{T}echnical {R}eport} CMU-CS-02-102.
  \bibinfo{institution}{Department of Computer Science, Carnegie Mellon
  University}.
\newblock
\newblock
\shownote{Revised May 2003.}


\bibitem[\protect\citeauthoryear{Chlipala}{Chlipala}{2008}]%
        {Chlipala:ICFP08}
\bibfield{author}{\bibinfo{person}{Adam~J. Chlipala}.}
  \bibinfo{year}{2008}\natexlab{}.
\newblock \showarticletitle{Parametric higher-order abstract syntax for
  mechanized semantics}. In \bibinfo{booktitle}{\emph{13th ACM SIGPLAN
  {I}nternational {C}onference on {F}unctional {P}rogramming (ICFP'08)}},
  \bibfield{editor}{\bibinfo{person}{James Hook} {and} \bibinfo{person}{Peter
  Thiemann}} (Eds.). \bibinfo{publisher}{ACM}, \bibinfo{pages}{143--156}.
\newblock


\bibitem[\protect\citeauthoryear{Crary}{Crary}{2010}]%
        {Crary10icfp}
\bibfield{author}{\bibinfo{person}{Karl Crary}.}
  \bibinfo{year}{2010}\natexlab{}.
\newblock \showarticletitle{Higher-order Representation of Substructural
  Logics}. In \bibinfo{booktitle}{\emph{Proceedings of the 15th International
  Conference on Functional Programming (ICFP 2010)}},
  \bibfield{editor}{\bibinfo{person}{P.Hudak} {and}
  \bibinfo{person}{S.Weirich}} (Eds.). \bibinfo{publisher}{ACM},
  \bibinfo{address}{Baltimore, Maryland}, \bibinfo{pages}{131--142}.
\newblock


\bibitem[\protect\citeauthoryear{Despeyroux}{Despeyroux}{2000}]%
        {Despeyroux00tcs}
\bibfield{author}{\bibinfo{person}{Jo{\"e}lle Despeyroux}.}
  \bibinfo{year}{2000}\natexlab{}.
\newblock \showarticletitle{A Higher-Order Specification of the
  $\pi$-Calculus}. In \bibinfo{booktitle}{\emph{Theoretical Computer Science:
  Exploring New Frontiers of Theoretical Informatics}},
  \bibfield{editor}{\bibinfo{person}{Jan van Leeuwen}, \bibinfo{person}{Osamu
  Watanabe}, \bibinfo{person}{Masami Hagiya}, \bibinfo{person}{Peter~D.
  Mosses}, {and} \bibinfo{person}{Takayasu Ito}} (Eds.).
  \bibinfo{publisher}{Springer Berlin Heidelberg}, \bibinfo{address}{Berlin,
  Heidelberg}, \bibinfo{pages}{425--439}.
\newblock
\showISBNx{978-3-540-44929-4}


\bibitem[\protect\citeauthoryear{Despeyroux, Felty, and Hirschowitz}{Despeyroux
  et~al\mbox{.}}{1995}]%
        {Despeyroux:TLCA95}
\bibfield{author}{\bibinfo{person}{Jo{\"{e}}lle Despeyroux},
  \bibinfo{person}{Amy~P. Felty}, {and} \bibinfo{person}{Andr{\'{e}}
  Hirschowitz}.} \bibinfo{year}{1995}\natexlab{}.
\newblock \showarticletitle{Higher-Order Abstract Syntax in Coq}. In
  \bibinfo{booktitle}{\emph{2nd International Conference on Typed Lambda
  Calculi and Applications ({TLCA} '95)}} \emph{(\bibinfo{series}{Lecture Notes
  in Computer Science (LNCS 902)})},
  \bibfield{editor}{\bibinfo{person}{Mariangiola Dezani{-}Ciancaglini} {and}
  \bibinfo{person}{Gordon~D. Plotkin}} (Eds.). \bibinfo{publisher}{Springer},
  \bibinfo{pages}{124--138}.
\newblock
\urldef\tempurl%
\url{https://doi.org/10.1007/BFb0014049}
\showDOI{\tempurl}


\bibitem[\protect\citeauthoryear{Felty}{Felty}{2019}]%
        {Felty19}
\bibfield{author}{\bibinfo{person}{Amy~P. Felty}.}
  \bibinfo{year}{2019}\natexlab{}.
\newblock \showarticletitle{A Linear Logical Framework in Hybrid (Invited
  Talk)}. In \bibinfo{booktitle}{\emph{{FSCD}}}
  \emph{(\bibinfo{series}{LIPIcs})}, Vol.~\bibinfo{volume}{131}.
  \bibinfo{publisher}{Schloss Dagstuhl - Leibniz-Zentrum f{\"{u}}r Informatik},
  \bibinfo{pages}{2:1--2:2}.
\newblock


\bibitem[\protect\citeauthoryear{Felty, Olarte, and Xavier}{Felty
  et~al\mbox{.}}{2021}]%
        {Felty:MSCS21}
\bibfield{author}{\bibinfo{person}{Amy~P. Felty}, \bibinfo{person}{Carlos
  Olarte}, {and} \bibinfo{person}{Bruno Xavier}.}
  \bibinfo{year}{2021}\natexlab{}.
\newblock \showarticletitle{A focused linear logical framework and its
  application to metatheory of object logics}.
\newblock \bibinfo{journal}{\emph{Math. Struct. Comput. Sci.}}
  \bibinfo{volume}{31}, \bibinfo{number}{3} (\bibinfo{year}{2021}),
  \bibinfo{pages}{312--340}.
\newblock


\bibitem[\protect\citeauthoryear{Gay}{Gay}{2001}]%
        {Gay01tphols}
\bibfield{author}{\bibinfo{person}{Simon~J. Gay}.}
  \bibinfo{year}{2001}\natexlab{}.
\newblock \showarticletitle{A Framework for the Formalisation of Pi Calculus
  Type Systems in Isabelle/HOL}. In \bibinfo{booktitle}{\emph{International
  Conference on Theorem Proving in Higher Order Logics}}.
\newblock


\bibitem[\protect\citeauthoryear{Gay and Hole}{Gay and Hole}{2005}]%
        {Gay05acta}
\bibfield{author}{\bibinfo{person}{Simon~J. Gay} {and} \bibinfo{person}{Malcolm
  Hole}.} \bibinfo{year}{2005}\natexlab{}.
\newblock \showarticletitle{Subtyping for Session Types in the
  {$\pi$}-Calculus}.
\newblock \bibinfo{journal}{\emph{Acta Informatica}} \bibinfo{volume}{42},
  \bibinfo{number}{2--3} (\bibinfo{year}{2005}), \bibinfo{pages}{191--225}.
\newblock


\bibitem[\protect\citeauthoryear{Gay and Vasconcelos}{Gay and
  Vasconcelos}{2010}]%
        {Gay10jfp}
\bibfield{author}{\bibinfo{person}{Simon~J. Gay} {and}
  \bibinfo{person}{Vasco~T. Vasconcelos}.} \bibinfo{year}{2010}\natexlab{}.
\newblock \showarticletitle{Linear Type Theory for Asynchronous Session Types}.
\newblock \bibinfo{journal}{\emph{Journal of Functional Programming}}
  \bibinfo{volume}{20}, \bibinfo{number}{1} (\bibinfo{date}{Jan.}
  \bibinfo{year}{2010}), \bibinfo{pages}{19--50}.
\newblock


\bibitem[\protect\citeauthoryear{Georges, Murawska, Otis, and Pientka}{Georges
  et~al\mbox{.}}{2017}]%
        {Georges:ESOP17}
\bibfield{author}{\bibinfo{person}{Aina~Linn Georges}, \bibinfo{person}{Agata
  Murawska}, \bibinfo{person}{Shawn Otis}, {and} \bibinfo{person}{Brigitte
  Pientka}.} \bibinfo{year}{2017}\natexlab{}.
\newblock \showarticletitle{{LINCX:} {A} Linear Logical Framework with
  First-Class Contexts}. In \bibinfo{booktitle}{\emph{26th European Symposium
  on Programming ({ESOP} 2017)}} \emph{(\bibinfo{series}{Lecture Notes in
  Computer Science (LNCS 20201)})}, \bibfield{editor}{\bibinfo{person}{Hongseok
  Yang}} (Ed.). \bibinfo{pages}{530--555}.
\newblock
\urldef\tempurl%
\url{https://doi.org/10.1007/978-3-662-54434-1_20}
\showDOI{\tempurl}


\bibitem[\protect\citeauthoryear{Harper, Licata, Lovas, Martens, and
  Simmons}{Harper et~al\mbox{.}}{2009}]%
        {twelf-popl-tutorial}
\bibfield{author}{\bibinfo{person}{Robert Harper}, \bibinfo{person}{Dan
  Licata}, \bibinfo{person}{William Lovas}, \bibinfo{person}{Chris Martens},
  {and} \bibinfo{person}{Robert Simmons}.} \bibinfo{year}{2009}\natexlab{}.
\newblock \bibinfo{title}{POPL Tutorial: Mechanizing Metatheory with LF and
  Twelf}.
\newblock
\newblock
\urldef\tempurl%
\url{http://twelf.org/wiki/POPL_Tutorial/Saturday}
\showURL{%
\tempurl}


\bibitem[\protect\citeauthoryear{Honda}{Honda}{1993}]%
        {Honda93concur}
\bibfield{author}{\bibinfo{person}{Kohei Honda}.}
  \bibinfo{year}{1993}\natexlab{}.
\newblock \showarticletitle{Types for Dyadic Interaction}. In
  \bibinfo{booktitle}{\emph{4th International Conference on Concurrency Theory
  (CONCUR 1993)}}, \bibfield{editor}{\bibinfo{person}{E.~Best}} (Ed.).
  \bibinfo{publisher}{Springer LNCS 715}, \bibinfo{pages}{509--523}.
\newblock


\bibitem[\protect\citeauthoryear{Honda, Vasconcelos, and Kubo}{Honda
  et~al\mbox{.}}{1998}]%
        {Honda98esop}
\bibfield{author}{\bibinfo{person}{Kohei Honda}, \bibinfo{person}{Vasco~T.
  Vasconcelos}, {and} \bibinfo{person}{Makoto Kubo}.}
  \bibinfo{year}{1998}\natexlab{}.
\newblock \showarticletitle{Language Primitives and Type Discipline for
  Structured Communication-Based Programming}. In \bibinfo{booktitle}{\emph{7th
  European Symposium on Programming Languages and Systems (ESOP 1998)}},
  \bibfield{editor}{\bibinfo{person}{C.~Hankin}} (Ed.).
  \bibinfo{publisher}{Springer LNCS 1381}, \bibinfo{pages}{122--138}.
\newblock


\bibitem[\protect\citeauthoryear{Jacobs, Balzer, and Krebbers}{Jacobs
  et~al\mbox{.}}{2022}]%
        {Jacobs22popl}
\bibfield{author}{\bibinfo{person}{Jules Jacobs}, \bibinfo{person}{Stephanie
  Balzer}, {and} \bibinfo{person}{Robbert Krebbers}.}
  \bibinfo{year}{2022}\natexlab{}.
\newblock \showarticletitle{Connectivity Graphs: A Method for Proving Deadlock
  Freedom Based on Separation Logic}.
\newblock \bibinfo{journal}{\emph{Proc. ACM Program. Lang.}}
  \bibinfo{volume}{6}, \bibinfo{number}{POPL}, Article \bibinfo{articleno}{1}
  (\bibinfo{date}{Jan.} \bibinfo{year}{2022}), \bibinfo{numpages}{33}~pages.
\newblock
\urldef\tempurl%
\url{https://doi.org/10.1145/3498662}
\showDOI{\tempurl}


\bibitem[\protect\citeauthoryear{Milner}{Milner}{1980}]%
        {Milner80}
\bibfield{author}{\bibinfo{person}{Robin Milner}.}
  \bibinfo{year}{1980}\natexlab{}.
\newblock \bibinfo{booktitle}{\emph{A Calculus of Communicating Systems}}.
\newblock \bibinfo{publisher}{Springer-Verlag LNCS 92}.
\newblock


\bibitem[\protect\citeauthoryear{Nanevski, Pfenning, and Pientka}{Nanevski
  et~al\mbox{.}}{2008}]%
        {Nanevski08tocl}
\bibfield{author}{\bibinfo{person}{Aleksandar Nanevski}, \bibinfo{person}{Frank
  Pfenning}, {and} \bibinfo{person}{Brigitte Pientka}.}
  \bibinfo{year}{2008}\natexlab{}.
\newblock \showarticletitle{Contextual Modal Type Theory}.
\newblock \bibinfo{journal}{\emph{Transactions on Computational Logic}}
  \bibinfo{volume}{9}, \bibinfo{number}{3} (\bibinfo{year}{2008}).
\newblock


\bibitem[\protect\citeauthoryear{Pfenning and Elliott}{Pfenning and
  Elliott}{1988}]%
        {Pfenning88pldi}
\bibfield{author}{\bibinfo{person}{Frank Pfenning} {and} \bibinfo{person}{Conal
  Elliott}.} \bibinfo{year}{1988}\natexlab{}.
\newblock \showarticletitle{Higher-Order Abstract Syntax}. In
  \bibinfo{booktitle}{\emph{Proceedings of the {ACM SIGPLAN '88} Symposium on
  Language Design and Implementation}}. \bibinfo{address}{Atlanta, Georgia},
  \bibinfo{pages}{199--208}.
\newblock


\bibitem[\protect\citeauthoryear{Pientka}{Pientka}{2008}]%
        {Pientka:POPL08}
\bibfield{author}{\bibinfo{person}{Brigitte Pientka}.}
  \bibinfo{year}{2008}\natexlab{}.
\newblock \showarticletitle{A type-theoretic foundation for programming with
  higher-order abstract syntax and first-class substitutions}. In
  \bibinfo{booktitle}{\emph{{35th {ACM} SIGPLAN-SIGACT Symposium on Principles
  of Programming Languages (POPL'08)}}}. \bibinfo{pages}{371--382}.
\newblock


\bibitem[\protect\citeauthoryear{Pientka and Dunfield}{Pientka and
  Dunfield}{2008}]%
        {Pientka:PPDP08}
\bibfield{author}{\bibinfo{person}{Brigitte Pientka} {and}
  \bibinfo{person}{Jana Dunfield}.} \bibinfo{year}{2008}\natexlab{}.
\newblock \showarticletitle{Programming with proofs and explicit contexts}. In
  \bibinfo{booktitle}{\emph{{ACM {SIGPLAN} Symposium on Principles and Practice
  of Declarative Programming (PPDP'08)}}}. \bibinfo{pages}{163--173}.
\newblock


\bibitem[\protect\citeauthoryear{Pientka and Dunfield}{Pientka and
  Dunfield}{2010}]%
        {Pientka10ijcar}
\bibfield{author}{\bibinfo{person}{Brigitte Pientka} {and}
  \bibinfo{person}{Jana Dunfield}.} \bibinfo{year}{2010}\natexlab{}.
\newblock \showarticletitle{Beluga: A Framework for Programming and Reasoning
  with Deductive Systems (System Description)}, Vol.~\bibinfo{volume}{6173}.
  \bibinfo{pages}{15--21}.
\newblock
\showISBNx{978-3-642-14202-4}
\urldef\tempurl%
\url{https://doi.org/10.1007/978-3-642-14203-1_2}
\showDOI{\tempurl}


\bibitem[\protect\citeauthoryear{Rocha and Caires}{Rocha and Caires}{2021}]%
        {Rocha21icfp}
\bibfield{author}{\bibinfo{person}{Pedro Rocha} {and} \bibinfo{person}{Lu\'{i}s
  Caires}.} \bibinfo{year}{2021}\natexlab{}.
\newblock \showarticletitle{Propositions-as-Types and Shared State}.
\newblock \bibinfo{journal}{\emph{Proc. ACM Program. Lang.}}
  \bibinfo{volume}{5}, \bibinfo{number}{ICFP}, Article \bibinfo{articleno}{79}
  (\bibinfo{date}{Aug.} \bibinfo{year}{2021}), \bibinfo{numpages}{30}~pages.
\newblock
\urldef\tempurl%
\url{https://doi.org/10.1145/3473584}
\showDOI{\tempurl}


\bibitem[\protect\citeauthoryear{R{\"o}ckl, Hirschkoff, and
  Berghofer}{R{\"o}ckl et~al\mbox{.}}{2001}]%
        {Rockl01fossacs}
\bibfield{author}{\bibinfo{person}{Christine R{\"o}ckl},
  \bibinfo{person}{Daniel Hirschkoff}, {and} \bibinfo{person}{Stefan
  Berghofer}.} \bibinfo{year}{2001}\natexlab{}.
\newblock \showarticletitle{Higher-Order Abstract Syntax with Induction in
  {Isabelle/HOL}: Formalizing the Pi-Calculus and Mechanizing the Theory of
  Contexts}. In \bibinfo{booktitle}{\emph{Proceedings of the 4th International
  Conference on Foundations of Software Science and Computation Structures
  (FOSSACS'01)}}, \bibfield{editor}{\bibinfo{person}{F.~Honsell} {and}
  \bibinfo{person}{M.~Miculan}} (Eds.). \bibinfo{publisher}{Springer Verlag
  LNCS 2030}, \bibinfo{address}{Genova, Italy}, \bibinfo{pages}{364--378}.
\newblock


\bibitem[\protect\citeauthoryear{Sano, Kavanagh, and Pientka}{Sano
  et~al\mbox{.}}{2023}]%
        {Sano23oopsla-art}
\bibfield{author}{\bibinfo{person}{Chuta Sano}, \bibinfo{person}{Ryan
  Kavanagh}, {and} \bibinfo{person}{Brigitte Pientka}.}
  \bibinfo{year}{2023}\natexlab{}.
\newblock \bibinfo{booktitle}{\emph{Mechanizing Session-Types Using a
  Structural View}}.
\newblock
\urldef\tempurl%
\url{https://doi.org/10.5281/zenodo.8329645}
\showDOI{\tempurl}


\bibitem[\protect\citeauthoryear{Schack{-}Nielsen and
  Sch{\"{u}}rmann}{Schack{-}Nielsen and Sch{\"{u}}rmann}{2008}]%
        {Schack-Nielsen:IJCAR08}
\bibfield{author}{\bibinfo{person}{Anders Schack{-}Nielsen} {and}
  \bibinfo{person}{Carsten Sch{\"{u}}rmann}.} \bibinfo{year}{2008}\natexlab{}.
\newblock \showarticletitle{Celf - {A} Logical Framework for Deductive and
  Concurrent Systems (System Description)}. In
  \bibinfo{booktitle}{\emph{{IJCAR}}} \emph{(\bibinfo{series}{Lecture Notes in
  Computer Science})}, Vol.~\bibinfo{volume}{5195}.
  \bibinfo{publisher}{Springer}, \bibinfo{pages}{320--326}.
\newblock


\bibitem[\protect\citeauthoryear{Thiemann}{Thiemann}{2019}]%
        {Thiemann19ppdp}
\bibfield{author}{\bibinfo{person}{Peter Thiemann}.}
  \bibinfo{year}{2019}\natexlab{}.
\newblock \showarticletitle{Intrinsically-Typed Mechanized Semantics for
  Session Types}. In \bibinfo{booktitle}{\emph{Proceedings of the 21st
  International Symposium on Principles and Practice of Declarative
  Programming}} \emph{(\bibinfo{series}{PPDP '19})}.
  \bibinfo{publisher}{Association for Computing Machinery},
  \bibinfo{address}{New York, NY, USA}, Article \bibinfo{articleno}{19},
  \bibinfo{numpages}{15}~pages.
\newblock
\showISBNx{9781450372497}
\urldef\tempurl%
\url{https://doi.org/10.1145/3354166.3354184}
\showDOI{\tempurl}


\bibitem[\protect\citeauthoryear{Tiu and Miller}{Tiu and Miller}{2010}]%
        {Tiu:TOCL10}
\bibfield{author}{\bibinfo{person}{Alwen Tiu} {and} \bibinfo{person}{Dale
  Miller}.} \bibinfo{year}{2010}\natexlab{}.
\newblock \showarticletitle{Proof search specifications of bisimulation and
  modal logics for the pi-calculus}.
\newblock \bibinfo{journal}{\emph{{ACM} Trans. Comput. Log.}}
  \bibinfo{volume}{11}, \bibinfo{number}{2} (\bibinfo{year}{2010}),
  \bibinfo{pages}{13:1--13:35}.
\newblock


\bibitem[\protect\citeauthoryear{Toninho, Caires, and Pfenning}{Toninho
  et~al\mbox{.}}{2013}]%
        {Toninho13esop}
\bibfield{author}{\bibinfo{person}{Bernardo Toninho},
  \bibinfo{person}{Lu{\'\i}s Caires}, {and} \bibinfo{person}{Frank Pfenning}.}
  \bibinfo{year}{2013}\natexlab{}.
\newblock \showarticletitle{Higher-Order Processes, Functions, and Sessions: A
  Monadic Integration}. In \bibinfo{booktitle}{\emph{Proceedings of the
  European Symposium on Programming (ESOP'13)}},
  \bibfield{editor}{\bibinfo{person}{M.~Felleisen} {and}
  \bibinfo{person}{P.~Gardner}} (Eds.). \bibinfo{publisher}{Springer LNCS
  7792}, \bibinfo{address}{Rome, Italy}, \bibinfo{pages}{350--369}.
\newblock


\bibitem[\protect\citeauthoryear{Wadler}{Wadler}{2012}]%
        {Wadler12icfp}
\bibfield{author}{\bibinfo{person}{Philip Wadler}.}
  \bibinfo{year}{2012}\natexlab{}.
\newblock \showarticletitle{Propositions as Sessions}. In
  \bibinfo{booktitle}{\emph{Proceedings of the 17th International Conference on
  Functional Programming (ICFP 2012)}}. \bibinfo{publisher}{ACM Press},
  \bibinfo{address}{Copenhagen, Denmark}, \bibinfo{pages}{273--286}.
\newblock


\bibitem[\protect\citeauthoryear{Zalakain}{Zalakain}{2019}]%
        {ZalakainMS19}
\bibfield{author}{\bibinfo{person}{Uma Zalakain}.}
  \bibinfo{year}{2019}\natexlab{}.
\newblock \emph{\bibinfo{title}{Type-checking session-typed \(\pi\)-calculus
  with Coq}}.
\newblock Masters Thesis. \bibinfo{school}{University of Glasgow}.
\newblock
\urldef\tempurl%
\url{https://www.dcs.gla.ac.uk/~ornela/projects/Uma\%20Zalakain.pdf}
\showURL{%
\tempurl}


\bibitem[\protect\citeauthoryear{Zalakain and Dardha}{Zalakain and
  Dardha}{2021}]%
        {Zalakain21forte}
\bibfield{author}{\bibinfo{person}{Uma Zalakain} {and} \bibinfo{person}{Ornela
  Dardha}.} \bibinfo{year}{2021}\natexlab{}.
\newblock \showarticletitle{$\pi$ with Leftovers: A Mechanisation in Agda}. In
  \bibinfo{booktitle}{\emph{Formal Techniques for Distributed Objects,
  Components, and Systems}}, \bibfield{editor}{\bibinfo{person}{Kirstin Peters}
  {and} \bibinfo{person}{Tim A.~C. Willemse}} (Eds.).
  \bibinfo{publisher}{Springer International Publishing},
  \bibinfo{address}{Cham}, \bibinfo{pages}{157--174}.
\newblock
\showISBNx{978-3-030-78089-0}


\end{thebibliography}
